\documentclass[english, 12pt]{article}
\RequirePackage[OT1]{fontenc}
\RequirePackage{graphicx}
\RequirePackage{amsthm}
\RequirePackage{amsmath}
\RequirePackage{yhmath}
\RequirePackage[authoryear]{natbib}
\RequirePackage[colorlinks,citecolor=blue,urlcolor=blue]{hyperref}
\usepackage{algorithm,algorithmic}
\usepackage{amssymb}
\usepackage{footnote}
\usepackage{subfig}
\usepackage{array}
\usepackage{setspace}
\usepackage{comment}
\usepackage[margin=1.25in]{geometry}

\theoremstyle{plain}
\newtheorem{thm}{Theorem}
\newtheorem{assumption}{Assumption}
\newtheorem{cor}{Corollary}
\numberwithin{equation}{section}
\numberwithin{figure}{section}

\theoremstyle{remark}

\usepackage{xr-hyper} 

\makeatletter

\newcommand{\indep}{\perp \!\!\! \perp}

\usepackage{mathtools}
\DeclarePairedDelimiter\abs{\lvert}{\rvert}
\DeclarePairedDelimiter\norm{\lVert}{\rVert}

\@ifundefined{showcaptionsetup}{}{
	\PassOptionsToPackage{caption=false}{subfig}}
\makeatother

\begin{document}
\title{Extrapolation in Regression Discontinuity Design Using Comonotonicity}
	\author{Ben Deaner (UCL) and Soonwoo Kwon (Brown)}
\maketitle
\begin{abstract}
  We present a novel approach for extrapolating causal effects away from the
  margin between treatment and non-treatment in sharp regression discontinuity
  designs with multiple covariates. Our methods apply both to settings in which
  treatment is a function of multiple observables and settings in which
  treatment is determined based on a single running variable. Our key
  identifying assumption is that conditional average treated and untreated
  potential outcomes are comonotonic: covariate values associated with higher
  average untreated potential outcomes are also associated with higher average
  treated potential outcomes. We provide an estimation method based on local
  linear regression. Our estimands are weighted average causal effects, even if
  comonotonicity fails. We apply our methods to evaluate counterfactual
  mandatory summer school policies.
\end{abstract}

Regression discontinuity design (RDD) provides a means for identifying causal
effects without the need to isolate a source of exogenous variation in
treatment. Instead, causal effects are identified under an a priori continuity
assumption which is plausible in many empirical settings. Unfortunately, while
causal inference in RDD design may be credible, its scope is limited. Standard
RDD analysis can identify average causal effects only within the sub-population
whose covariates place them on the margin between treatment and
non-treatment. As such, identification of the causal impact of counterfactual
policies is limited to those under which there is at most a marginal change in
the treatment rule.

To be more precise, consider the sharp RDD framework which is the focus of the
present work. In sharp RDD, the binary treatment $D$ is deterministic given some
observable individual characteristics $X$. For example, treatment may be
assigned if and only if a particular observable covariate known as the `running
variable' exceeds a certain threshold. In other cases treatment may be assigned
by a more complicated rule that depends on multiple observables. Sharp RDD
settings arise when receipt of an intervention is decided by some central
authority purely on the basis of individual characteristics that are observed by
both the authority and researcher. For example, eligibility for a particular
welfare program may be decided according to a strict rule based on an
individual's age, income, household size, and marital status. Let $Y$ be an
outcome of interest and $Y(d)$ the potential outcome from treatment level
$d$. In RDD, treatment is entirely explained by the covariates, that is, the
value of $D$ depends entirely on whether $X$ lies in a `treatment region' or a
`no-treatment region'. Under a continuity assumption, the conditional average
treatment effect (CATE), given by $E[Y(1)-Y(0)|X=x]$, is identified for values
of $x$ that lie on the `frontier' between the treatment and no-treatment
regions. Without further assumptions, the CATE is not identified away from the
frontier.

In this work we provide a novel means for identifying causal effects away from
the frontier in sharp RDD settings in which $X$ is multivariate: that is, more
than one covariate is available. We show that our approach is valid when a
particular comonotonicity (in the sense of e.g., \cite{Schmeidler1989})
condition holds. In its simplest form, our comonotonicity condition states that
the mean \textbf{untreated} potential outcome is weakly greater conditional on a
covariate value $x$ that conditional on a value $x'$, if and only if the
conditional average \textbf{treated} potential outcome is also greater for value
$x$ than $x'$. For example, if wealthier and older individuals have higher
potential outcomes on average when they are not treated, they must also have
higher average potential outcomes when treated. This assumption has testable
implications. Under the standard RDD assumptions, both treated and untreated
conditional average potential outcomes are identified along the frontier, and so
one can assess whether the comonotonicity condition holds for covariate values
on the frontier.

The comonotonicity assumption suggests a straightforward means of extrapolating
conditional average treatment effects away from the frontier. Given a vector of
covariate values $x$ that lie in say, the interior of the treatment region, we
estimate the conditional average treated potential outcome for these values
$E[Y(1)|X=x]$, which is equal to the conditional mean of the factual outcome
$E[Y|X=x]$. We then find a point $x^*$ on the frontier with the same associated
conditional mean treated potential outcome. The conditional means of both
treated and untreated potential outcomes are identified on the frontier under
standard RDD assumptions. If such a point exists, then the conditional average
treatment effect at $x$ is identical to that at $x^*$.

We develop a novel empirical approach based on our identification results. In
particular, we develop a non-parametric method based on local linear
regression. Notably, the methods do not require applied researchers to specify
the frontier between the treatment and non-treatment region and thus our
techniques are applicable even when treatment status is observed but the
assignment rule is not a priori known. This facilitates implementation of our
methods, an applied researcher need only input the outcomes, covariates, and
treatments into a provided software package <>, with no need to encode the
formula for the treatment rule.

Our estimator targets a particular weighted average of causal effects for
individuals along the frontier. Thus our estimand is a weighted average of
causal effects, regardless of whether comonotonicity holds. Comonotonicity then
ensures that this weighted average is equal to the CATE at a chosen covariate
value away from the frontier.

We apply our methods to analyse the impact of mandatory summer school
attendance. This empirical setting was previously analysed by
\citet{Matsudaira2008}. We use our methods to examine heterogeneity in treatment
effects.

\subsection*{Related Literature}

This paper contributes to the extensive literature on RDD (see
\citealp{Cattaneo2022} for a recent review). A number of works use multivariate
RDD to test a kind of selection-on-observables assumption. These include
\cite{Battistin2008} and \cite{Mealli2012}.  A closely related approach is that
of \cite{Angrist2015}. \cite{Angrist2015} extrapolate treatment effects away
from the cut-off by assuming that, conditional on some covariates, a scalar
running variable is unrelated to potential outcomes. Extrapolation from the
cut-off using higher order derivatives is considered in \cite{DiNardo2011}. A
full statistical analysis using first derivatives and extension to fuzzy RDD is
given in \cite{Dong2015} who identify the impact of a marginal change in the
threshold. \cite{Cattaneo2021} uses the existence of different cut-offs for
certain sub-populations to extrapolate causal effects in univariate RDD under a
type of parallel trends assumption. Another strand of the literature identifies
causal effects in RDD away from the frontier by using auxiliary information, for
example, \cite{Grembi2016} and \cite{Wing2013}.

A number of other works consider RDD settings with one or more covariates in
addition to a running variable. \cite{Calonico2019} and \cite{Noack2024} develop
inference methods that can be applied to RDD setting with covariates, where the
primary role of covariates is efficiency gain. In our work, we exploit the
availability of additional covariates in order to identify richer causal objects
rather than for statistical precision. Our paper also relates to the literature
that considers multivariate running variables. For example, papers such as
\cite{Papay2011} in the canonical RD design or \cite{Keele2015} in the
geographic RDD setting. The methods by \cite{armstrong2018optimal} and
\cite{Imbens2019} that put a bound on the smoothness and work out the optimal
bias-variance trade-off naturally apply to the multivariate setting as well.

Comonotonicity is closely related to, but distinct from, rank invariance and
rank similarity of potential outcomes. An extensive literature employs rank
invariance or similarity assumptions in order to achieve causal inference using
instrumental variables. For example, \cite{Chernozhukov2005},
\cite{Chernozhukov2007}, \cite{Horowitz2007}, \cite{Imbens2009},
\cite{Torgovitsky2015}, and \cite{DHaultfoeuille2015}. In contrast to the rank
invariance/similarity assumptions employed in these works, comonotonicity
restricts conditional mean treated and untreated potential outcomes rather than
the joint distribution of treated and untreated potential outcomes.

We provide further comparison with \cite{Angrist2015} and IV methods using
rank-invariance in Appendix B.4.

\section{A Simple Demonstration}

Before we present formal results, let us begin with a worked example. For
concreteness, we frame this in a simplified version of the empirical setting in
Section 5, however we have not explicitly calibrated the example to that data.

We suppose the vector of covariates $X$ is two-dimensional, that is
$X=(X_{1},X_{2})'$ where $X_{1}$ and $X_{2}$ are both scalars. In line with the
application in Section 5, let us suppose $X_1$ and $X_2$ are, respectively,
students scores on a math and reading test. We assume that an individual
receives treatment, which we indicate by $D=1$, if and only if a linear
combination of the individual's test scores falls below a threshold. In our
empirical application, treatment is mandatory attendance of summer school. Thus
treatment is deterministic given the covariates: the individual receives
treatment if the vector of test scores $X$ falls in some region
$\mathcal{X}_{1}$. Conversely, the individual does not receive treatment ($D=0$)
if $X$ is in the region $\mathcal{X}_{0}$.  Together $\mathcal{X}_{1}$ and
$\mathcal{X}_{0}$ make up the support of $X$.

In the Figure 1 we have normalized $X_{1}$ and $X_{2}$ so that their joint
support is the unit square. An individual is treated when $X$ lies in the
turquoise region labelled $\mathcal{X}_{1}$ and is not treated if $X$ falls into
the magenta region marked $\mathcal{X}_{0}$.  The intersection of the boundaries
of these two regions is marked with a blue line labelled $\mathcal{F}$ which we
refer to as the `frontier'.  To be precise, the figure indicates that an
individual is treated if $0.4 X_{1}+X_{2}\leq0.7$.
\begin{figure}[h]
	\caption{Identification in two-dimensional RDD}
	\centering
	\includegraphics[scale=0.3]{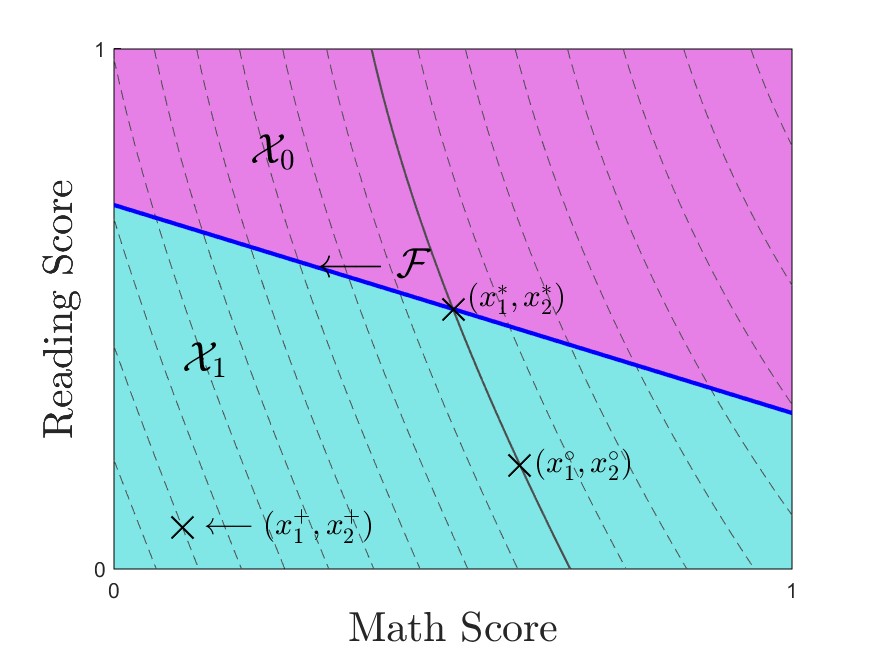}
\end{figure}

Consider an outcome, $Y$ that is an individual's normalized score on a math test
taken a year after the initial tests in $X$ and thus after any treatment has
been administered. The dashed curves in the figure are contour curves (level
sets) of the conditional average outcome $E[Y|X=x]$ within the treated or
untreated regions. That is, if two points $x_1,x_2\in\mathcal{X}_{1}$ lie on the
same contour, then $E[Y|X=x_1]=E[Y|X=x_2]$ and likewise if
$x_1,x_2\in\mathcal{X}_{0}$, but not if say, $x_1\in\mathcal{X}_{1}$ and
$x_2\in\mathcal{X}_{0}$.

In Figure 1, $(x_{1}^{*},x_{2}^{*})$ lies on the frontier $\mathcal{F}$.  As
such, both the conditional average treated and untreated potential outcomes
$E[Y(0)|X=x^{*}]$ and $E[Y(1)|X=x^{*}]$, are identified at this point under
standard RDD assumptions.  The point $(x_{1}^{\circ},x_{2}^{\circ})$ is in the
interior of the treatment region $\mathcal{X}_{1}$ and as such,
$E[Y(1)|X=x^{\circ}]$ is identified, indeed this quantity simply equals
$E[Y|X=x^{\circ}]$. However, without further assumptions we cannot identify
$E[Y(0)|X=x^{\circ}]$ nor the conditional average treatment effect
$E[\tau|X=x^{\circ}]$.

To progress, we make the following `comonotonicity' assumption. Consider two
vectors of initial math and reading test scores $x=(x_1,x_2)$ and
$x'=(x_1',x_2')$. We suppose $E[Y(0)|X=x]$ is greater than $E[Y(0)|X=x']$ if and
only if $E[Y(1)|X=x]$ is greater than $E[Y(1)|X=x']$. That is, the functions
$E[Y(0)|X=\cdot]$ and $E[Y(1)|X=\cdot]$ are `comonotonic' in the sense of e.g.,
\cite{Schmeidler1989}. In the present context, this assumption states that
initial test scores associated with higher average math outcomes under the
counterfactual of no-treatment, are also associated with higher average math
scores under the counterfactual of treatment. We provide further discussion of
the assumption in the context of our empirical application in Section 5. In
addition, Appendix B.1 specifies a model of skill accumulation motivated by the
test score application, and we show that the assumption holds under this model.

Note that $x^{\circ}$
lies on the same contour as $x^{*}$ which is indicated by the solid
black curve. For all values of $x\in\mathcal{X}_{1}$ along this
contour, $E[Y|X=x]$ and thus $E[Y(1)|X=x]$, is constant. It follows from the
 comonotonicity condition that $E[Y(0)|X=x^{*}]$ is identical to $E[Y(0)|X=x^{\circ}]$.
Therefore, the conditional average treatment effect is the same at
$x^{\circ}$ as at $x^{*}$, and so $E[\tau|X=x^{\circ}]$ is identified.

In short, we extrapolate conditional average treatment effects away
from the frontier and into the treatment region by matching values
of $x^{\circ}$ with a value $x^{*}$ on the frontier with the same
conditional average treated potential outcome. Similarly, we can extrapolate
into the interior of the no-treatment region by finding covariate
values on the frontier with the same conditional average untreated
potential outcomes.

Note that this approach does not always allow
us to extrapolate to all points in the support of $X$. For example,
for some $x\in\mathcal{X}_{1}$ there may be no $x'\in\mathcal{F}$
such that $E[Y(1)|X=x']=E[Y(1)|X=x]$. Indeed, this is true for the
point $(x_{1}^{+},x_{2}^{+})$ in Figure 1. Consider the extreme case in which the contours of $E[Y|X=x]$ are parallel to the frontier. Then we cannot extrapolate anywhere. Conversely, if contours  are perpendicular to the frontier and $X$ has sufficiently large support, then we can extrapolate everywhere.
Which case is closer to the truth depends on the context.
Under comonotonicity the contours are contours of treatment effects. If a social planner wishes to treat individuals with high causal effects, it is optimal to set the frontier equal to a contour. In practise, the social planner may be constrained to employ simple treatment rules, may not be able to accurately assess conditional average causal effects, or may care about more than one outcome or outcomes that differ from those of interest to the researcher. Such cases may be intermediate to these two extreme examples allowing us to extrapolate causal effects to some but not all points in the support of $X$. Indeed, this appears to be the case in our empirical application.

Recall that along the frontier, both the treated and untreated conditional
average potential outcomes can be identified under standard RDD assumptions.  On
the left panel in Figure 1.2, we plot conditional average potential outcomes at
each point along the frontier as a function of the $x_1$-coordinate of that
point. The conditional average treated potential outcome is in turquoise and the
conditional average untreated potential outcome in magenta. The vertical dashed
blue line is at the location $x^*_1$ and intersects each of the curves at the
corresponding dashed horizontal line. The elevation of the dashed horizontal
turquoise line is equal to $E[Y(1)|X=x^*]$ which is also equal to
$E[Y(1)|X=x^\circ]$. The elevation of the dashed magenta line is
$E[Y(0)|X=x^*]$.

On the right panel in Figure 1.2, we plot the pairs $(E[Y(1)|X=x],E[Y(0)|X=x])$
for all values of $x$ along the frontier $\mathcal{F}$.  This curve is shown in
blue. The curve is strictly increasing which suggests that along the frontier
values of $x$ associated with a greater conditional average treated potential
outcomes also have greater conditional average untreated potential outcomes. We
can also see this on the left panel, an $x_1$ with a larger associated value of
$E[Y(0)|X=x]$ also has a larger value of $E[Y(1)|X=x]$.  Note that the
comonotonicity condition implies that the blue curve in the right panel must be
strictly increasing, and this can be used as the basis for a test of this
assumption.

\begin{figure}[h]
\caption{Conditional Average Potential Outcomes}
\centering
\subfloat[]{
\includegraphics[scale=0.19]{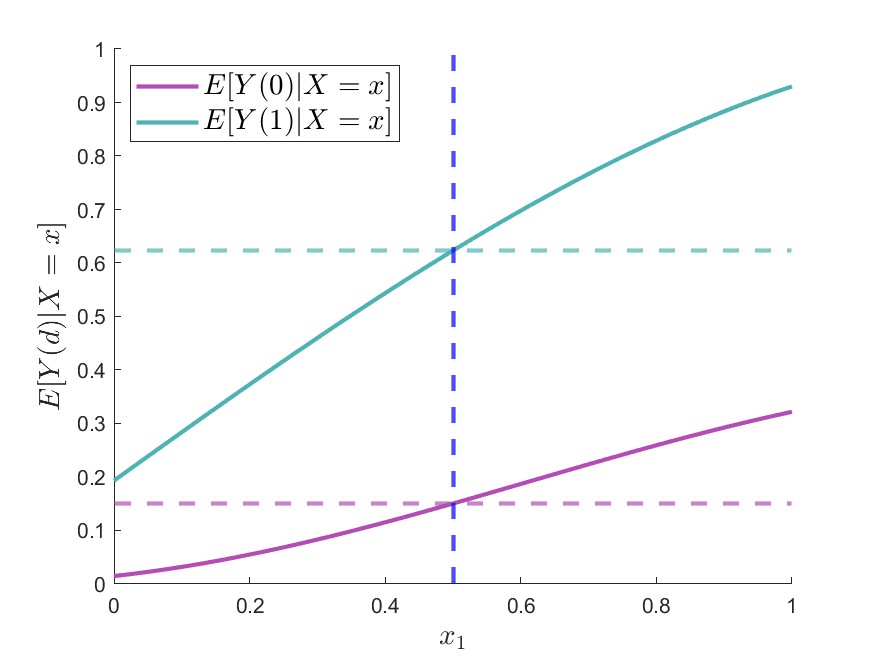}
}
\subfloat[]{
\includegraphics[scale=0.19]{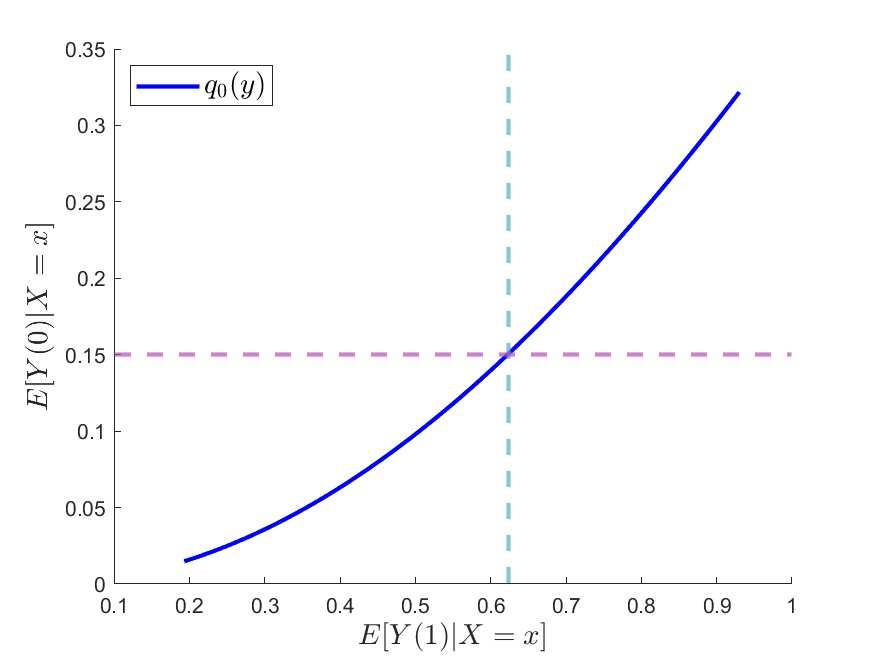}
}
\end{figure}

The vertical dashed line (shown in turquois) in the right panel of Figure 1.2 is
located at the value $E[Y(1)|X=x^{\circ}]$. As we discuss above, this quantity
is trivially identified. This is the same value as $E[Y(1)|X=x^{*}]$, which is
located along the frontier.  The corresponding value $E[Y(0)|X=x^{*}]$ is thus
equal to the elevation of the horizontal dashed line, shown in magenta. Under
the comonotonicity condition, this is equal to the value of
$E[Y(1)|X=x^{\circ}]$.

More generally, for a given point $x$ in the treated region, we can find the
value of $E[Y(0)|X=x]$ from the right panel by the $y$-coordinate of the point
on the blue curve with $x$-coordinate equal to $E[Y|X=x]$. Conversely, if $x$ is
in the no-treatment region, then $E[Y(1)|X=x]$ is the $x$-coordinate of the
point on the blue curve whose $y$-coordinate is $E[Y|X=x]$. For some values of
$x$ there may be no corresponding point on the blue curve in which case we
cannot extrapolate conditional average treatment effects to that value of $x$.

We can understand the blue curve in Figure 1.2 as a function that returns the
value of $E[Y(0)|X=x]$ given the value of $E[Y(1)|X=x]$. We denote this function
by $q_0$ and conversely, the function that returns $E[Y(1)|X=x]$ in given
$E[Y(0)|X=x]$ by $q_1$.

\section{Formal Identification Results}

We now state our formal results, we begin by re-stating the standard sharp RDD
identification result that conditional average treatment effects are identified
for covariate values on the frontier between the treatment and no-treatment
regions. Again we let $Y$ be a scalar outcome of interest, $D$ a binary
treatment indicator. We denote by $Y(d)$ as the potential outcome under
treatment level $d$, $\tau=Y(1)-Y(0)$ as the individual level treatment effect,
and $X$ a vector of covariates with support $\mathcal{X} \subset
\mathbb{R}^{d_{X}}$. The support $\mathcal{X}$ is partitioned into a treatment
region $\mathcal{X}_{1}$ and no-treatment region $\mathcal{X}_{0}$.

The frontier consists of those covariate values that lie on the margin between
the treated and untreated regions. To define this formally, let
$int(\mathcal{X}_1)$ denote the interior of the treated region and
$int(\mathcal{X}_0)$ the interior of the untreated region. Then the frontier
$\mathcal{F}$ is the intersection of the boundaries between these two
sets. Equivalently, letting $cl(\cdot)$ return the closure of a set,
\[ \mathcal{F}=cl\big(int(\mathcal{X}_{1})\big)\cap
cl\big(int(\mathcal{X}_{0})\big).
\]

If $\mathcal{X}_{1}$ and $\mathcal{X}_{0}$ are sufficiently regular then the
above is simply the intersection of their boundaries.\footnote{In particular, if
the closures of $\mathcal{X}_{1}$ and $\mathcal{X}_{0}$ are `regular closed
sets'.} The use of the interiors in the definition allows us to avoid cluttering
the analysis with regularity conditions on $\mathcal{X}_{1}$ and
$\mathcal{X}_{0}$.

Standard RDD analysis identifies conditional average treatment effects at the
frontier under the conditions below.

\theoremstyle{definition} \newtheorem*{A01}{Assumption 1.1} 
\begin{A01}[Deterministic Treatment]
  $D=1\{X\in\mathcal{X}_{1}\}$, where $\mathcal{X}_{1}$
  and $\mathcal{X}_{0}$ do not overlap and their union is the
  support of $X$.
\end{A01}

\theoremstyle{definition} \newtheorem*{A02}{Assumption 1.2}
\begin{A02}[Continuity]
  The functions $x\mapsto E[Y(0)|X=x]$ and $x\mapsto E[Y(1)|X=x]$ are continuous.
\end{A02}

Assumption 1.1 simply states that treatment is deterministic given $X$ and
defines the treatment and no-treatment regions $\mathcal{X}_{1}$ and
$\mathcal{X}_{0}$. Assumption 1.2 imposes that the conditional average potential
outcomes are continuous. Strictly speaking, in order to identify conditional
average treatment effects at the frontier, it would suffice to replace
Assumption 1.2 with the weaker condition that the conditional expectations are
continuous at each $x$ on the frontier. However, the stronger continuity
condition allows us to sidestep some technicalities.\footnote{In particular,
when $X$ is continuously distributed, conditional expectations like
$E[Y(0)|X=x]$ and $E[Y(1)|X=x]$ are not uniquely defined. From a
measure-theoretic perspective, $x\mapsto E[Y(d)|X=x]$ is a measurable function
$g$ on the support of $X$ so that for any measurable set $\mathcal{S}$ with
$P(X\in\mathcal{S})>0$, $E[g(X)|X\in\mathcal{S}]=E[Y(d)|X\in\mathcal{S}]$, but
such a function $g$ is not unique.  However, if there exists a continuous $g$
with this property, then it is unique. Thus we understand $x\mapsto E[Y(d)|X=x]$
for $d=0,1$ to refer to the unique continuous conditional expectation
functions. We employ this convention throughout the text.}

Theorem 1 below is standard, it establishes identification of the conditional
average potential outcomes, and thus conditional average treatment effect
$E[\tau|X=x]$ for each $x$ on the frontier $\mathcal{F}$.

\theoremstyle{plain}
\newtheorem*{T1}{Theorem 1} 
\begin{T1}
  Suppose Assumptions 1.1 and 1.2 hold. Then $E[Y(1)|X=x]$ and $E[Y(0)|X=x]$
  are identified for every $x\in\mathcal{F}$.

  More precisely, for each $d=0,1$ there is a function $g_d$ on
$\mathcal{X}_d\cup \mathcal{F}$ so that for all $x\in \mathcal{X}_d$,
$g_d(x)=E[Y|X=x]$, and $g_d$ is continuous at each point in $\mathcal{F}$. For
each $x\in\mathcal{F}$, conditional average potential outcomes are identified as
follows:
  \begin{equation}
    E[Y(d)|X=x]=g_d(x) \label{ident}
  \end{equation}
\end{T1}

\subsection{Extrapolation Under Comonotonicity}

Theorem 1 identifies conditional average potential outcomes at the frontier. In
order to identify effects away from the frontier we require an additional
identifying assumption.
	\theoremstyle{definition} \newtheorem*{A2}{Assumption 2}
\begin{A2}[Comonotonicity]
 for any $x_{2},x_{1}\in\mathcal{X}$: 
\begin{equation*}
 E[Y(0)|X=x_{1}]\geq E[Y(0)|X=x_{2}] \iff E[Y(1)|X=x_{1}]\geq E[Y(1)|X=x_{2}]
\end{equation*}
\end{A2}

Assumption 2 states that values of the covariates associated with higher average
untreated potential outcomes also correspond to higher treated potential
outcomes.Before we discuss Assumption 2 in greater depth and provide sufficient
conditions, let us first state a key consequence of Assumption 2 which
constitutes our main identification result.
\theoremstyle{plain}
\newtheorem*{T2}{Theorem 2} 
\begin{T2}
  Suppose Assumptions 1 and 2 hold and define $g_0$ and $g_1$ as in Theorem 1. Suppose that for some $x\in\mathcal{X}_{d}$,
  there exists an $x^{*}\in\mathcal{F}$ so that $E[Y|X=x]=g_{d}(x^{*})$. Then $E[Y(1)|X=x]=g_{1}(x^{*})$
  and $E[Y(0)|X=x]=g_{0}(x^{*})$. Moreover, for any $x_1,x_2\in\mathcal{F}$, $g_0(x_1)\geq g_0(x_2)\iff g_1(x_1)\geq g_1(x_2)$.
\end{T2}

The first statement of Theorem 2 shows how comonotonicity allows us to impute
conditional average potential outcomes away from the frontier for certain values
of $x$. The second statement notes that under our assumptions, $g_1(x)$ is
increasing in the corresponding value of $g_0(x)$ along the frontier and vice
versa. The latter condition suggests that comonotonicity is testable because
$g_0(x)$ and $g_1(x)$ are identified on the frontier.\footnote{An implication
  (in fact, a sharp testable implication) of comonotonicity is
  $$\inf_{x,x' \in \mathcal{F}} (g_{d}(x) - g_{d}(x'))(g_{1-d}(x) - g_{1-d}(x'))
  \geq 0.$$ Hence, one can directly apply the methods of
  \cite{chernozhukov2013intersection} to falsify comonotonicity. However, since
  $g_{d}$ and $g_{1-d}$ can be estimated independently, we suspect that more
  powerful tests tailored to our setting can be developed. We view this as an
  interesting direction for future research.}

\begin{comment}
When we consider estimators motivated by Theorem 2 it is helpful to re-frame the result as follows. Define the function $q_{1-d}$
so that for any $x$ on the frontier, \[q_{1-d}\big(E[Y(d)|X=x]\big)=E[Y(1-d)|X=x].\]
The domain of this function consists of all the values of $E[Y(d)|X=x]$ for some $x$. Under regularity conditions on the frontier this set is an interval with lower and upper end points $\underline{y}_{1-d}$ and and $\bar{y}_{1-d}$ (which can be positive or negative infinity) defined below:
\begin{equation}
	\underline{y}_{1-d}=\inf_{x\in\mathcal{F}}E[Y(d)|X=x],\,\,\,\,\,\bar{y}_{1-d} =\sup_{x\in\mathcal{F}}E[Y(d)|X=x] \label{endpoints}
\end{equation}
Under the conditions of Theorem 2, the function $q_{1-d}$ is well-defined and weakly increasing. Note that the curve in Figure 1.2 is precisely the function $q_0$.

Now,  consider some $x\in\mathcal{X}_d$ not necessarily on the frontier.  Under the conditions of Theorem 2, if $E[Y|X=x]$ is in the domain of $q_{1-d}$, then the conditional mean potential outcomes at $x$ are identified as follows:
\begin{align}
	E[Y(d)|X=x]&=E[Y|X=x]\nonumber\\
	E[Y(1-d)|X=x]&=q_{1-d}\big(E[Y|X=x]\big)\label{identuse}
\end{align}

That is, the conditional average treated and untreated potential outcomes can be written in terms of the conditional mean outcome and the function $q_{1-d}$  which is identified by standard RDD arguments.
\end{comment}

When we consider estimators motivated by Theorem 2 it is helpful to re-frame the
result as follows. Under the conditions of Theorem 2, there is a unique,
increasing function $q_{1-d}$ so that for every $x\in\mathcal{F}$,
\begin{equation}
  E[Y(1-d)|X=x]=q_{1-d}\big(E[Y|X=x]\big).\label{qdef1}
\end{equation}
The domain of $q_{1-d}$ consists of all the values of $E[Y(d)|X=x]$ for some
$x\in\mathcal{F}$. Under regularity conditions on the frontier this set is an
interval with lower and upper end points $\underline{y}_{1-d}$ and and
$\bar{y}_{1-d}$ (which can be positive or negative infinity) with formulae
 \begin{equation}
 \underline{y}_{1-d}=\inf_{x\in\mathcal{F}}E[Y(d)|X=x],\,\,\,\,\,\bar{y}_{1-d} =\sup_{x\in\mathcal{F}}E[Y(d)|X=x]. \label{endpoints}
 \end{equation}
 Under the conditions of Theorem 2, $q_{1-d}$ is weakly increasing. The curve in
Figure 1.2 is precisely the function $q_0$.

Now, consider some $x\in\mathcal{X}_d$ not necessarily on the frontier.  Under
the conditions of Theorem 2, if $E[Y|X=x]$ is in the domain of $q_{1-d}$, then
the conditional mean potential outcomes at $x$ are identified as follows.
\begin{align}
	E[Y(d)|X=x]&=E[Y|X=x]\nonumber\\
	E[Y(1-d)|X=x]&=q_{1-d}\big(E[Y|X=x]\big)\label{identuse}
\end{align}
That is, the conditional average treated and untreated potential outcomes can be
written in terms of the conditional mean outcome and the function $q_{1-d}$,
which are both identified. Conditional average treatment effects may thus be
identified at a point $x\in\mathcal{X}_d$ by
\begin{equation}
	E[\tau|X=x]=(2d-1)\big(E[Y|X=x]-q_{1-d}(E[Y|X=x])\big). \label{tauimpute}
\end{equation}

If comonotonicity does not hold, then there may be no function that satisfies
(\ref{qdef1}) and so it is helpful to use the following definition of $q_{1-d}$,
which coincides with (\ref{qdef1}) under the conditions of Theorem 2.
\begin{equation}q_{1-d}(y):=E\big[Y(1-d)\big|E[Y(d)|X]=y,X\in\mathcal{F}\big].\label{qdef2}
\end{equation}
This function remains well-defined even if comonotoncity fails. Moreover, using
this definition, the right hand side of (\ref{tauimpute}) is a weighted average
of conditional average treatment effects, regardless of whether comonotonicity
holds.

Proposition 1 states two sufficient, but by no means necessary, conditions for
Assumption 2. The first sufficient condition states that if average untreated
potential outcomes are higher within some covariate stratum, then average
treatment effects are also greater in this stratum. The second condition states
that, in a certain sense, there is less heterogeneity in treatment effects
between strata of $X$ than heterogeneity in average baseline potential outcomes.

\theoremstyle{plain}
\newtheorem*{P1}{Proposition 1} 
\begin{P1}
	Either of the following two conditions implies Assumption 2.
	
	a. For all $x_{2},x_{1}\in\mathcal{X}$
	\[
	E[Y(0)|X=x_{1}]\geq E[Y(0)|X=x_{2}] \implies E[\tau|X=x_{1}]\geq E[\tau|X=x_{2}].
	\]

	b. For all $x_{2},x_{1}\in\mathcal{X}$
	\begin{equation}
		|E[\tau|X=x_{1}]-E[\tau|X=x_{2}]|\leq|E[Y(0)|X=x_{2}]-E[Y(0)|X=x_{1}]|,
	\end{equation}
	and the inequality is strict when the quantity on the RHS is non-zero. 
\end{P1}

An immediate corollary of Proposition 1 is that, if average treatment effects do
not differ between strata of $X$, then comonotonicity holds. Of course,
homogeneity in treatment effects makes extrapolation away from the frontier
trivial. Proposition 1 demonstrates two ways in which our assumptions weaken
treatment effect homogeneity. In particular, it suffices for Assumption 2 that
differences in average treatment effects between covariate strata have the same
sign as differences in average untreated potential outcomes, or that the
difference in treatment effects is smaller than the difference in baseline
potential outcomes.

In a randomized controlled trial (RCT), both conditional average treated and
untreated potential outcomes are directly identified. Thus one can formally test
or informally assess whether the assumption of comonotonicity holds in the
setting of a given RCT. In Appendix B.5 we provide an informal empirical
assessment of whether comonotonicity is credible in the context of the RCT in
\cite{Alan}. The empirical setting in \cite{Alan} is apt in that it shares key
features with the setting for our empirical application, in particular both the
outcomes and covariates consist of students' reading and math test scores as in
our application.

In particular settings, comonotonicity may be justified by primitive modelling
conditions. In Appendix B.1 we show that the assumption holds under a model of
skill formation and noisy test scores that is motivated by our empirical
application.

Comonotoncity is related to, but distinct from, rank invariance of potential
outcomes. This rank invariance condition states that if we draw two individuals
$i$ and $j$ iid from the population, so that $(Y_i(0),Y_i(1))$ and
$(Y_j(0),Y_j(1))$ are independent copies of $(Y(0),Y(1))$, then $Y_{i}(0)\geq
Y_{j}(0)\iff Y_{i}(1)\geq Y_{j}(1)$.  Note that rank-invariance has to hold for
every pair of individuals in the population. By contrast, Assumption 2 concerns
only averages of potential outcomes between subgroups. Moreover, rank invariance
of potential outcomes is not verifiable, even with an ideal randomized
controlled trial, because it involves a comparison between an individual's
outcomes in two distinct counterfactual worlds, whereas all of the objects in
Assumption 2 can be identified in an ideal randomized experiment.

In a formal sense, rank-invariance of potential outcomes is neither stronger nor
weaker than Assumption 2 (that is, neither one implies the other). However, one
can show that rank-invariance implies a conditional quantile version of
Assumption 2:
\begin{equation}
Q_{Y(1)|X}(t|x_1)\geq Q_{Y(1)|X}(t|x_2)\iff Q_{Y(0)|X}(t|x_1)\geq Q_{Y(0)|X}(t|x_2) \label{QRV}
\end{equation}
where $Q_{Y(d)|X}(t|x)$ is the $q$-th conditional quantile of $Y(d)$ given
$X=x$. Thus if the conditional distributions of potential outcomes are
symmetric, then rank-invariance implies Assumption 2 (because the conditional
means and medians are identical). Note that if we were to assume rank invariance
in place of Assumption 2, we could use (\ref{QRV}) to identify conditional
quantiles of potential outcomes and in some cases the entire distribution of
individual treatment effects within some covariate strata. We discuss this
further in Appendix B along with formal results.

\subsection{Conditional Comonotoncity}

In cases with strictly more than two covariates, the comonotonicity condition is
generally stronger than necessary for identification.  Suppose $X$ has three or
more components. Let $X^{(1)}$ be the subvector containing precisely two of the
characteristics in $X$ and let $X^{(2)}$ be the subvector that contains the
remaining characteristics. Assumption 3 below weakens Assumption 2 in that it
only requires comonotonicity to hold within each stratum of the characteristics
in $X^{(2)}$ and not between strata.

\theoremstyle{definition} \newtheorem*{A3}{Assumption 3}
\begin{A3}[Conditional Comonotonicity]
	For any $x_{1},x_{2}\in\mathcal{X}$ with $x_{1}^{(2)}=x_{2}^{(2)}$,
	we have: 
	\begin{equation*}
		E[Y(1)|X=x_{1}]\geq E[Y(1)|X=x_{2}]\iff E[Y(0)|X=x_{1}]\geq E[Y(0)|X=x_{2}]
	\end{equation*}	
\end{A3}

Theorem 3 is similar to Theorem 2. Again we extrapolate conditional average
potential outcomes from a point on the frontier $x^*$ to a point $x$ away from
the frontier. However, Theorem 3 strengthens the requirement on $x^*$ so that
its subvector $x^{*(2)}$ is identical to the corresponding subvector $x^{(2)}$
of $x$.

\theoremstyle{plain}
\newtheorem*{T3}{Theorem 3} 
\begin{T3}
  Suppose Assumptions 1 and 3 hold and define $g_0$ and $g_1$ as in Theorem 1. Suppose that for some $x\in\mathcal{X}_{d}$,
  there exists an $x^{*}\in\mathcal{F}$ so that $x^{*(2)}=x^{(2)}$ and  $E[Y|X=x]=g_{d}(x^{*})$. Then $E[Y(1)|X=x]=g_{1}(x^{*})$
  and $E[Y(0)|X=x]=g_{0}(x^{*})$.  Moreover, for any $x_1,x_2\in\mathcal{F}$,with $x_1^{(2)}=x_2^{(2)}$,  $g_0(x_1)\geq g_0(x_2)\iff g_1(x_1)\geq g_1(x_2)$.
\end{T3}

The second statement in Theorem 3 states the weaker condition that Assumption 3
places on $g_0$ and $g_1$ along the frontier. Whereas under full comonotonicity
they must move in the same direction, under the conditional comonotonicity
assumption, this co-movement need only apply within each fixed stratum of
$X^{(2)}$.

In the appendix we provide further alternatives that weaken Assumption 2 and
which may still be sufficient for identification. In particular, we consider
local comonotonicity. In this case, we require that the vector of first
derivatives of the conditional average treated and untreated potential outcomes
are identical. Moreover, we provide an even weaker, conditional version of local
comonotonicity, akin to the conditioning in Assumption 3.

\section{Estimation}
\label{sec:estimation}
Our analysis motivates estimates of conditional average potential outcomes and
treatment effects. From these, one may obtain estimates of the causal effect of
an alternative treatment allocation policy. For example, the average effect of
treating all individuals. We first consider estimation under full comonotonicity
based on Theorem 2.  We then consider estimation under the conditional
comonotonicity assumption based on the identification result in Theorem 3.

\subsection{Conditional Average Potential Outcome Estimates}

We obtain conditional average potential outcome estimates by the plug-in
principle. That is, we replace the functions $E[Y|X=\cdot]$, $q_{0}$, and
$q_{1}$ with estimates in (\ref{identuse}). Let $\hat{q}_{d}$ be an estimate of
$q_{d}$ for $d=0,1$. We estimate $E[Y|X=x]$ separately for $x$ in the treated
and untreated regions, using only the data for individuals within the
corresponding region. Let $\hat{g}_d(x)$ denote the estimate of $E[Y|X=x]$ for
$x\in\mathcal{X}_d$. Consider an $x$ in the treated region $\mathcal{X}_1$. We
estimate the conditional average potential outcomes as follows.
\begin{align}
	\widehat{E[Y(1)|X=x]}=\hat{g}_1(x),\,\,\,\,\,
	\widehat{E[Y(0)|X=x]}=\hat{q}_{0}\big(\hat{g}_1(x)\big)
	\label{reguse2}
\end{align}
Taking the difference between these two estimates yields an estimate of the
conditional average treatment effect at $x$:
\[
\hat{\tau}(x)=\hat{g}_1(x)-\hat{q}_{0}\big(\hat{g}_1(x)\big)
\]
Conversely, for $x$ in the untreated region $\mathcal{X}_0$, we simply switch
$0$ with $1$ in the formula above and estimate conditional average causal
effects accordingly.
\begin{align*}
	\widehat{E[Y(0)|X=x]}&=\hat{g}_0(x),\,\,\,\,\,
	\widehat{E[Y(1)|X=x]}=\hat{q}_{1}\big(\hat{g}_0(x)\big)
	\label{reguse}\\
	\hat{\tau}(x)&=\hat{q}_{1}\big(\hat{g}_0(x)\big)-\hat{g}_0(x)
\end{align*}

We obtain $\hat{g}_1$ by non-parametric regression of $Y$ on $X$, using only
data on individuals in the treated region. Similarly, for $x$ in the untreated
region, we obtain $\hat{g}_0(x)$ by regression using only data in the untreated
region. We suggest local linear regression for this purpose but one could use
alternative non-parametric regression methods. The local linear regression
estimates take the form below for $d=0,1$, where $K$ is a kernel (e.g., uniform,
Gaussian, or triangular) and $h$ is a bandwidth:
\begin{align*}
	\hat{g}_{d}(x)	=e_{1}'\hat{\beta}_{x},\,\,\,\,\,
	\hat{\beta}_{x}	=\arg\min_{\beta}\sum_{i \in \mathcal{I}_{d}}K\bigg(\frac{\|X_{i}-x\|}{h}\bigg)\big(Y_{i}-(1,(X_{i}-x)')\beta\big)^{2},
\end{align*}
where we define $\mathcal{I}_{d} := \{i:X_{i}\in\mathcal{X}_{d}\}$ set of
indices for the observations with treatment status $d$.

Estimation of $q_{0}$ and $q_1$ is a little more involved. For ease of
exposition, we focus below on estimation of conditional average untreated
potential outcomes in the treated region and thus on the estimation of
$q_0$. Estimation of conditional average treated potential outcomes in the
untreated region proceeds symmetrically.

Recall the definition of $q_{1-d}$ given in equation (\ref{qdef2}). We see
$q_0(y)$ is the mean untreated potential outcome among individuals on the
frontier whose conditional mean treated potential outcome is $y$. This
characterization motivates a nonparametric regression-based
approach.\footnote{Note that under comonotonicity, $q_{1-d}$ is also uniquely
  defined by (\ref{qdef1}).   However, if we replace $E[Y(1)|X=x]$ and
  $E[Y(0)|X=x]$ with noisy regression estimates $\hat{g}_1(x)$ and
  $\hat{g}_0(x)$, there is no guarantee that there is a functional relationship
  between them on the frontier. That is, there may not exist any function
  $\tilde{q}_0$ so that $\tilde{q}_0\big(\hat{g}_1(x)\big)=\hat{g}_0(x)$ for all
  $x\in\mathcal{F}$. This motivates the use of the characterization in
  (\ref{qdef2}) which is remains well-defined even if we replace conditional
  average potential outcomes on the frontier with noisy estimates, and even if
  comonotoncity fails.} In brief, we take the untreated individuals who lie
within a small shrinking neighbourhood of the frontier and we regress their
outcomes on estimates of $E[Y(1)|X]$. To assess whether an untreated individual
$i$ is within a small neighbourhood of the frontier, we check whether $i$'s
covariates $X_i$ are sufficiently close to those of $i$'s nearest treated
neighbour. Let $NN(x)$ be the index of this nearest neighbour for $x \in
\mathcal{X}_{0}$ (in the event of ties, one can choose among them at random), formally
\[
NN(x):=\arg\min_{j\in\mathcal{I}_1}\|X_j-x\|,
\]
and $NN_{i} := NN(X_{i})$. If $\|X_i-X_{NN_{i}}\|\leq d$, and treated and
untreated regions satisfy some regularity conditions, it follows that $X_i$ is
within distance $d$ of a point on the frontier. This approach allows researchers
to apply our method in settings in which treatment statuses are observed but the
location of the frontier itself is unknown. It also obviates the need for
applied practitioners to calculate distances of points to the frontier which
could complicate practical implementation.

For untreated individuals within a shrinking distance from the frontier, the
conditional mean treated potential outcome can be consistently estimated. A
crude estimate is $\hat{g}_1(X_{NN_{i}})$, which is an estimate of the
conditional mean outcome of the nearest treated neighbour. However, if the
conditional mean potential outcome is twice differentiable, one can reduce bias
by estimating the derivative of the conditional mean outcome at
$X_{NN_{i}}$. Let $\triangledown\hat{g}_1(x)$ be an estimate of the vector of
partial derivatives $\frac{\partial}{\partial x} E[Y|X=x]$, then for
$x \in \mathcal{X}_{0}$, we estimate $E[Y(1)|X=x]$ by
\[\tilde{g}_1(x)=\hat{g}_1(X_{NN(x)})+\triangledown\hat{g}_1(X_{NN(x)})'(X_i-X_{NN(x)}).\]
In the case of local linear regression, the local linear regression coefficients
$\hat{\beta}_{X_{NN(x)}}$ (absent the intercept) are a natural choice for
$\triangledown\hat{g}_1(X_{NN(x)})$. The resulting estimator can be written
succinctly as follows:
\begin{align*}\tilde{g}_1(x)&=(1,x')\hat{\beta}_{X_{NN(x)}}\\
\hat{\beta}_{X_{NN(x)}}
                            &=\arg\min_{\beta}\sum_{j\in\mathcal{I}_{1}}K\bigg(\frac{\|X_{j}-X_{NN(x)}\|}{h}\bigg)\big(Y_{j}-(1,X_{j}')\beta\big)^{2}
\end{align*}
As described above, we estimate $q_{0}(y)$ by nonparametric regression of $Y_i$
on $\tilde{g}_1(X_i)$ using only treated individuals $i$ for whom
$\|X_i-X_{NN_{i}}\|$ is sufficiently small. Again we suggest local linear
regression. Let $b$ be a bandwidth and
$W_i=1\{\|X_i-X_{NN_{i}}\|\leq \varepsilon\}$.
The estimate is given below.
\begin{align}\label{eq:def-q}
	\hat{q}_{0}(y)&=(1,0)\hat{\gamma}_{y}\nonumber\\
	\hat{\gamma}_{y}	&=\arg\min_{\gamma}\sum_{i \in \mathcal{I}_{0}}W_{i}K\bigg(\frac{|\tilde{g}_{1}(X_{i})-y|}{b}\bigg)\big[Y_{i}-\big(1,\tilde{g}_{1}(X_{i})-y\big)\gamma\big]^{2}
\end{align}
The bandwidth $b$ may be chosen by cross-validation. One could replace $Y_i$ in
the above with $\hat{g}_0(X_i)$, which may increase precision but at the expense
of greater bias. The weighting by the binary indicator $W_i$ ensures that the
local linear regression only includes data on individuals sufficiently close to
the frontier. As a rule of thumb, we suggest setting $d=\omega h$, where
$\omega$ is twice the $0.75$ quantile of the kernel $K$ and $h$ is the bandwidth
used in the estimate $\hat{g}_1$. This choice is designed to ensure that the
bias from using points not exactly at the frontier is of a similar magnitude to
the bias in the estimation of $\hat{g}_1$.

Recall that under regularity conditions, the domain of $q_0$ is an interval with
end points $\underline{y}_0$ and $\bar{y}_0$ defined in (\ref{endpoints}). We
estimate these end points by $\hat{\underline{y}}_0$ and $\hat{\bar{y}}_0$
defined below.
\[
\hat{\underline{y}}_0=\min_{i:i\in\mathcal{I}_1,W_i=1}\tilde{g}_{1}(X_{i}),\,\,\,\,\,\hat{\bar{y}}_0=\max_{i:
  i \in \mathcal{I}_{1},W_i=1}\tilde{g}_{1}(X_{i})
\]

Figure 3.1 provides a visual demonstration of the method in which the
non-parametric regressions are carried out with a uniform kernel. Figure 3.1.a
contains simulated data from the expository model in Section 1. The crosses
represent the covariate values of different data points.  Let $i$ be the index
of the untreated individual with covariate values indicated by the large cross
in Figure 3.1.a. The covariate vales of $i$'s nearest treated neighbour
$X_{NN_{i}}$, are indicated by the large plus sign. The large black circle is
centered at $X_{NN_{i}}$ and has radius $h$ (which was here chosen arbitrarily
to equal $0.2$). Those data points in the treated region within distance $h$ of
$X_{NN_{i}}$ are indicated by the small circles.  The estimate
$\tilde{g}_0(X_i)$ is the fitted value at $X_i$ from linear regression of the
outcomes on the covariates using only those data points whose covariates are
indicated by small circles. The distance between $X_{i}$ and $X_{NN_{i}}$ is the
length of the black dashed line between these two points. Because we use a
uniform kernel $\omega=1$, and thus the suggested value of $d$ described above
is simply $h$. Therefore, $W_i=1$ because $\|X_{i}-X_{NN_{i}}\|\leq h$ as can be
seen from the fact $X_i$ lies within the black circle.

\begin{figure}[h]
	\caption{Estimation Detail}
	\centering
	\subfloat[]{
		\includegraphics[scale=0.22]{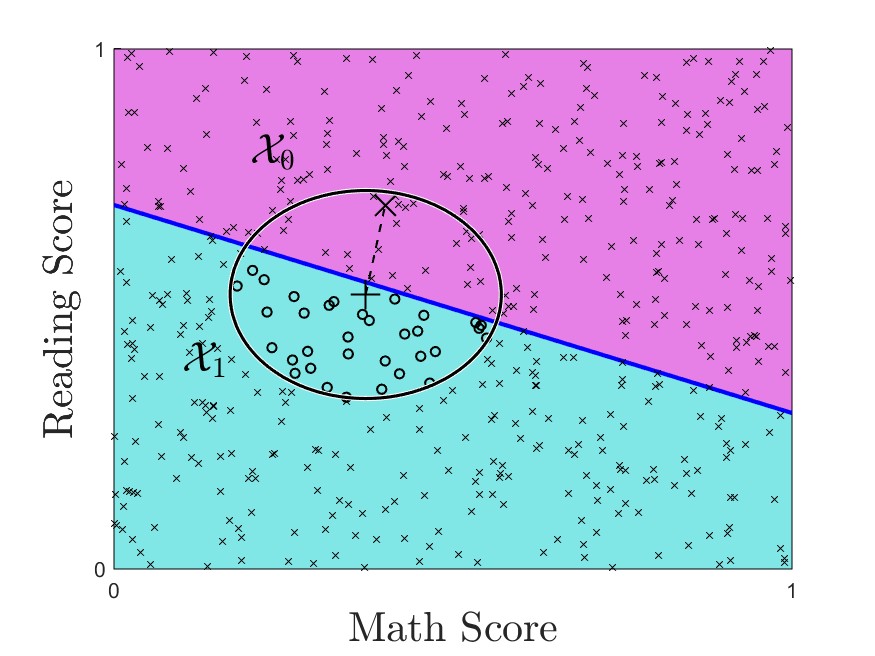}
	}
	\subfloat[]{
		\includegraphics[scale=0.22]{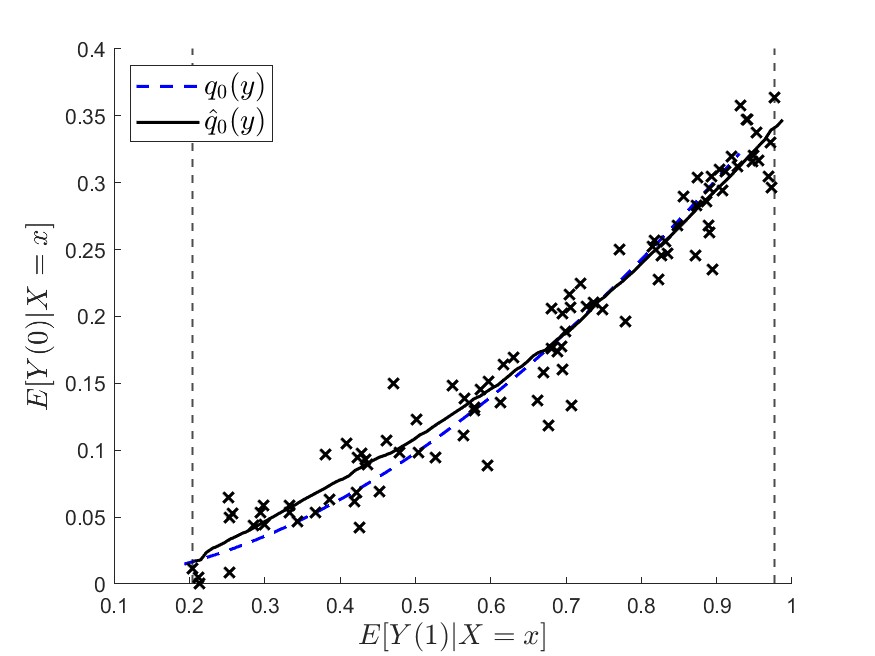}
	}
\end{figure}

Using the same simulated data, Figure 3.1.b plots values of $Y_i$ for untreated
individuals with $W_i=1$ against corresponding values of $\tilde{g}_1(X_i)$
which is evaluated as in the previous paragraph. Regressing these outcomes on
the corresponding values of $\tilde{g}_1(X_i)$ by local linear regression with a
uniform kernel yields the solid black curve which is the estimate of $q_0$ (here
we simply use the same ad hoc bandwidth of $0.2$). The true value of $q_0$ is
indicated by the dashed blue curve. The values of $\hat{\underline{y}}_0$ and
$\hat{\bar{y}}_0$ are indicated by the vertical dashed lines.

\subsection{Effects of Counterfactual Policies}

Given estimates of conditional average potential outcomes, one can estimate the
causal effects of counterfactual treatment regimes. Consider a policy in which
an individual with $X=x$ is treated with probability $p(x)$. Note that if
$p(x)\in\{0,1\}$ for all $x$, then treatment remains deterministic under this
regime. The mean causal impact of this counterfactual policy on the outcomes of
individuals with values of $X$ in a set $\mathcal{S}$ denoted $\theta_0$, is
given below.
\begin{align}
	\theta_0=&E\big[\big(1-D\big)p(X)\big(E[Y(1)|X]-Y)\big|X\in\mathcal{S}\big]\nonumber\\
	+&E\big[D\big(1-p(X)\big)\big(E[Y(0)|X]-Y\big)\big| X\in\mathcal{S}\big]
	\label{counter}
\end{align}

The restriction $X\in\mathcal{S}$ allows us to restrict attention to those
values of $x$ at which both conditional average potential outcomes are
identified. Recall that under regularity conditions on the frontier, for $d=0,1$
we identify $E[Y(0)|X=x]$ and $E[Y(1)|X=x]$ for each $x$ in $\mathcal{X}_{d}$
such that that $\underline{y}_{1-d}<E[Y|X=x]<\bar{y}_{1-d}$. As such, we focus
on estimation for $\mathcal{S}$ the set that contains these values of $x$. In
order to estimate (\ref{counter}) given this value of $\mathcal{S}$ we use the
plug-in principle. The resulting estimate $\hat{\theta}$ is given by
\begin{align}
\hat{\theta}=&\frac{1}{\sum_{i=1}^{n}S_{i}}{\sum_{i=1}^{n}S_{i}(1-D_{i})p(X_{i})\big(\hat{q}_{1}\big(\hat{g}_{0}(X_{i})\big)-Y_{i})}\nonumber\\
+&\frac{1}{\sum_{i=1}^{n}S_{i}}{\sum_{i=1}^{n}S_{i}D_{i}\big(1-p(X_{i})\big)\big(\hat{q}_{0}\big(\hat{g}_{1}(X_{i})\big)-Y_{i}\big)},
\label{estimatecounter}
\end{align}
where $S_i$ is a binary indicator defined below.
\[S_i=(1-D_i)1\{\hat{\underline{y}}_{1}\leq\hat{g}_{0}(X_{i})\leq\hat{\bar{y}}_{1}\}+D_i 1\{\hat{\underline{y}}_{0}\leq\hat{g}_{1}(X_{i})\leq\hat{\bar{y}}_{0}\}\]

Figure 3.2 plots the same simulated data as in Figure 3.1. In Figure 3.2.a, the
covariate values of individuals with $S_i=1$ are indicated by crosses, whereas
those for whom $S_i=0$ (i.e., those for whom conditional average causal effects
are not identified) are indicated by circles. The dashed blue line in the figure
indicates a counterfactual treatment rule, individuals are treated if and only
if their covariate values are to the right of this dashed line (That is, if
their math score exceeds $0.5$). Note that all individuals whose treatment under
this rule differs from their factual treatment (those with covariate values in
the orange and deep blue triangles) have $S_i=1$, and thus the average treatment
effect of this change is fully identified.

Using the simulated data, Figure 3.2.b plots estimated mean causal effects of a
counterfactual treatment rule in which individuals are treated if and only if
their math score is greater than some cut-off. The x-axis gives the value of
this cut-off and the y-value of the solid blue curve is the corresponding causal
effect estimate calculated from the formula (\ref{estimatecounter}).

\begin{figure}[h]
	\caption{Counterfactual Estimation}
	\centering
	\subfloat[]{
		\includegraphics[scale=0.23]{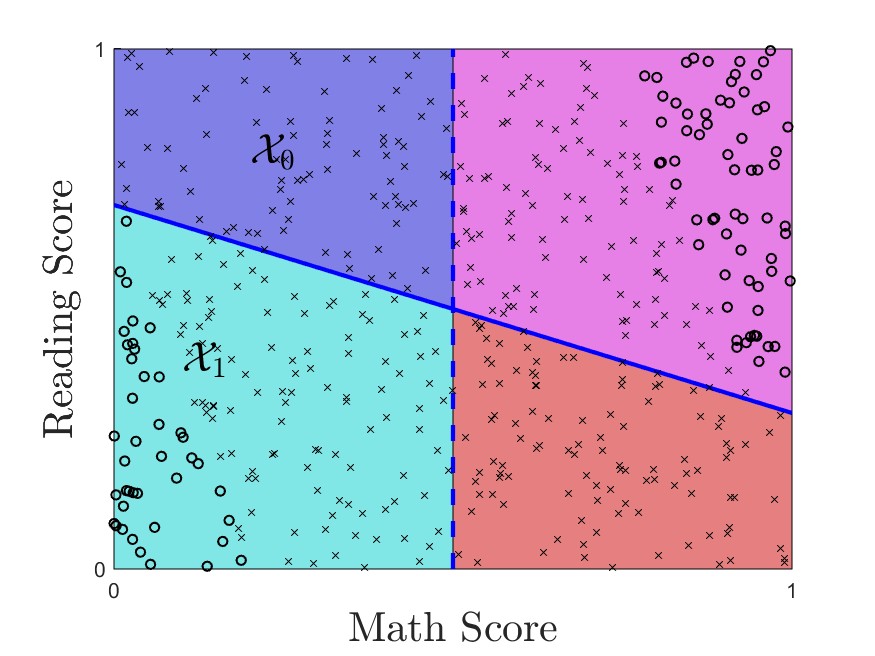}
	}
	\subfloat[]{
		\includegraphics[scale=0.23]{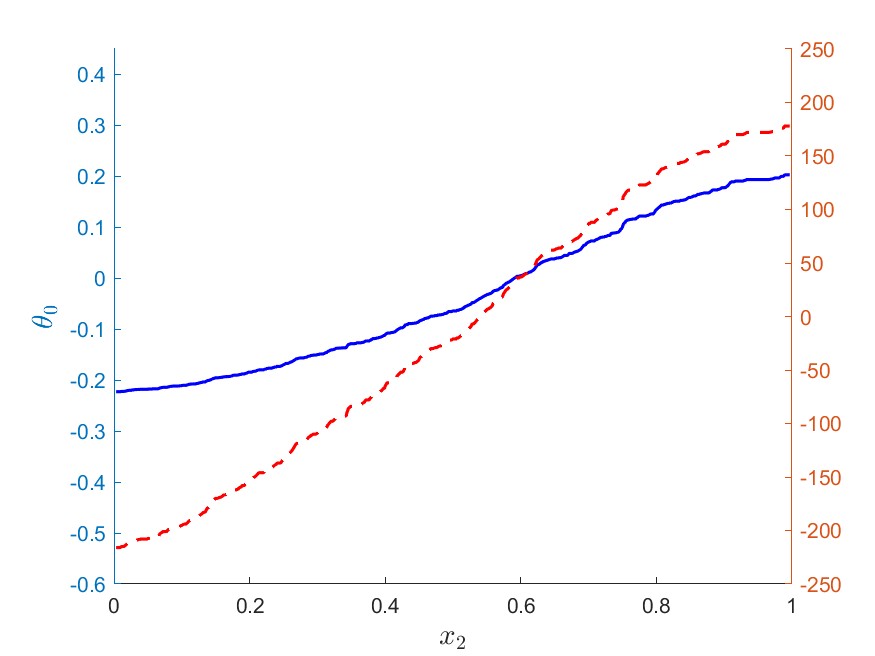}
	}
\end{figure}

It may be helpful to compare $\hat{\theta}$ to the expected change in treatment
status under the counterfactual policy among those individuals in the sample
with $S_i=1$. That is, to compare it with
$\frac{1}{\sum_{i=1}^{n}S_{i}}
\sum_{i=1}^{n}S_{i}[(1-D_{i})p(X_{i})-D_{i}(1-p(X_{i}))]$. If the costs of a
policy are proportional to the number of individuals treated, then this quantity
is proportional to the expected cost of the counterfactual treatment policy on
the subsample with $S_i=1$ less the cost of the factual treatment
assignments. This quantity is plotted by the dashed red curve in Figure 3.2.b
with units given by the y-axis on the right of the figure.

\subsection{Conditional Comonotonicity}

In settings with more than two covariates in $X_i$ we suggest researchers
estimate effects using the weaker conditional comonotonicity condition in
Section 2.2. In these settings $X_i$ can be decomposed into two subvectors
$X_i^{(1)}$ and $X_i^{(2)}$. If $X_i^{(2)}$ is discrete (with a finite number of
support points), then one can apply the estimation method described above
separately within the different strata of $X_i^{(2)}$. More generally, if
$X_i^{(2)}$ is continuously distributed, one needs only slightly adjust the
method in the previous section.

In particular, $\tilde{g}_1$ is estimated exactly as in the previous
subsection. However, in place of $q_{0}(y)$ we estimate a value of this quantity
within each stratum of $X_i^{(2)}$. That is, for each $x^{(2)}$ in the support
of $X_i^{(2)}$ we estimate a stratum-specific value of $q_{0}(y)$ by
$\hat{q}_{0}(y,x^{(2)}):=(1,y,{x^{(2)}}')\hat{\gamma}_{y,x^{(2)}}$, where
$\hat{\gamma}_{y,x^{(2)}}$ is given by
\begin{align*}
  \arg\min_{\gamma}\sum_{i:X_{i}\in\mathcal{X}_{1}}W_{i}K\bigg(\frac{|\tilde{g}_{1}(X_{i})-y|+\|X_i^{(2)}-x^{(2)}\|}{b}\bigg)\big[Y_{i}-\big(1,\tilde{g}_{1}(X_{i}),{X_i^{(2)}}'\big)\gamma\big]^{2}.
\end{align*}
That is, $\hat{q}_{0}(y,x^{(2)})$ is a fitted value from regressing outcomes
$Y_i$ for treated individuals near the frontier on both an estimate of
$E[Y(1)|X=x]$ and also $X^{(2)}$. The corresponding estimated conditional
average treated and untreated potential outcome estimates at
$x=({x^{(1)}}',{x^{(2)}}')'$ are then
\begin{align*}
  \widehat{E[Y(1)|X=x]}=\hat{g}_1(x),\,\,\,\,\,
  \widehat{E[Y(0)|X=x]}=\hat{q}_{0}\big(\hat{g}_1(x),x^{(2)}\big).
\end{align*}

\section{Asymptotic Analysis}

We now provide asymptotic properties of the estimation procedure introduced in
Section \ref{sec:estimation}. We refer to the estimation of $g_{d}$ as the first
stage and estimation of $q_{1-d}$ as the second stage. We first show a (strong)
uniform convergence result for the first stage, which is rather standard in the
literature of local polynomial estimation (e.g., \citealp{stone1982optimal,
  masry1996multivariate, hansen2008uniform}). A minor deviation is that we
consider a local extrapolation to estimate $g_{d}$ for observations are
``slightly outside'' of $\mathcal{X}_{d}$.

Then, we establish asymptotic normality of $\hat{q}_{1-d}(y)$. As evident from
its definition given in \eqref{eq:def-q}, the estimator $\hat{q}_{1-d}(y)$ is
obtained by a local linear regression on regressors that are also generated by a
local linear regression. Hence, the estimation method falls within the general
framework of \citet[hereafter MRS]{mammen2012nonparametric}, whose arguments we
closely follow, with two key differences. First, we incorporate weights $W_i$ to
ensure that only observations near the boundary are used in estimating $g_d$
(for units in $\mathcal{X}_{1-d}$). These weights do not affect the convergence
rate of the first stage but do enter the variance of the second stage,
introducing a trade-off absent in MRS. Second, the first stage estimation error
is independent of the data used for the second stage since the estimation of
$\tilde{g}_{d}(x)$ only uses observations in $\mathcal{X}_{d}$, whereas
$\hat{q}_{1-d}(y)$ only uses observations in $\mathcal{X}_{1-d}$.

We start with showing strong uniform consistency of $\tilde{g}_{d}(x)$ over
$$\mathcal{X}_{d, \varepsilon} :=  \mathcal{X}_{1-d} \cap \{x: \min_{X_{i} \in
  \mathcal{X}_{d}}\lVert x - X_{i} \rVert \leq \varepsilon \},$$ which is the
set of $x \in \mathcal{X}_{1-d}$ such that there exists at least one observation
in $ \mathcal{X}_{d}$ that lies within an $\varepsilon$-ball centered at $x$. By
the triangle inequality and mean-value theorem, we have
\begin{align*}
  &\abs{\tilde{g}_{d}(x) - g_{d}(x)} \\ 
\leq & \abs{(x - NN(x))'\hat{\beta}_{\mathcal{N}_d(x)} - (x - NN(x))'
      \nabla g_{d}(x_{t})} + \abs{\hat{g}_{d}(NN(x)) -
         g_{d}(NN(x))},
\end{align*}
where $x_{t} = tx + (1-t)NN(x)$ for some $t \in [0, 1]$. The convergence of the
second term in the last line is standard. By Cauchy-Schwarz and another
application of the mean-value theorem, the first term can be bounded uniformly
over $\mathcal{X}_{d,\varepsilon}$ by
$$\norm{\hat{\beta}_{NN(x)} - \nabla g_{d}(NN(x))}\varepsilon +
M\varepsilon^{2},$$ which shows that the convergence rate is essentially driven
by the convergence rate by $\norm{\hat{\beta}_{NN(x)} - \nabla g_{d}(NN(x))},$
which is again standard in the local polynomial literature.

We assume the following to establish the said uniform convergence result.
\setcounter{assumption}{3}
\begin{assumption}[Support and Treatment Region]
  \label{assu:supp} The support $\mathcal{X}$ of $X$ is compact. Assume that for
  $d \in \{0,1 \}$, $\pi_{d} :=P(X \in \mathcal{X}_{d}) \in [\delta_{\pi}, 1-\delta_{\pi}]$
  for some $\delta_{\pi} > 0.$
\end{assumption}

The compactness of $\mathcal{X}$ guarantees that the closures of both
$\mathcal{X}_{1}$ and $\mathcal{X}_{0}$ are compact. Such compactness conditions
are used to derive uniform convergence results for $\tilde{g}_{d}(\cdot)$. These
can be relaxed by imposing tail conditions on $f_{X}(\cdot)$ as in
\cite{hansen2008uniform}.

We make the following smoothness and boundedness conditions on the kernel,
density of $X$, and conditional expectation functions $g_{d}(\cdot)$, which are
typical in the nonparametric literature.

\begin{assumption}[Smoothness and boundedness]
  \label{assu:smoooth-bd}(i) The kernel $K$ is a symmetric density function of
  order 2, Lipschitz continuous and has bounded support, (ii)
  $E[Y^{2+\delta_{y}} | X] < \infty$ for some $\delta_{y} > 0$, (iii) the
  density $f_{X}(\cdot)$ is twice differentiable and satisfies
  $0 < \inf_{x \in \mathcal{X}} f_{X}(x) $ and (iv) $g_{d}$ is twice
  differentiable with Lipschitz continuous second derivatives.
\end{assumption}

Define
$a_{n} = O\left( \left(\frac{\log n}{n h^{k}} \right)^{1/2} + h^{2}\right)$,
which is the usual uniform convergence rate of local linear estimators. The
following shows that the convergence rate of the local linear estimator is
unaffected under our local extrapolation. The result is intuitive given the
boundary properties of the local linear estimator and that points that lie
within $h$ distance from the support are essentially the same with being on the
boundary.

\setcounter{thm}{3}
\begin{thm}[Strong Uniform Consistency of $\tilde{g}_{d}(\cdot)$] Suppose
  Assumptions \ref{assu:supp} and \ref{assu:smoooth-bd} hold, and set
  $\varepsilon \asymp h$. Then,
  \begin{equation}
    \label{eq:5}
    \sup_{x \in \mathcal{X}_{d,\varepsilon}} \abs{\tilde{g}_{d}(x) - g_{d}(x)} = O(a_{n}), \text{ a.s. }
\end{equation}
\end{thm}

\begin{proof}
  The result follows immediately by the decomposition and bounds provided above,
  and then applying Theorem 6 of \cite{masry1996multivariate} to our i.i.d. setting.
\end{proof}

We now move on to establishing asymptotic results for $\hat{q}_{1-d}(y)$ for
$y \in \mathcal{Y}:= \{g_{d}(x): x \in \mathcal{F}\}$. We first derive an
asymptotic normality result for $q^{\ast}_{1-d}$, which is the infeasible
estimator that uses $g_{1-d}$ instead of the estimated
$\tilde{g}_{d}$.\footnote{Formally,
  ${q}^{\ast}_{0}(y)=(1,y){\gamma}^{\ast}_{y},$ where
  ${\gamma}^{\ast}_{y}=\arg\min_{\gamma}\sum_{i \in
    \mathcal{I}_{_{1-d}}}W_{i}K\left(\frac{|{g}_{d}(X_{i})-y|}{b}\right)[Y_{i}-\big(1,{g}_{d}(X_{i})\big)\gamma]^{2}$.
} Then, we follow MRS and first derive a stochastic expansion of
$\hat{q}_{1-d}(y)$ around $q^{\ast}_{1-d}(y)$. Combining these results we
provide conditions under which $\hat{q}_{1-d}(y)$ is asymptotically normal.

Due to the fact that $\tilde{g}_{d}$ is independent of
$(X_{i})_{i\in \mathcal{I}_{1-d}}$, the more involved assumptions of MRS (such
as Assumptions 3 and 4) are unnecessary. Since we rely on local linear
estimation, we need a smoothness assumption on $q_{1-d}$. We assume the
following.
\begin{assumption}[Smoothness of $q_{1-d}$]\label{assu:smooth-q}
 The function $q_{1-d}(\cdot)$ is twice differentiable with uniformly bounded
second derivatives.  
\end{assumption}

Due to the ``local to the boundary'' nature of the estimation procedure, we
impose some regularity condition on the boundary. These assumptions are not
restrictive and satisfied in most empirical contexts.  Let $v_{\varepsilon}$
denote the volume of an $\varepsilon$-ball in $\mathbb{R}^{k}$.

\begin{assumption}[Regularity of $\mathcal{F}$] \label{assu:reg-bound}
  (i) $\mathcal{F}$ is piecewise $C^{2}$ and continuous.
  (ii) For each $x \in \mathcal{F}$ and
  $d \in \{0, 1\}$, 
  $\mathrm{vol}\, (\mathcal{X}_{d} \cap B_{\varepsilon}(x))/v_{\varepsilon} \to
  \kappa_{d}(x) $ as $\varepsilon \to 0$ where
  $\kappa_{d}(x) \in [\delta, 1-\delta]$ for some $\delta > 0$.  
\end{assumption}

\begin{assumption}[Locally nonvanishing $\nabla g_{d}(\cdot)$] \label{assu:nz-gprime}
   The conditional expectation function $g_{d}$ satisfies $\inf_{x \in B_{\delta_{g}}(y)} \norm{\nabla g_{d}(x)} > 0
  $ for some $\delta_{g} > 0$.
\end{assumption}

Assumption \ref{assu:reg-bound}(i) and \ref{assu:nz-gprime} are used to perform
a change of variables to calculate the density of $g_{d}(X)$ local to the
boundary. They can be relaxed at the cost of a longer proof. Assumption
\ref{assu:reg-bound}(ii) ensures that there is enough volume in either side of
the boundary, which ensures that the probabilty that $X$ lies on either side of
the boundary scales like $\varepsilon^{d}$ near the boundary. These assumptions
ensure that $P(W_{i}=1|X_{i}=1)$ scales linearly in $\varepsilon$.

Define
$f_{g}(u) = v_{1} \int_{g^{-1}_{d}(u) \cap \mathcal{F}} \kappa_{d}(z)
\frac{f_{X}(z)}{\lVert\nabla g_d(z)\rVert} d\mathcal H^{k-2}(z),
$\footnote{$\mathcal{H}^{k-2}$ denotes the $k-2$ dimensional Hausdorff measure.}
which is effectively the density of $g_{d}(X_{i})$ that accounts for the weight
$W_{i}$ at a boundary point, the conditional variance
$\sigma^{2}(y) := E[(Y_{i}-E[Y_{i}|g_{x}(X_{i})])^{2}| g_{d}(X_{i}) = y]$ and
second moment of the kernel
$\mu_{2} = \int u^{2}K(u)du$. The following theorem establishes the asymptotic
normality of $q^{\ast}_{1-d}(y)$.
\begin{thm}[Asymptotic normality of $q^*_{1-d}(y)$]\label{thm:an-qstar} Suppose
  Assumptions \ref{assu:supp}-\ref{assu:nz-gprime} hold, $a_{n}/b \to 0$, $n\varepsilon^{d}
  \to \infty$, and $n \varepsilon b \to \infty$. Then, we have
  \begin{equation*}
    (n \varepsilon b)^{1/2} \left({q}^{\ast}_{1-d}(y) - q_{1-d}(y) - \frac{b^{2}}{2} q''(y)\mu_{2}\right) \to N\left(0,\frac{ \sigma^2(y) R(K)}{f_g(y)}\right).
  \end{equation*}
\end{thm}

The asymptotic distribution is essentially the same as what one would expect
from a local regression of $Y_{i}$ on $g_{d}(X_{i})$. The difference is that the
weights affect the variance through the normalizing factor
$(n\varepsilon b)^{1/2}$ and density $f_{g}(y)$. However, the weights do not
affect the bias term.

We now analyze how close $q^{\ast}_{1-d}(y)$ is to $\hat{q}_{1-d}(y)$, and
derive conditions under which the difference is negligible. Similar to MRS, define
\begin{equation*}
  \hat{\Delta}(y) = E[
  (\varepsilon^{-1}W_{i})\,b^{-1}K(({g}_{d}(X_{i}) - y))/b
  ( \tilde{g}_{d}(X_i) - g_{d}(X_i)
   )] /f_{g}(y).
\end{equation*}
To derive uniform results, we need tail conditions on the regression error
$\eta_{i} := Y_{i}-g_{d}(X_{i})$.  We impose a sub-exponential moment condition
on the error term as follows, which is the same with Assumption 1(iv) of MRS.

\begin{assumption}[Tail behavior of $\eta_i$]\label{assu:eta-tail}
The regression error $\eta_{i}$ satisfies $E[\exp(\ell |\eta_{i}|) \mid
g_{d}(X_{i})] \leq C$ for some $\ell, C > 0$. 
\end{assumption}

The following expansion, which is a minor modification of that by
\cite{mammen2012nonparametric}, characterizes the difference between the oracle
and feasible estimators.
\begin{thm}[Expansion of $\hat{q}_{1-d}(y) 
    - {q}^{*}_{1-d}(y)$]\label{thm:qstar-qhat}  Suppose
  Assumptions \ref{assu:supp}-\ref{assu:eta-tail} hold. Then, we have
  \begin{align*}
    \sup_{y} 
    \big| \hat{q}_{1-d}(y) 
    - {q}^{*}_{1-d}(y) 
    + q_{1-d}'(y) \hat{\Delta}(y) 
    \big| 
    &= O_P(c_{n}), 
  \end{align*}
  where $c_{n} =(\log n\varepsilon)^{1/2}((n\varepsilon)^{-1/2}a_{n}b^{-3/2} \vee a_{n}b \vee
a_{n}^{2}/b))$.
\end{thm}

Compared with MRS, the main differences are the presence of the
weights $W_{i}$ and the fact that the estimation of the first stage is
independent of the second stage data. The former requires minor modifications to
the proof and latter in fact makes the proof easier. This difference is also
reflected in the convergence rate; in MRS, the first component of the rate
$((n\varepsilon)^{-1/2}a_{n}b^{-3/2}$ includes an additional correction term to
account for the possible correlation between the two stages, which leads to
slower convergence rates.

Also, under possible violations of comonotonicity, the expansion includes an
additional bias term
$\hat{\Gamma}(y) := E[ (W_{i}/\varepsilon)\,K_b'({g}_{d}(X_{i}) - y)
(\tilde{g}_{d}(X_i) - g_{d}(X_i) )\xi_{i}]$. This is because
$\xi_{i}:=E[Y_{i}- q_{1-d}(g_{d}(X_{i}))|X_{i}] \neq 0 $ in general when
comonotonicity is violated. However, we note that dropping the assumption that
$\xi_{i} = 0$, does not affect the convergence rate above or the following
corollary.

If $(n \varepsilon b)^{1/2}c_{n} = o(1)$, the first stage estimation is
negligible so that the asymptotic distribution derived in Theorem
\ref{thm:an-qstar} holds with $\hat{q}_{1-d}(y)$ in place of $q^{\ast}_{1-d}(y)$
as well. The following corollary shows when this is possible under $k=2$.

\begin{cor} Suppose Assumptions \ref{assu:supp}-\ref{assu:eta-tail} hold, and
  let $k=2$, $\varepsilon \asymp h \asymp n^{-r_{h}}$ and $b \asymp
  n^{-r_{b}}$. Then, if $r_h \in (1/8, 4/15)$ and $r_{b}\in (\{\max\{ (3/4)r_h,\, {(1 - 5r_h)}/{3} \}, \min\{ 1/2 - r_h,\, 2r_h,\, 1 - 3r_h,\, 9r_h - 1 \})$, then
\begin{equation*}
  (n \varepsilon b)^{1/2} \left(\hat{q}_{1-d}(y) - q_{1-d}(y) - \frac{b^{2}}{2} q''(y)\mu_{2}\right) \to N\left(0,\frac{ \sigma^2(y) R(K)}{f_g(y)}\right).
\end{equation*}
 Moreover, setting $(1 - r_{h})/5 < r_{b}$ corresponding to the undersmoothing
 regime where the $O(b^{2})$ bias term can be dropped.
\end{cor}

For example, when the mean squared error optimal bandwidth of $r_{h} = 1/6$ is
chosen for the first stage, the bandwidth choice of $r_{b} \in (1/6, 1/3)$ is
feasible for the second stage. That is, one can choose $r_{b} = 1/5$ yet
``undersmooth.'' This seems counterintuitive but in fact it is due to the fact
that the first stage bandwidth enters the second stage variance that reduces the
effect of the bias in the second stage.

We leave the development of a full inferential theory under minimal conditions
for future work, including the problem of testing comonotonicity (see the
discussion following Theorem 2). However, the results of this section show that
for appropriate bandwidth choices, the estimate $\hat{q}_{1-d}$ is
asymptotically equivalent to an oracle estimator which has the form of a
standard local linear regression estimate. It is shown in, for example,
\cite{chernozhukov2014gaussian}, that for such estimates multiplier bootstrap
confidence bands are valid under appropriate assumptions, which we implement in
our empirical application. Moreover, extending the robust bias correction method
by \cite{calonico2014robust} to our setting with generated regressors seems to
be a promising direction for future research.

\section{Application: The Impact of Mandatory Summer School}
We apply the methods to the empirical setting of \citet{Matsudaira2008} who investigate the impact of mandatory summer schooling on future test scores in reading and mathematics. The authors use data on 5th-grade students in a large school district in the northeastern US in 2001. Students in this district scored below a threshold on year-end reading or math tests were required to either repeat a grade or attend mandatory summer schooling. The primary outcomes considered are reading and math test scores one year after the initial tests. It is important to note that some students attended summer school regardless, leading to an intent-to-treat interpretation of the results. The test scores that determine mandatory summer school are observed by the researcher, and so conditional on these observed covariates, treatment is deterministic. 

Following \citet{Imbens2019}, we include only students with scores 40 points above or below the cutoffs, leaving a sample size of $n = 30,741$. The distribution of test scores in this restricted sample is presented in Figure \ref{fig:score_distribution}. The blue curve is the frontier: students whose scores were below either cutoff were faced with the option of summer school or grade repetition.
\begin{figure}[ht]
	\centering
	\includegraphics[scale=0.20]{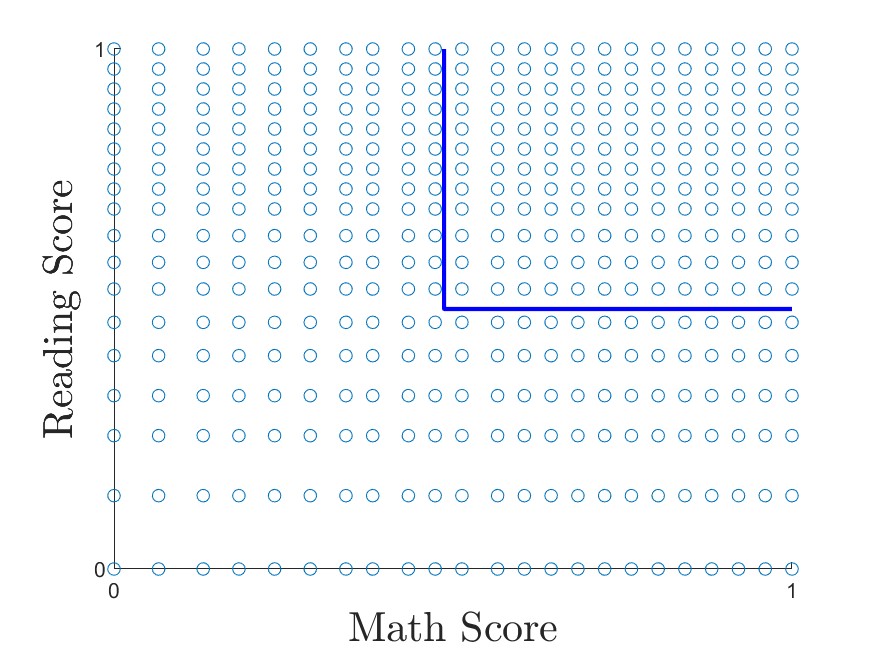}
	\caption{Test Score Distribution}
	\label{fig:score_distribution}
\end{figure}

Below we provide the results from regressing the outcome on covariates using only observations in either the treated or untreated.
\begin{figure}[ht]
		\centering
	\subfloat[]{
	\includegraphics[scale=0.20]{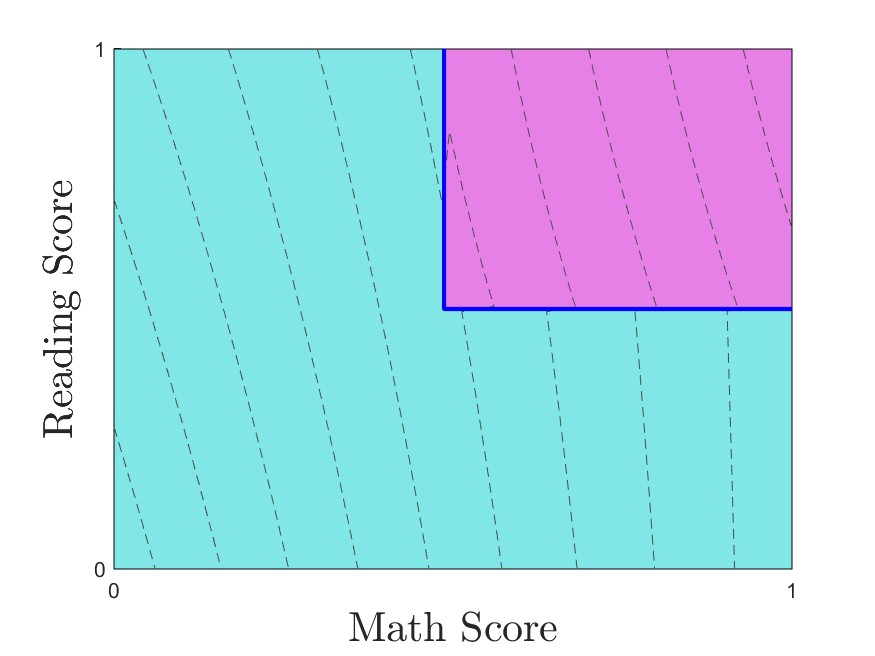}
	\label{fig:quadratic_math}
}
	\subfloat[]{
	\includegraphics[scale=0.20]{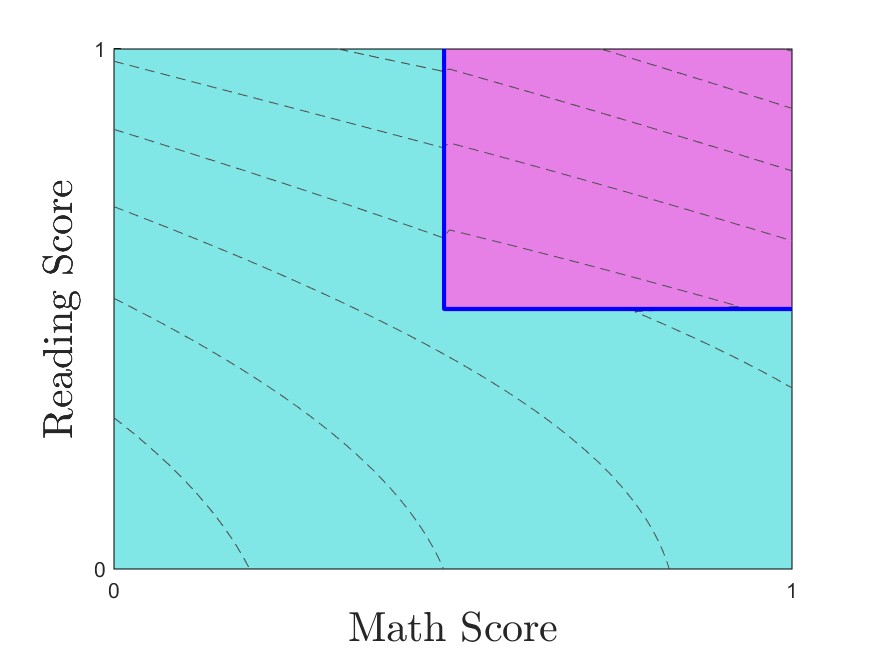}
	\label{fig:quadratic_reading}
}\\
	\subfloat[]{
	\includegraphics[scale=0.20]{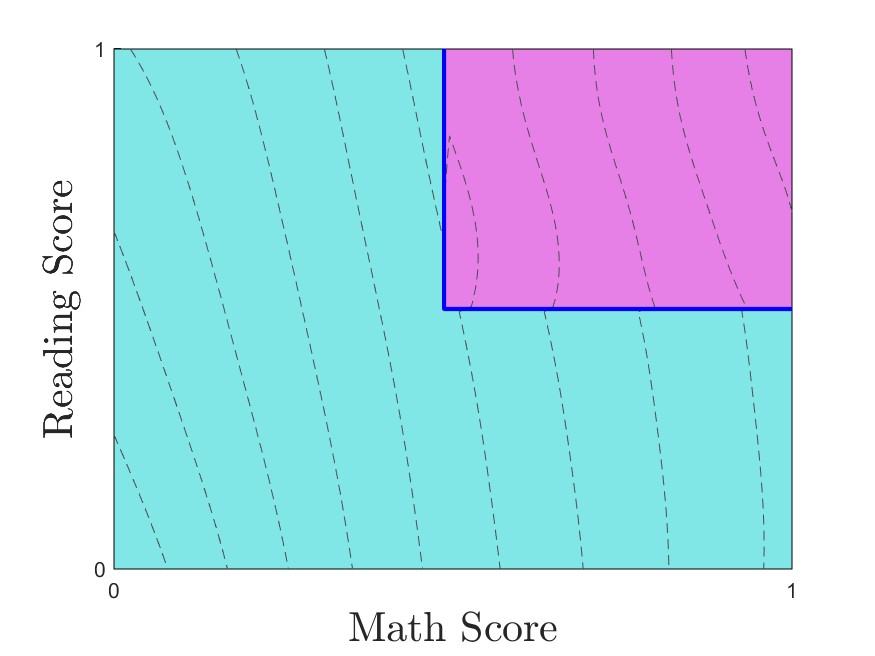}
	\label{fig:local_math}
}
\subfloat[]{
	\includegraphics[scale=0.20]{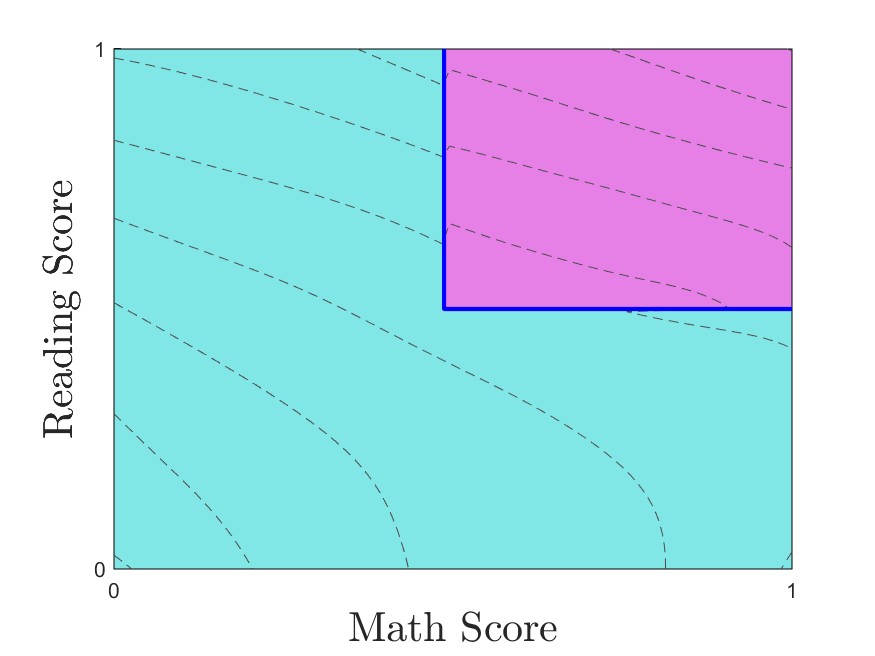}
\label{fig:local_reading}
}
\caption{Contour Estimates}

\end{figure}
Figure \ref{fig:quadratic_math} plots quadratic regression results for the math scores. The contours in the magenta untreated region correspond to regression using only untreated individuals while the contours in the turquoise treated region from regression only on treated individuals. Figure \ref{fig:quadratic_reading} presents results from the same exercise using reading scores. In both cases, the contours in the two regions appear approximately aligned, consistent with the hypothesis of comonotonicity. Figures \ref{fig:local_math} and \ref{fig:local_reading} plot results from the same exercise but with local linear regression instead of quadratic and gives similar results to quadratic.

In this setting, we may justify the comonotonicity condition by considering the underlying skills in reading and math measured by the test scores. In Appendix B.1 we specify a formal model of skill formation in which test scores are noisy measures of underlying skills and show that the model implies comonotonicity. Here we provide a verbal summary. Suppose students begin with some levels of skill in reading and math and the initial tests are noisy measures of these skills. If reading and math skills are positively correlated in the population, then both reading and math test scores are informative about the reading skill, even if only the reading test directly measures this quantity. Likewise, both reading and math scores are informative about math skills. 

The outcomes, reading and math scores a year after initial tests, may be understood to measure reading and math skills at this later date. The reading outcome would thus reflect reading skill at this date which may be depend on reading skill at the time of the initial test, the (possibly heterogeneous) impact of treatment, and exogenous day-of-test noise. One may then expect that the reading outcome is positively correlated with both initial reading and math scores due to the mutual positive association with initial reading skill, and similarly for the math outcome. Indeed, in Figures \ref{fig:quadratic_reading} and \ref{fig:local_reading} we see that this is the case. A higher initial math score is associated with a higher reading outcome, but this association is weaker than that between initial reading scores and the reading outcome, which presumably reflects that the initial reading tests directly measure initial reading skill and thus have a higher association with this quantity than the math scores. A similar pattern is evident for the math outcome.

Now suppose some initial test scores $x$ in the treated region are associated with higher average reading outcomes than some scores $x'$. Suppose we accept the premises above, then this observation suggests individuals with scores $x$ tend to have higher initial reading skills than those with scores $x'$. Both the treated and untreated potential outcomes are increasing in initial reading skills. As such, we may expect individuals with scores $x$ to also have higher average reading outcomes than those with $x'$ under a counterfactual in which none are treated. Indeed we show formally that comonotonicity applies under the skill formation model in Appendix B.1.

\begin{figure}[ht]
	\centering
	\subfloat[]{
	\includegraphics[scale=0.16]{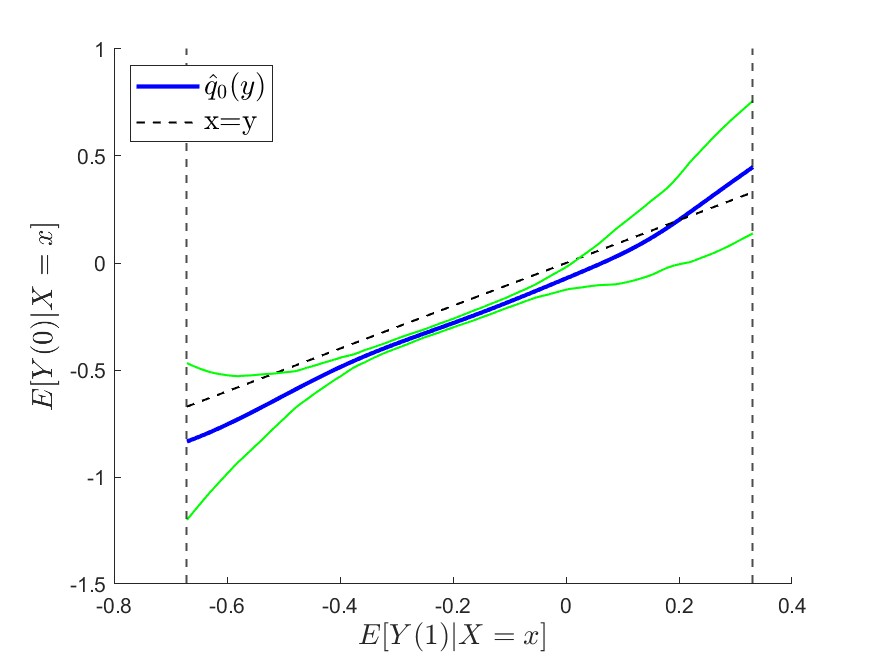}
	\label{fig:q0_math}
}
	\subfloat[]{
	\includegraphics[scale=0.16]{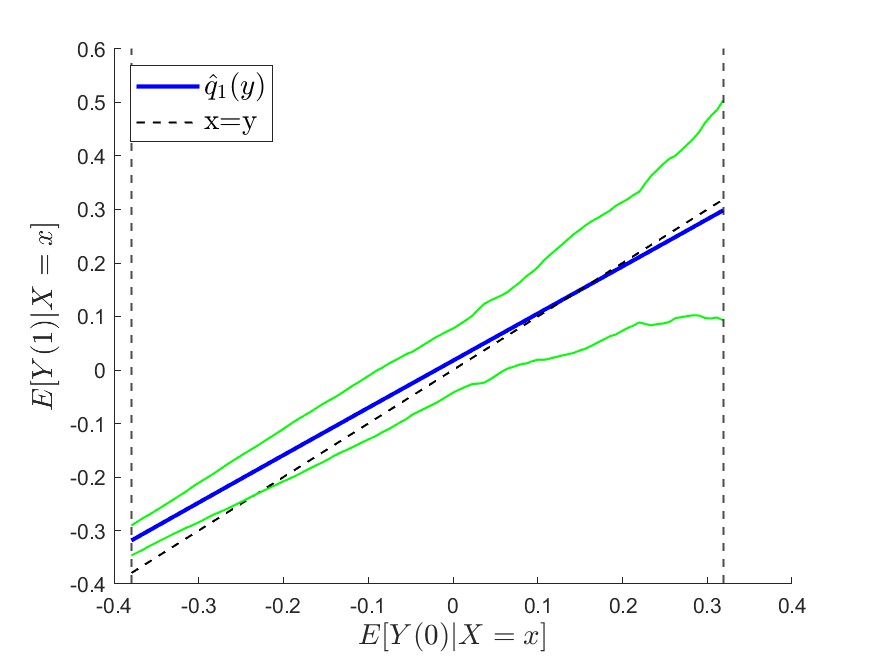}
	\label{fig:q1_math}
}
	\subfloat[]{
	\includegraphics[scale=0.16]{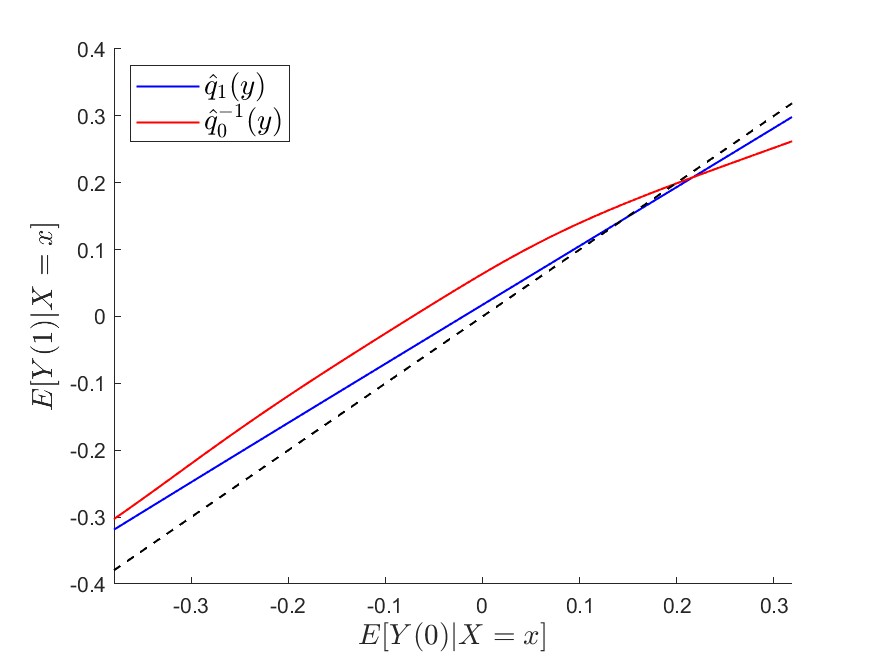}
	\label{fig:q_compare_math}
}\\
	\centering
\subfloat[]{
	\includegraphics[scale=0.16]{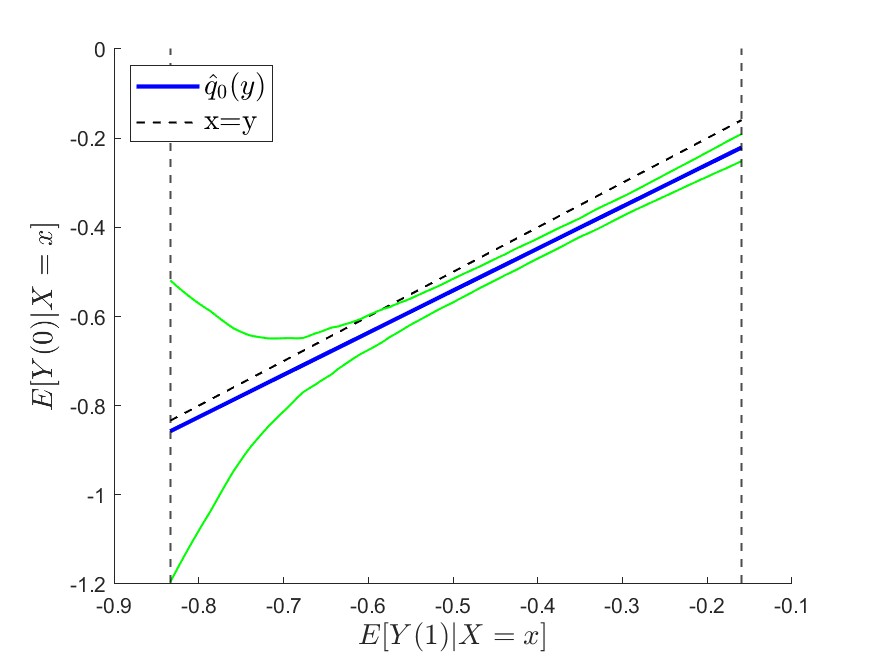}
	\label{fig:q0_read}
}
\subfloat[]{
	\includegraphics[scale=0.16]{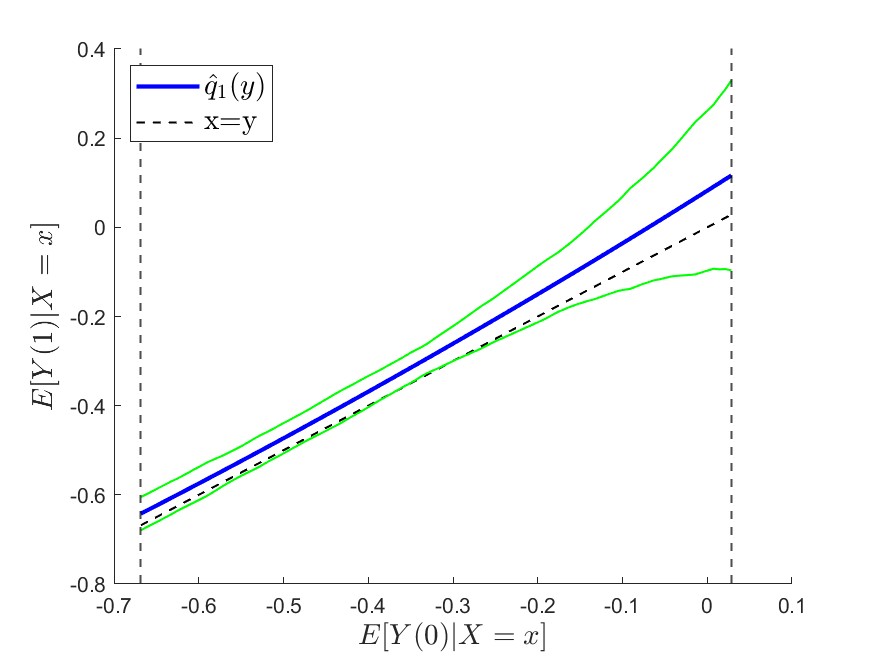}
	\label{fig:q1_read}
}
\subfloat[]{
	\includegraphics[scale=0.16]{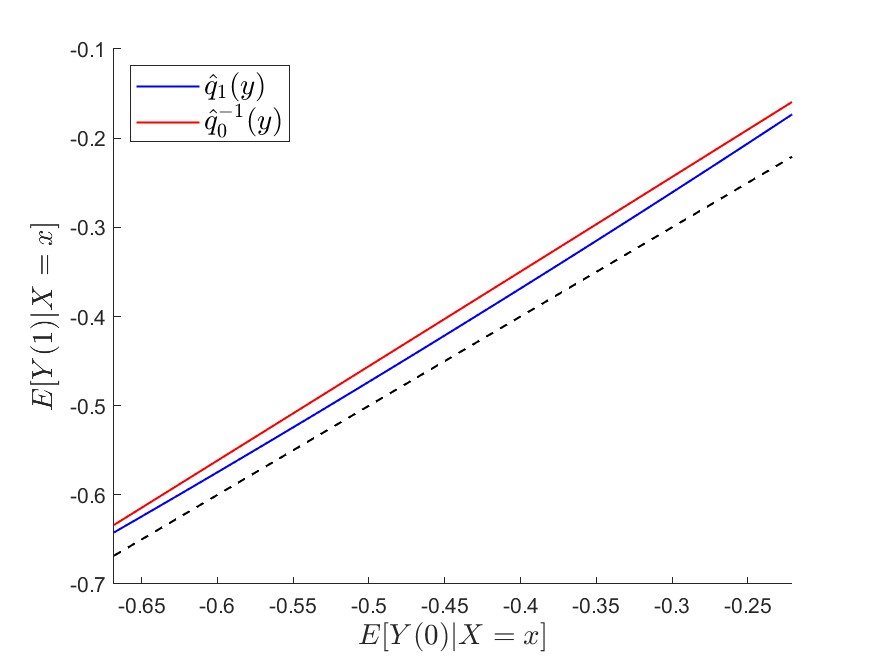}
	\label{fig:q_compare_read}
}
\caption{Conditional Mean Potential Outcomes}
\end{figure}

Figures \ref{fig:q0_math} and \ref{fig:q1_math} present estimates of $\hat{q}_0$ and $\hat{q}_1$ respectively for math scores using the methods in Section 3. $\hat{q}_0$ and $\hat{q}_1$ are respectively below and above the 45 degree line for all but the largest values of the conditional average potential outcomes. Thus the estimates provide evidence of a positive treatment effect for individuals with low and moderate conditional average baseline potential outcomes. More generally the slope of $\hat{q}_0$ ($\hat{q}_1$) is greater (lower) than that of the 45 degree line, which suggests the treatment effect is lower on average for individuals whose covariate values are associated with greater untreated potential outcomes. That is, that individuals whose scores would be worse without treatment, tend to benefit more from treatment. In all figures, green curves are upper and lower $90\%$ pointwise confidence bands evaluated using the multiplier bootstrap.\footnote{To calculate the multiplier boostrap intervals, for each bootstrap simulation $s$ we draw $n$ independent standard exponential random variables. In each of the kernel sums in the first and second stages of the estimator $\hat{q}_d$ we multiply the $i$-th term by the corresponding exponential random variable to obtain a bootstrap estimate $\hat{q}_{d,s}$. The band at $y$ is then $\hat{q}_{d}(y)$ plus and minus the $90$-th percentile of $|\hat{q}_{d,s}(y)-\hat{q}_{d}(y)|$ over the bootstrap draws. We use $100$ bootstrap draws.}

Figure \ref{fig:q_compare_math} compares these estimates. Under comonotonicity the population versions of these functions should align, and indeed in \ref{fig:q_compare_math} these appear close. 

Figures \ref{fig:q0_read} and \ref{fig:q1_read} present estimates of $\hat{q}_0$ and $\hat{q}_1$ for reading scores, respectively. In this case, conditional average treatment effect estimates are positive for all individuals and increasing in conditional average baseline potential outcomes. Figure \ref{fig:q_compare_read} compares these estimates which are closely aligned.
\begin{figure}[ht]
	\centering
	\subfloat[]{
		\includegraphics[scale=0.20]{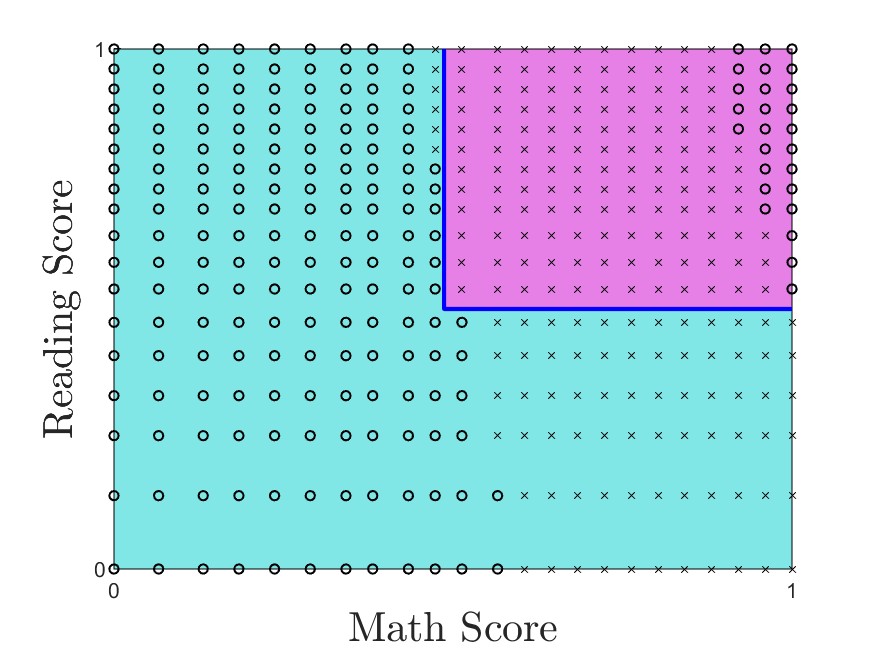}
		\label{fig:extrap_math}
	}
	\subfloat[]{
		\includegraphics[scale=0.20]{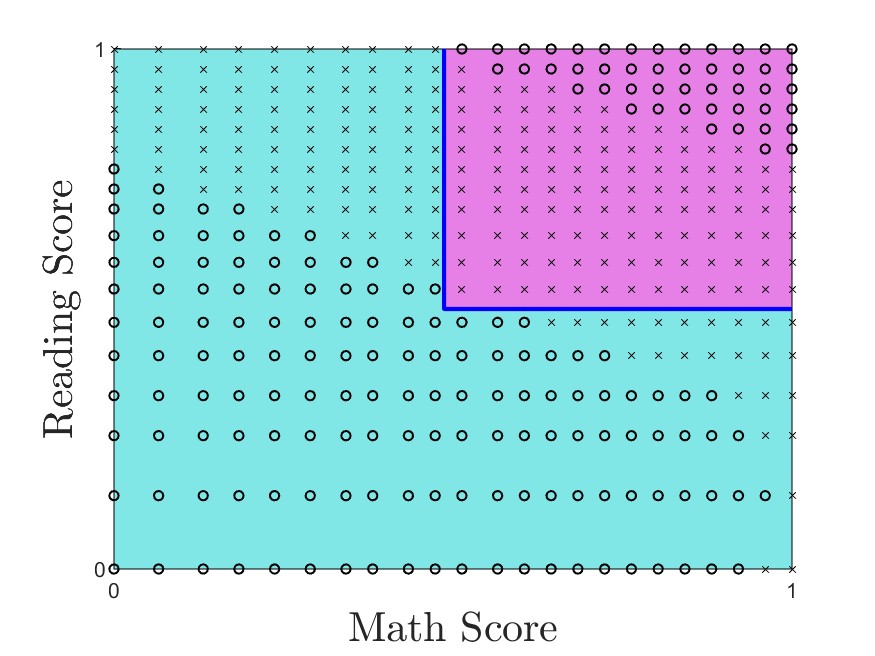}
		\label{fig:extrap_read}
	}
	
	\caption{Values of $S_i$}
\end{figure}

Figures \ref{fig:extrap_math} and \ref{fig:extrap_read} respectively display the
covariate values at which we are able to extrapolate conditional average
treatment effects. Crosses represent values at which identification is achieved
under comonotonicity and circles indicate points to which we cannot
extrapolate. In terms of the notation introduced in Section 3, points marked
with crosses represent values of $X_i$ for which we estimate $S_i=1$ and those
with circles, values for which $S_i=0$.

Figures \ref{fig:extrap_math} and \ref{fig:extrap_read} indicate which
counterfactual objects one can identify using comonotonicity in this
setting. For example, Figure \ref{fig:extrap_math} suggests that we can identify
the average effect on math outcomes of counterfactual treatment policies in
which the math score threshold for mandatory summer school is moderately
increased. However, the average effect of a moderate decrease in this cut-off is
not identified. Similarly, Figure \ref{fig:extrap_read} suggests that the
average counterfactual effect on reading outcomes of a moderate increase in the
reading score is identified.

\begin{figure}[ht]
	\centering
	\subfloat[]{
		\includegraphics[scale=0.20]{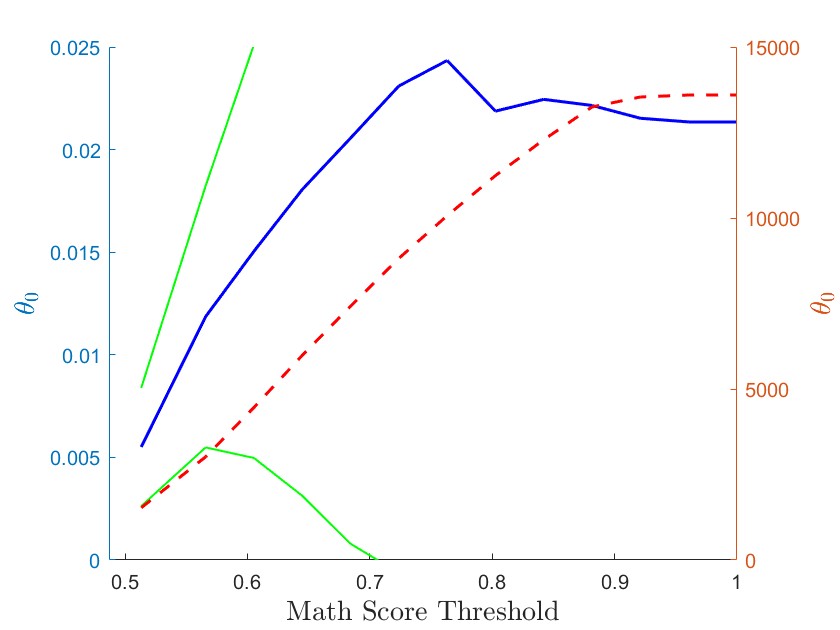}
		\label{fig:counter_math}
	}
	\subfloat[]{
		\includegraphics[scale=0.20]{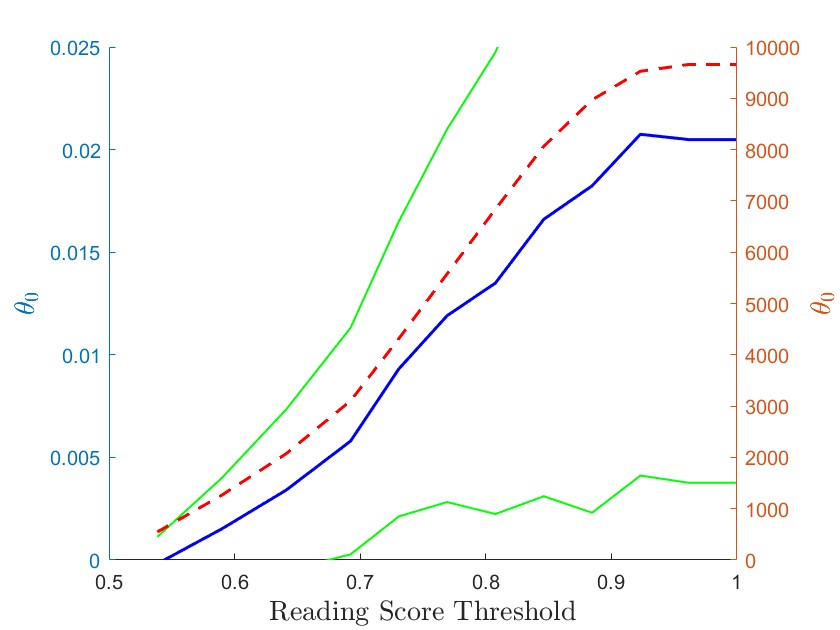}
		\label{fig:counter_read}
	}
	
	\caption{Effects of Counterfactual Policies}
\end{figure}

We estimate average causal effects on math scores from counterfactual increases
in the math cut-off and the effects on reading outcomes of increases in the
reading threshold. The methods in Section 3 allow us to estimate the average
effects among those individuals for whom the CATE is identified, i.e., for whom
$S_i=1$. From Figures \ref{fig:extrap_math} and \ref{fig:extrap_read} we see
that for all but very large increases in the threshold, every individual
impacted by the increase in the cut-off has $S_i=1$ and so for all but the very
largest increases in the cut-off these conditional (on $S_i=1$) effects are
exactly the unconditional average effects. The solid blue curves in
\ref{fig:counter_math} and \ref{fig:counter_read} plot these causal effects for
math and reading respectively. The dashed red curves show the number of
individuals in the sample with $S_i=1$ whose treatment status is impacted by the
counterfactual policy.

Figure \ref{fig:counter_math}  shows that as the math threshold increases, the
causal effect increases initially but the rate at a rate that slows with the
level of the threshold. When the threshold is raised to $0.75$ increases in the
threshold lead to a slight decrease in the effect which then levels out. The
number of individuals with $S_i=1$ impacted by the policy increases roughly
linearly with the threshold until the threshold reaches around $0.9$ before it
levels off, which may reflect the lack of identification of causal effects for
individuals with very high initial math scores seen in
\ref{fig:extrap_math}. The slow down and eventual reversal of the increase in
the causal effect up to threshold level $0.9$ can be understood from Figures
\ref{fig:local_math} and \ref{fig:q0_math}. When the threshold is high, any
marginal increase impacts only individuals with high initial math scores and
from \ref{fig:local_math} we see these individuals have high average untreated
potential outcomes.  \ref{fig:q0_math} suggests that this is associated with a
smaller CATE, and thus a smaller increase in the causal effect of the policy. 

By contrast, \ref{fig:q0_read} shows a steadily increasing relationship between
the CATE and conditional average untreated potential outcome and the causal
effect shown in \ref{fig:counter_read} increases at a similar rate to the number
of individuals impacted by the policy. Note that in both cases, the multiplier
bootstrap confidence bands exclude zero average impact of some of the included
counterfactual policy rules.

		\bibliographystyle{authordate1}
\bibliography{RDDcites}

\begin{thebibliography}{}

\bibitem[\protect\citename{Alan {\em et~al.}, }2019]{Alan}
Alan, Sule, Boneva, Teodora, \& Ertac, Seda. 2019.
\newblock Ever Failed, Try Again, Succeed Better: Results from a Randomized
  Educational Intervention on Grit*.
\newblock {\em The Quarterly Journal of Economics}.

\bibitem[\protect\citename{Andrews, }1992]{andrews1992generic}
Andrews, Donald~WK. 1992.
\newblock Generic uniform convergence.
\newblock {\em Econometric theory}, {\bf 8}(2), 241--257.

\bibitem[\protect\citename{Angrist \& Rokkanen, }2015]{Angrist2015}
Angrist, Joshua~D., \& Rokkanen, Miikka. 2015.
\newblock Wanna Get Away? Regression Discontinuity Estimation of Exam School
  Effects Away From the Cutoff.
\newblock {\em JASA}, {\bf 110}, 1331--1344.

\bibitem[\protect\citename{Armstrong \& Koles{\'a}r,
  }2018]{armstrong2018optimal}
Armstrong, Timothy~B, \& Koles{\'a}r, Michal. 2018.
\newblock Optimal inference in a class of regression models.
\newblock {\em Econometrica}, {\bf 86}(2), 655--683.

\bibitem[\protect\citename{Battistin \& Rettore, }2008]{Battistin2008}
Battistin, Erich, \& Rettore, Enrico. 2008.
\newblock Ineligibles and eligible non-participants as a double comparison
  group in regression-discontinuity designs.
\newblock {\em Journal of Economtrics}, {\bf 142}, 715--730.

\bibitem[\protect\citename{Calonico {\em et~al.}, }2014]{calonico2014robust}
Calonico, Sebastian, Cattaneo, Matias~D, \& Titiunik, Rocio. 2014.
\newblock Robust nonparametric confidence intervals for
  regression-discontinuity designs.
\newblock {\em Econometrica}, {\bf 82}(6), 2295--2326.

\bibitem[\protect\citename{Calonico {\em et~al.}, }2019]{Calonico2019}
Calonico, Sebastian, Cattaneo, Matias~D, Farrell, Max~H, \& Titiunik, Rocio.
  2019.
\newblock Regression discontinuity designs using covariates.
\newblock {\em Review of Economics and Statistics}, {\bf 101}(3), 442--451.

\bibitem[\protect\citename{Cattaneo \& Titiunik, }2022]{Cattaneo2022}
Cattaneo, Matias~D, \& Titiunik, Rocio. 2022.
\newblock Regression discontinuity designs.
\newblock {\em Annual Review of Economics}, {\bf 14}(1), 821--851.

\bibitem[\protect\citename{Cattaneo {\em et~al.}, }2021]{Cattaneo2021}
Cattaneo, Matias~D., Keele, Luke, Titiunik, Rocío, \& Vazquez-Bare, Gonzalo.
  2021.
\newblock Extrapolating Treatment Effects in Multi-Cutoff Regression
  Discontinuity Designs.
\newblock {\em JASA}, {\bf 116}, 1941--1952.

\bibitem[\protect\citename{Chernozhukov \& Hansen, }2005]{Chernozhukov2005}
Chernozhukov, Victor, \& Hansen, Christian. 2005.
\newblock An IV Model of Quantile Treatment Effects.
\newblock {\em Econometrica}, {\bf 73}(1), 245--261.

\bibitem[\protect\citename{Chernozhukov {\em et~al.}, }2007]{Chernozhukov2007}
Chernozhukov, Victor, Imbens, Guido~W., \& Newey, Whitney~K. 2007.
\newblock Instrumental variable estimation of nonseparable models.
\newblock {\em Journal of Econometrics}, {\bf 139}(1), 4--14.

\bibitem[\protect\citename{Chernozhukov {\em et~al.},
  }2013]{chernozhukov2013intersection}
Chernozhukov, Victor, Lee, Sokbae, \& Rosen, Adam~M. 2013.
\newblock Intersection bounds: Estimation and inference.
\newblock {\em Econometrica}, {\bf 81}(2), 667--737.

\bibitem[\protect\citename{Chernozhukov {\em et~al.},
  }2014]{chernozhukov2014gaussian}
Chernozhukov, Victor, Chetverikov, Denis, \& Kato, Kengo. 2014.
\newblock Gaussian approximation of suprema of empirical processes.

\bibitem[\protect\citename{D'Haultfœuille \& Février,
  }2015]{DHaultfoeuille2015}
D'Haultfœuille, Xavier, \& Février, Philippe. 2015.
\newblock Identification of Nonseparable Triangular Models With Discrete
  Instruments.
\newblock {\em Econometrica}, {\bf 83}(3), 1199--1210.

\bibitem[\protect\citename{DiNardo \& Lee, }2011]{DiNardo2011}
DiNardo, John, \& Lee, David~S. 2011.
\newblock {\em Program Evaluation and Research Designs}.
\newblock Chap. Chapter 5, Handbook of Labor Economics, pages  463--536.

\bibitem[\protect\citename{Dong \& Lewbel, }2015]{Dong2015}
Dong, Yingying, \& Lewbel, Arthur. 2015.
\newblock Identifying the Effect of Changing the Policy Threshold in Regression
  Discontinuity Models.
\newblock {\em Review of Economics and Statistics}, {\bf 97}, 1081--1092.

\bibitem[\protect\citename{Fan \& Gijbels, }1996]{fan1996local}
Fan, Jianqing, \& Gijbels, Irene. 1996.
\newblock {\em Local Polynomial Modelling and Its Applications: Monographs on
  Statistics and Applied Probability 66}.
\newblock  Vol. 66.
\newblock CRC Press.

\bibitem[\protect\citename{Grembi {\em et~al.}, }2016]{Grembi2016}
Grembi, Veronica, Nannicini, Tommaso, \& Troiano, Ugo. 2016.
\newblock Do Fiscal Rules Matter?
\newblock {\em American Economic Journal: Applied Economics}, {\bf 8}(3),
  1--30.

\bibitem[\protect\citename{Hansen, }2008]{hansen2008uniform}
Hansen, Bruce~E. 2008.
\newblock Uniform Convergence Rates for Kernel Estimation with Dependent Data.
\newblock {\em Econometric Theory}, {\bf 24}(3), 726--748.

\bibitem[\protect\citename{Horowitz \& Lee, }2007]{Horowitz2007}
Horowitz, Joel~L., \& Lee, Sokbae. 2007.
\newblock Nonparametric Instrumental Variables Estimation of a Quantile
  Regression Model.
\newblock {\em Econometrica}, {\bf 75}(4), 1191--1208.

\bibitem[\protect\citename{Imbens \& Newey, }2009]{Imbens2009}
Imbens, Guido, \& Newey, Whitney. 2009.
\newblock Identification and Estimation of Triangular Simultaneous Equations
  Models Without Additivity.
\newblock {\em Econometrica}, {\bf 77}(5), 1481--1512.

\bibitem[\protect\citename{Imbens \& Wager, }2019]{Imbens2019}
Imbens, Guido, \& Wager, Stefan. 2019.
\newblock Optimized Regression Discontinuity Designs.
\newblock {\em The Review of Economics and Statistics}, {\bf 101}, 264--278.

\bibitem[\protect\citename{Keele \& Titiunik, }2015]{Keele2015}
Keele, Luke~J, \& Titiunik, Rocio. 2015.
\newblock Geographic boundaries as regression discontinuities.
\newblock {\em Political Analysis}, {\bf 23}(1), 127--155.

\bibitem[\protect\citename{Mammen {\em et~al.}, }2012]{mammen2012nonparametric}
Mammen, Enno, Rothe, Christoph, \& Schienle, Melanie. 2012.
\newblock Nonparametric regression with nonparametrically generated covariates.

\bibitem[\protect\citename{Masry, }1996]{masry1996multivariate}
Masry, Elias. 1996.
\newblock Multivariate local polynomial regression for time series: uniform
  strong consistency and rates.
\newblock {\em Journal of Time Series Analysis}, {\bf 17}(6), 571--599.

\bibitem[\protect\citename{Matsudaira, }2008]{Matsudaira2008}
Matsudaira, Jordan~D. 2008.
\newblock Mandatory summer school and student achievement.
\newblock {\em Journal of Econometrics}, {\bf 142}, 829--850.

\bibitem[\protect\citename{Mealli \& Rampichini, }2012]{Mealli2012}
Mealli, Fabrizia, \& Rampichini, Carla. 2012.
\newblock Evaluating the Effects of University Grants by Using Regression
  Discontinuity Designs.
\newblock {\em Journal of the Royal Statistical Society: Series A}, {\bf 175},
  775--798.

\bibitem[\protect\citename{Noack {\em et~al.}, }2024]{Noack2024}
Noack, Claudia, Olma, Tomasz, \& Rothe, Christoph. 2024.
\newblock Flexible covariate adjustments in regression discontinuity designs.
\newblock {\em arXiv preprint arXiv:2107.07942}.

\bibitem[\protect\citename{Papay {\em et~al.}, }2011]{Papay2011}
Papay, John~P, Willett, John~B, \& Murnane, Richard~J. 2011.
\newblock Extending the regression-discontinuity approach to multiple
  assignment variables.
\newblock {\em Journal of Econometrics}, {\bf 161}(2), 203--207.

\bibitem[\protect\citename{Schmeidler, }1989]{Schmeidler1989}
Schmeidler, David. 1989.
\newblock Subjective Probability and Expected Utility without Additivity.
\newblock {\em Econometrica}, {\bf 57}(3), 571.

\bibitem[\protect\citename{Stone, }1982]{stone1982optimal}
Stone, Charles~J. 1982.
\newblock Optimal {{Global Rates}} of {{Convergence}} for {{Nonparametric
  Regression}}.
\newblock {\em The Annals of Statistics}, {\bf 10}(4), 1040--1053.

\bibitem[\protect\citename{Torgovitsky, }2015]{Torgovitsky2015}
Torgovitsky, Alexander. 2015.
\newblock Identification of Nonseparable Models Using Instruments With Small
  Support.
\newblock {\em Econometrica}, {\bf 83}(3), 1185--1197.

\bibitem[\protect\citename{Wing \& Cook, }2013]{Wing2013}
Wing, Coady, \& Cook, Thomas~D. 2013.
\newblock Strengthening the Regression Discontinuity Design Using Additional
  Design Elements: A Within‐Study Comparison.
\newblock {\em Journal of Policy Analysis and Management}, {\bf 32}(4),
  853--877.

\end{thebibliography}

\appendix

\section{Proofs.}

\begin{proof}[Proof of Theorem 1]
	
	By Assumption 1.1, for any $x\in\mathcal{X}_{d}$, we have $E[Y(d)|X=x]=E[Y|X=x]$,
	and so $g_{d}(x)=E[Y(d)|X=x]$. Moreover, by Assumption 1.2, $E[Y(d)|X=x]$
	is continuous at $x\in\mathcal{F}$. Thus the restriction of $E[Y(d)|X=x]$
	to $\mathcal{X}_{d}\cup\mathcal{F}$ satisfies the restrictions we
	place on $g_{d}$. Thus it suffices to show that there is a unique
	function $g_{d}$ that satisfies these conditions. Clearly $g_{d}(x)$
	is uniquely defined for $x\in\mathcal{X}_{d}$ because it is equal
	to $E[Y|X=x]$ so we need only consider $x\in\mathcal{F}$. By definition,
	such a point $x$ is in the closure of $int(\mathcal{X}_{1})$ and
	the closure of $int(\mathcal{X}_{0})$. Thus there is a sequence $\{x_{d,k}\}_{k=1}^{\infty}$
	in $int(\mathcal{X}_{1})$ so that $x_{d,k}\to x$. Then by the continuity
	condition Assumption 1.2, for any such :a sequence
	
	\begin{align*}
		E[Y(d)|X=x] & =\lim_{k\to\infty}E[Y(d)|X=x_{d,k}]\\
		& =\lim_{k\to\infty}g_{d}(x_{d,k})
	\end{align*}
	Where the second equality uses that $g_{d}(x)=E[Y(d)|X=x]$ for all
	$x\in\mathcal{X}_{d}$. Since $g_{d}$ is also continuous at $\mathcal{F}$
	we thus have: 
	\[
	g_{d}(x)=E[Y(d)|X=x]
	\]
	So $g_{d}(x)$ is also uniquely defined for $x\in\mathcal{F}$. \end{proof}

\begin{proof}[Proof of Theorem 2]
	
	From Theorem 1, if $x\in\mathcal{X}_{d}$ and $x^{*}\in\mathcal{F}$,
	then $g_{d}(x)=g_{d}(x^{*})$ implies $E[Y(d)|X=x]=E[Y(d)|X=x^{*}]$,
	so from Assumption 2 we get: 
	\begin{align*}
		E[Y(1-d)|X=x] & =E[Y(1-d)|X=x^{*}]\\
		& =g_{1-d}(x^{*})
	\end{align*}
	
	For the second statement, using $g_{d}(x_{1})=E[Y(d)|X=x_{1}]$ and
	$g_{d}(x_{2})=E[Y(d)|X=x_{2}]$, we see that: 
	\[
	g_{d}(x_{1})\geq g_{d}(x_{2})\iff E[Y(d)|X=x_{1}]\geq E[Y(d)|X=x_{2}]
	\]
	Then by Assumption 2, the above implies: 
	\[
	E[Y(1-d)|X=x_{1}]\geq E[Y(1-d)|X=x_{2}]
	\]
	
	Substituting $g_{1-d}(x_{1})=E[Y(1-d)|X=x_{1}]$ and $g_{1-d}(x_{2})=E[Y(1-d)|X=x_{2}]$
	(which holds by Theorem 1) then gives the result.
	
\end{proof}

\begin{proof}[Proof of Proposition 1]
	
	\textbf{Part a.}
	
	Suppose that $E[Y(0)|X=x_{1}]\geq E[Y(0)|X=x_{2}]$, then by supposition
	$E[\tau|X=x_{1}]\geq E[\tau|X=x_{2}]$, or equivalently:
	
	\begin{align*}
		E[Y(1)|X=x_{1}] & \geq E[Y(1)|X=x_{2}]+E[Y(0)|X=x_{1}]-E[Y(0)|X=x_{2}]\\
		& \geq E[Y(1)|X=x_{2}]
	\end{align*}
	
	Where the second inequality again uses $E[Y(0)|X=x_{1}]\geq E[Y(0)|X=x_{2}]$.
	So we see that: 
	\begin{equation}
		E[Y(0)|X=x_{1}]\geq E[Y(0)|X=x_{2}]\implies E[Y(1)|X=x_{1}]\geq E[Y(1)|X=x_{2}]\label{eq:oneway}
	\end{equation}
	
	Conversely, if $E[Y(0)|X=x_{1}]<E[Y(0)|X=x_{2}]$ then then by supposition
	$E[\tau|X=x_{1}]\leq E[\tau|X=x_{2}]$, or equivalently:
	
	\begin{align*}
		E[Y(1)|X=x_{1}] & \leq E[Y(1)|X=x_{2}]+E[Y(0)|X=x_{1}]-E[Y(0)|X=x_{2}]\\
		& <E[Y(1)|X=x_{2}]
	\end{align*}
	
	And so,
	\begin{equation}
		E[Y(0)|X=x_{1}]<E[Y(0)|X=x_{2}]\implies E[Y(1)|X=x_{1}]<E[Y(1)|X=x_{2}]\label{eq:2nd}
	\end{equation}
	
	Together, (\ref{eq:oneway}) and (\ref{eq:2nd})are equivalent to
	comonotonicity.
	
	\textbf{Part b.}
	
	First note that by definition of $\tau$: 
	\begin{align}
		E[Y(1)|X=x_{2}]-E[Y(1)|X=x_{1}] & =E[Y(0)|X=x_{2}]-E[Y(0)|X=x_{1}]\nonumber \\
		& +E[\tau|X=x_{2}]-E[\tau|X=x_{1}]\label{eq:compmain}
	\end{align}
	
	From the above we see that: 
	\begin{align}
		E[Y(1)|X=x_{2}]-E[Y(1)|X=x_{1}] & \geq E[Y(0)|X=x_{2}]-E[Y(0)|X=x_{1}]\nonumber \\
		& -|E[\tau|X=x_{2}]-E[\tau|X=x_{1}]|\label{eq:comp2}
	\end{align}
	
	Now, suppose $E[Y(0)|X=x_{2}]\geq E[Y(0)|X=x_{1}]$. Then by supposition
	we have:
	
	\[
	|E[\tau|X=x_{2}]-E[\tau|X=x_{1}]|\leq E[Y(0)|X=x_{2}]-E[Y(0)|X=x_{1}]
	\]
	
	And so from \eqref{eq:comp2}: 
	\[
	E[Y(1)|X=x_{2}]-E[Y(1)|X=x_{1}]\geq0
	\]
	
	So we see $E[Y(0)|X=x_{2}]\geq E[Y(0)|X=x_{1}]$ implies $E[Y(1)|X=x_{2}]\geq E[Y(1)|X=x_{1}]$.
	
	Now, from \eqref{eq:compmain} we also get the following:
	
	\begin{align}
		E[Y(1)|X=x_{2}]-E[Y(1)|X=x_{1}] & \leq E[Y(0)|X=x_{2}]-E[Y(0)|X=x_{1}]\nonumber \\
		& +|E[\tau|X=x_{1}]-E[\tau|X=x_{2}]|\label{eq:comp1}
	\end{align}
	
	Now, suppose $E[Y(0)|X=x_{2}]<E[Y(0)|X=x_{1}]$. Then by supposition
	we have:
	
	\[
	|E[\tau|X=x_{2}]-E[\tau|X=x_{1}]|<E[Y(0)|X=x_{1}]-E[Y(0)|X=x_{2}]
	\]
	
	And so from \eqref{eq:comp1}: 
	\[
	E[Y(1)|X=x_{2}]-E[Y(1)|X=x_{1}]<0
	\]
	
	So we see $E[Y(0)|X=x_{2}]<E[Y(0)|X=x_{1}]$ implies $E[Y(1)|X=x_{2}]<E[Y(1)|X=x_{1}]$
	and thus if $E[Y(1)|X=x_{2}]\geq E[Y(1)|X=x_{1}]$ we must have $E[Y(0)|X=x_{2}]\geq E[Y(0)|X=x_{1}]$
	
\end{proof}

\begin{proof}[Proof of Theorem 3]
	
	From Theorem 1, if $x\in\mathcal{X}_{d}$ and $x^{*}\in\mathcal{F}$,
	then $g_{d}(x)=g_{d}(x^{*})$ implies $E[Y(d)|X=x]=E[Y(d)|X=x^{*}]$,
	so if $x^{*(2)}=x^{(2)}$ then from Assumption 3 we get: 
	\begin{align*}
		E[Y(1-d)|X=x] & =E[Y(1-d)|X=x^{*}]\\
		& =g_{1-d}(x^{*})
	\end{align*}
	
	For the second statement, using $g_{d}(x_{1})=E[Y(d)|X=x_{1}]$ and
	$g_{d}(x_{2})=E[Y(d)|X=x_{2}]$, we see that: 
	\[
	g_{d}(x_{1})\geq g_{d}(x_{2})\iff E[Y(d)|X=x_{1}]\geq E[Y(d)|X=x_{2}]
	\]
	Then if $x_{1}^{(2)}=x_{2}^{(2)}$ then by Assumption 3, the above
	implies: 
	\[
	E[Y(1-d)|X=x_{1}]\geq E[Y(1-d)|X=x_{2}]
	\]
	
	Substituting $g_{1-d}(x_{1})=E[Y(1-d)|X=x_{1}]$ and $g_{1-d}(x_{2})=E[Y(1-d)|X=x_{2}]$
	(which holds by Theorem 1) then gives the result.
	
\end{proof}

\begin{proof}[Proof of Theorem \ref{thm:an-qstar}]
  With some abuse of notation, we redefine the weights as
  $W_{i} = \mathrm{1}\{X_{i} \in \mathcal{X}_{1-d}, \norm{X_{i} - X_{NN(i)}}
  \leq \varepsilon,\,\, X_{{NN(i)}} \text{ exists} \}$. Note that this results in the
  exact same estimator with the advantage being that we can take the summation
  over all $i\in \{1, \dots, n\}$ rather than over a random sample size. Write
  $K_{b}(u) : = K(u/b)/b$. Note that $X_{{NN(i)}}$ exists if and only if at
  least one observation in $\mathcal{X}_{d}$. 
  
The infeasible estimator $q^{\ast}_{1-d}(y)$ is a standard univariate local linear regression but with
weights $W_{i}$. Most of derivation remains identical to the standard case
(e.g., \cite{fan1996local}) and thus we focus on how the
weights affect the asymptotic distribution. Write
\begin{align*}
  \label{eq:3}
  \tilde{S}_r(y) &:= \sum_i W_i\,K_b(\tilde{g}_{d}(X_{i}) -
                   y)\,((\tilde{g}_{d}(X_{i}) - y)/b)^{r}\\
    \tilde{T}_r(y) &:= \sum_i W_i\,K_b(\tilde{g}_{d}(X_{i}) -
                   y)\,((\tilde{g}_{d}(X_{i}) - y)/b)^{r}Y_{i},
\end{align*}
where $\sum_{i}$ indicates $\sum_{i=1}^{n}$ unless
specified otherwise. The argument $y$ is often omitted for notational
simplicity; e.g., we write $\tilde{S}_{0}$ to indicate $\tilde{S}_{0}(y)$ unless
confusing otherwise. Further define $
\tilde{S}(y) :=
  \begin{pmatrix}
    \tilde{S}_{0}(y) & \tilde{S}_{1}(y) \\
    \tilde{S}_{1}(y) & \tilde{S}_{2}(y)
  \end{pmatrix}$ and $\tilde{T}(y) :=   \begin{pmatrix}
    \tilde{T}_{0}(y) \\
    \tilde{T}_{1}(y)
  \end{pmatrix}$ so that
  \(e_{1}'\hat{\gamma}_{y} := \tilde{S}(y)^{-1}\tilde{T}(y),\)
where $\hat{\gamma}(y) $ is defined in \eqref{eq:def-q}. Define $S_{r}(y)$,
$T_{r}(y)$, $S(y)$ and $T(y)$ analogously, but with $g$ in place of $\tilde g$.

We first investigate
$q^{\ast}_{1-d}(y) - q_{1-d}(y)$ where $q_{1-d}^{\ast}(y)$ is obtained by the infeasible local
linear regression that replaces $\tilde{g}$ with $g$. We can write
\begin{equation*}
  q^{\ast}_{1-d}(y) = \frac{S_{2}(y)T_{0}(y) -
                        S_{1}(y)T_{1}(y)}{S_{0}(y)S_{2}(y) - S_{1}^{2}(y)} =: \frac{S_{2}(y)T_{0}(y) -
                        S_{1}(y)T_{1}(y)}{D(y)}
                    \end{equation*}

We first characterize the moments of $S_{r}$. Note that
\begin{align*}
\mu_{r} :=   & E[W_i\,K_b({g}_{d}(X_{i}) - y) (({g}_{d}(X_{i}) - y)/b)^{r}
                  ] \\
  =& E[p_{\varepsilon}( X_{i})K_b({g}_{d}(X_{i}) - y)(({g}_{d}(X_{i}) - y)/b)^{r}
                   ] \\
  = & \int_{x \in \mathcal{X}_{1-d}}p_{\varepsilon}(x)K_{b}({g}_{d}(x) -
      y)(({g}_{d}(x) - y)/b)^{r}f_{X}(x) dx
\end{align*}
where we define $p_{\varepsilon}(x) := P(W_{i}=1 | X_{i} = x)$.

Now, by Assumption \ref{assu:smoooth-bd}, for large enough $n$ a change of
variables give

\begin{equation*}
  \mu_{r} = \int_{\mathbb{R}}K_{b}(u-y)((u-y)/b)^{r} \int_{{g}_d^{-1}(u) \cap \mathcal{X}_{1-d}}
       \frac{p_{\varepsilon}(x)\,f_X(x)}
            {\|\nabla {g}_d(x)\|}\,
     d \mathcal{H}^{k-1}(x)du,
\end{equation*}
where $\mathcal{H}^{k-1}$ is the $(k\!-\!1)$-dimensional Hausdorff
measure and $g^{-1}_{d}(u) = \{x \in \mathbb{R}^{d}: g_{d}(x) = u\}$. Define the inner integral
\begin{equation*}
f_{{g}, \varepsilon}(u)
  := \int_{{g}_d^{-1}(u) \cap \mathcal{X}_{1-d}}
       \frac{p_{\varepsilon }(x)\,f_X(x)}
            {\|\nabla {g}_d(x)\|}\,
     d\mathcal{H}^{k-1}(x),
\end{equation*}
which is the effective density.
For $x \in \mathcal{X}_{1-d}$, note that $1\{\norm{x -
  X_{NN(x)}} \leq \varepsilon, \,\, X_{NN(x)} \text{ exists}\} = 0$ if and only
if either i) $X_{i} \in \mathcal{X}_{1-d}$ or ii) $X_{i} \in \mathcal{X}_{d}$
and $\norm{x- X_{i}} > \varepsilon$ for all $i$. Here, we abuse notation to
denote by $NN(x)$ the nearest neighbor of $x \in \mathcal{X}_{1-d}$ that lies
in $\mathcal{X}_{d}$. Therefore, we have
\begin{align*}
  p_{\varepsilon}(x) = & 1- P(1\{\norm{x -
  X_{NN(x)}} \leq \varepsilon, \,\, X_{NN(x)} \text{ exists}\} = 0) \\
  =& 1- (P(X \in \mathcal{X}_{1-d}) + P(X_{i} \in \mathcal{X}_{d}, \norm{x -
    X_{i}} > \varepsilon))^{n} \\
  = & 1 - (1- P(X \in \mathcal{X}_{d}, \norm{x -
    X_{i}} \leq \varepsilon))^{n}
\end{align*}

Now, define
$h_{\varepsilon}(x) = P(X \in \mathcal{X}_{d}, \norm{x - X_{i}} \leq
\varepsilon) = P( X \in \mathcal{X}_{d} \cap B_{\varepsilon}(x)) $. It follows
that, under the assumption that $\varepsilon^{d}n_{d}\to \infty $,
$f_{g,\varepsilon} (u) = \varepsilon f_{g}(u) + O(\varepsilon^{2})$, uniformly
over $u$, where
$f_{g}(u) = v_{1} \int_{g^{-1}_{d}(u) \cap \mathcal{F}} \kappa_{d}(z)
\frac{f_{X}(z)}{\lVert\nabla g_d(z)\rVert} d\mathcal H^{k-2}(z) $. Likewise, it
is easy to show $f_{g,\varepsilon}^{(d)} (u) = O(\varepsilon)$ for $d \in {1,2}$
where $f^{(d)}$ denotes the $d$th derivative of $f$.

Define $\pi_{1-d}:= P(X \in \mathcal{X}_{1-d})$ and observe that $E[W_{i}] = \pi_{1-d}E[p_{\varepsilon}(X_{i})|X_{i} \in
\mathcal{X}_{1-d}] = \pi_{1-d} \varepsilon \int_{\mathcal{F}} f_X(z)\,d\mathcal{H}^{\,k-1}(z)
+ O(\varepsilon^{2})$ so that $\frac{1}{n \varepsilon}\sum_{i=1}^{n} W_{i} =
\pi_{1-d}a_{\mathcal{F}} + o_{p}(1)$ where $a_{\mathcal{F}} := \int_{\mathcal{F}}
f_X(z)\,d\mathcal{H}^{\,k-1}(z)$. An immediate consequence is that $ \frac{1}{n_{1-d}\varepsilon}\sum_{i \in
  \mathcal{I}_{1-d}}W_{i} =a_{\mathcal{F}} +
o_{p}(1)$. Note that such convergences also hold almost surely as well, by a
strong law of large numbers for triangular arrays. 

Plugging this into the original integral,
we have
\begin{align*}
  \mu_{r} = &\int_{\mathbb{R}}K_{b}(u-y)((u-y)/b)^{r}f_{g,\varepsilon}(u)du 
              \\
  =&  \int_{\mathbb{R}}K(t)t^{r} f_{g,\varepsilon}(y + bt)dt
\end{align*}
where the second equality follows from  another change of variables $u =
y+bt$. It follows that

\begin{equation*}
      \mu_{r} =
  \begin{cases}
     f_{g,\varepsilon}(y) + (b^{2}/2)\mu_{2}(K)f''_{g,\varepsilon}(y) + O(b^{4}) &\text{if } r = 0 \\
  bf'_{g,\varepsilon}(y)\mu_{2}(K) + O(b^{3}) &\text{if } r = 1 \\
       f_{g,\varepsilon}(y)\mu_{2}(K) + O(b)  &\text{if } r = 2,
  \end{cases}
\end{equation*}
where we write $\mu_{2}(K) = \int u^{2}K(u) du$. This gives
\begin{equation*}
  \frac{1}{n^{2}\varepsilon^{2}}D(y) \to  f_{g}^{2}(y)\mu_{2}(K) 
\end{equation*}

By a similar calculation and a standard Lindeberg central limit theorem, we can show
\[
(n\varepsilon)^{-1/2}b^{1/2} \, T(y) \xrightarrow{d} \mathcal{N} \left( 0, \, \sigma^2(y) R(K) f_g^{-1}(y) \begin{pmatrix}
1 & 0 \\
0 &  \mu_2^{-1}
\end{pmatrix} \right),
\]
where $\sigma^{2}(y) := E[e_{i}^{2}| g_{d}(X_{i}) = y]$ and $R(K) := \int
K(u)^{2}du$. The bias term remains the same with a typical local linear
regression of $Y_{i}$ on $g_{d}(X_{i})$ under the
weighting scheme, which gives the result.
\end{proof}

\begin{proof}[Proof of Theorem \ref{thm:qstar-qhat}]
	Since the arguments close follow those given by MRS with minor modifications,
	we provide a detailed sketch of the proof and focus on the parts that differ.
	Following MRS, we derive a stochastic expansion of $\hat{q}_{1-d}(y)$ around
	$q^{*}_{1-d}(y)$. Define $\eta_{i} = Y_{i} - q_{1-d}(g_{d}(X_{i}))$. Under
	comonotonicity, we have $E[\eta_{i}| X_{i}] = 0$. Define
	$\tilde w_{i} (y) = \begin{pmatrix}
		1 \\
		(\tilde{g}_{d}(X_{i}) - y)/b
	\end{pmatrix}$ and $ w_{i}(y) = \begin{pmatrix}
		1 \\
		({g}_{d}(X_{i}) - y)/b
	\end{pmatrix}$.
	
	The outcome can be decomposed into
	\begin{align*}
		Y_{i} =& q_{1-d}(g_{d}(X_{i})) + \eta_{i} \\
		= & q_{1-d}(y) +  q_{1-d}(g_{d}(X_i)) - q_{1-d}(y) - q_{1-d}'(y)(g_{d}(X_i) - y) \\
		& -     q_{1-d}'(y) \bigl( \tilde{g}_{d}(X_i) - g_{d}(X_i) \bigr) +q_{1-d}'(y) \bigl( \tilde{g}_{d}(X_i) - y \bigr) + \eta_{i},
	\end{align*}
	which shows that
	\begin{equation*}
		\hat{q}_{1-d}(y) = q_{1-d}(y) + \hat{R}_{1} + \hat{R}_{2} + \hat{R}_{3},
	\end{equation*}
	where
	\begin{align*}
		\hat{R}_{1} & :=  e_{1}' \tilde{S}(y)^{-1}\sum_{i}
		W_i\,K_b(\tilde{g}_{d}(X_{i}) - y)
		\tilde{w}_i(y) \eta_{i} \\
		\hat{R}_{2} & :=   e_{1}'\tilde{S}(y)^{-1}\sum_{i}
		W_i\,K_b(\tilde{g}_{d}(X_{i}) - y)
		\tilde{w}_i(y) ( q_{1-d}(g_{d}(X_i)) - q_{1-d}(y) - q_{1-d}'(y)(g_{d}(X_i) -
		y)) \\
		\hat{R}_{3} & :=   e_{1}'\tilde{S}(y)^{-1}\sum_{i}
		W_i\,K_b(\tilde{g}_{d}(X_{i}) - y)
		\tilde{w}_i(y) ( -     q_{1-d}'(y) \bigl( \tilde{g}_{d}(X_i) - g_{d}(X_i)
		\bigr) )
	\end{align*}
	Likewise, we have
	\begin{align*}
		Y_{i} =& q_{1-d}(g_{d}(X_{i})) + \eta_{i} \\
		= & q_{1-d}(y) +  q_{1-d}(g_{d}(X_i)) - q_{1-d}(y) - q_{1-d}'(y)(g_{d}(X_i) - y) \\
		& +q_{1-d}'(y) \bigl({g}_{d}(X_i) - y \bigr) + \eta_{i}
	\end{align*}
	and thus the oracle estimator can be decomposed as
	\begin{align*}
		{q}^{*}_{1-d}(y) = & q_{1-d}(y) + {R}^{*}_{1} + {R}^{*}_{2} \\
		= & q_{1-d}(y) + {R}^{*}_{1} + {R}^{*}_{2} + {R}^{*}_{3} -  {R}^{*}_{3}
	\end{align*}
	where
	\begin{align*}
		R^{*}_{1} & :=  e_{1}' {S}(y)^{-1}\sum_{i}
		W_i\,K_b({g}_{d}(X_{i}) - y)
		w_{i}(y) \eta_{i} \\
		R^{*}_{2} & :=   e_{1}'{S}(y)^{-1}\sum_{i}
		W_i\,K_b({g}_{d}(X_{i}) - y)
		w_{i}(y) ( q_{1-d}(g_{d}(X_i)) - q_{1-d}(y) - q_{1-d}'(y)(g_{d}(X_i) -
		y)) \\
		R^{*}_{3} & :=   e_{1}'{S}(y)^{-1}\sum_{i}
		W_i\,K_b({g}_{d}(X_{i}) - y)
		w_{i}(y) ( -     q_{1-d}'(y) \bigl( \tilde{g}_{d}(X_i) - g_{d}(X_i)
		\bigr) )
	\end{align*}
	
	Define $\tilde{N}_{j}(y)$ so that $\hat{R}_{j} =
	e_{1}'\tilde{S}(y)^{-1}\tilde{N}_{j}(y)$ and define $N_{j}(y)$ likewise.
	
	We first derive a convergence rate for $\tilde{S}_{r}(y) - {S}_{r}(y,
	b)$. We have
	\begin{align*}
		& \tilde{S}_{r}(y) - {S}_{r}(y) \\
		= & \sum_{i} W_i  K_b(\tilde{g}_{d}(X_{i}) -
		y)\,b^{-r}(\tilde{g}_{d}(X_{i}) - y)^{r} - \sum_{i} W_{i}K_b({g}_{d}(X_{i}) -
		y)\,b^{-r}(\tilde{g}_{d}(X_{i}) - y)^{r} \\
		&+ \sum_{i} W_{i}K_b({g}_{d}(X_{i}) -
		y)\,b^{-r}(\tilde{g}_{d}(X_{i}) - y)^{r}  - \sum_{i} W_{i}K_b({g}_{d}(X_{i}) -
		y)\,b^{-r}({g}_{d}(X_{i}) - y)^{r} \\
		=:& (I) + (II)
	\end{align*}

	Since the weights satisfy $(\sum_{i=1}W_{i})/n\varepsilon \to_{p}
	\pi_{1-d}\alpha_{\mathcal{F}}$, we normalize the terms by
	$n\varepsilon$ instead of $n$. We have
	\begin{align*}
		(n\varepsilon)^{-1} \abs{(I)} \leq  \frac{1}{n\varepsilon} \sum_{i} W_i  \rvert K_b(\tilde{g}_{d}(X_{i}) -
		y) - K_b({g}_{d}(X_{i}) -
		y))\lvert b^{-r}(\tilde{g}_{d}(X_{i}) - y)^{r}
	\end{align*}
	Note that $ \rvert K_b(\tilde{g}_{d}(X_{i}) -
	y) - K_b({g}_{d}(X_{i}) -
	y))\lvert \neq 0$ only if either $\lvert \tilde{g}_{d}(X_{i}) -
	y \rvert \leq b$ or  $ \lvert {g}_{d}(X_{i}) -
	y \rvert \leq b$. In the former case, we have $\abs{b^{-r}(\tilde{g}_{d}(X_{i}) - y)^{r}}
	\leq 1$; in the latter we have  $\abs{b^{-r}(\tilde{g}_{d}(X_{i}) - y)^{r}}
	\leq (1 + r_{n}/b)^{r}$ where $r_{n} = O_{p}(a_{n})$. Since
	$a_{n}/b = o(1)$, the convergence rate of $(I)$ is determined
	by
	\begin{align*}
		& \frac{1}{n\varepsilon} \sum_{i} W_i  \rvert K_b(\tilde{g}_{d}(X_{i}) -
		y) - K_b({g}_{d}(X_{i}) -
		y))\lvert (1(|\tilde{g}_{d}(X_{i}) -
		y| \leq b) + 1(|{g}_{d}(X_{i}) -
		y| \leq b)) \\
		\leq & \frac{C}{b^{2} n\varepsilon} \sum_{i} W_i
		\rvert \tilde{g}_{d}(X_{i}) - {g}_{d}(X_{i})\lvert (1(|\tilde{g}_{d}(X_{i}) -
		y| \leq b) + 1(|{g}_{d}(X_{i}) -
		y| \leq b))
		\\
		\leq & \frac{Ca_{n}}{b^{2} n\varepsilon} \sum_{i} W_i
		\rvert (1(|\tilde{g}_{d}(X_{i}) -
		y| \leq b) + 1(|{g}_{d}(X_{i}) -
		y| \leq b))
		\\
		= & O_{p}(a_{n}/b)
	\end{align*}
	
	Now, for term $(II)$:
	\begin{align*}
		n^{-1} \abs{(II)} \leq  & \frac{1}{n} \sum_{i} W_{i}K_b({g}_{d}(X_{i}) -
		y)\,\lvert b^{-r}(\tilde{g}_{d}(X_{i}) - y)^{r}  -
		b^{-r}({g}_{d}(X_{i}) - y)^{r} \rvert \\
		\leq & \frac{C}{b n\varepsilon} \sum_{i} W_{i}K_b({g}_{d}(X_{i}) -
		y)\,\lvert \tilde{g}_{d}(X_{i}) - {g}_{d}(X_{i}) \rvert \\
		= & O_{p}(a_{n}/b)
	\end{align*}
	Noting that both bounds above are uniform in $y \in \mathcal{Y}$, we have
	\begin{equation}\label{eq:des_conv}
		\sup_{y}\frac{1}{n\varepsilon} \lvert \tilde{S}_{r}(y) - {S}_{r}(y)
		\rvert = O_{p}(a_{n}/b)
	\end{equation}
	for $r \in \{0, 1, 2\}$
	
	Now we move on bounding the difference between terms that appear in the
	decomposition above.
	
	\textbf{$\hat R_1 -R^*_1$: } This is a term that required a significant amount
	of effort in, e.g., \cite{mammen2012nonparametric}. However, in our case, the
	independence between $\tilde g_{d}(\cdot)$ and
	$(X_{i}, \eta_{i})_{i \in \mathcal{I}_{1-d}}$ simplifies the argument. To see this,
	note that
	\begin{align*}
		&(n\varepsilon)^{-1}(\tilde{N}_{1}(y) -
		N_{1}(y))) \\
		=& (n\varepsilon)^{-1} \sum_{i}
		W_i\,\left( K_b({g}_{d}(X_{i}) - y)
		w_{i}(y) - K_b(\tilde{g}_{d}(X_{i}) - y)
		\tilde{w}(y) \right) \eta_{i} \\
		= & (n\varepsilon)^{-1} \sum_{i}
		W_i\,\left( (K_b({g}_{d}(X_{i}) - y)
		-  K_b(\tilde{g}_{d}(X_{i}) - y))
		w_{i}(y)+ K_b(\tilde{g}_{d}(X_{i}) - y)
		(w_{i}(y) -
		\tilde{w}(y)) \right) \eta_{i} \\
		= & (n\varepsilon)^{-1} \sum_{i}
		W_i\,( K_b({g}_{d}(X_{i}) - y)
		-  K_b(\tilde{g}_{d}(X_{i}) - y))
		w_{i}(y)) \eta_{i} \\
		& +  (n\varepsilon)^{-1} \sum_{i}
		W_i\, K_b(\tilde{g}_{d}(X_{i}) - y)
		(w_{i}(y) -
		\tilde{w}(y)) \eta_{i}.
	\end{align*}
	Conditioning on $(X_{i})_{i \in \mathcal{I}_{d}}$, using the almost sure
	convergence results for $\tilde{g}_{d}$, and the tail condition on the
	distribution of $\eta_{i}$, these terms can be bounded using basic concentration
	inequalities yielding a bound of
	$O_{p}\Big( \frac{a_{n}}{b} \left(\frac{ \log (n\varepsilon)}{n \varepsilon b}
	\right)^{1/2} \Big)$. This gives
	\begin{align*}
		& \hat R_{1} - R^{*}_{1} \\
		= & e_{1}'\tilde{S}(y)^{-1}\tilde{N}_{1}(y) - e_{1}'{S}(y)^{-1}\tilde{N}_{1}(y) + e_{1}'{S}(y)^{-1}\tilde{N}_{1}(y) -
		e_{1}'{S}(y)^{-1}{N}_{1}(y) \\
		= & e_{1}'(\tilde{S}(y)^{-1}-{S}(y)^{-1})\tilde{N}_{1}(y) + e_{1}'{S}(y)^{-1}(\tilde{N}_{1}(y) -
		N_{1}(y)) \\
		=& O_{p}\left( \frac{a_{n}}{b} \left(\frac{
			\log (n\varepsilon)}{n \varepsilon b} \right)^{1/2}\right)
	\end{align*}
	
	\textbf{$\hat R_2 -R^*_2$: } We first investigate the difference between the
	numerators
	\begin{align*}
		& \tilde{N}_{2}(y) - N_{2}(y) \\
		= & \sum_{i}
		W_i\left(K_b({g}_{d}(X_{i}) - y)w_{i}(y)
		- K_b(\tilde{g}_{d}(X_{i}) - y)\tilde w_{i}(y)
		\right) ( q_{1-d}(g_{d}(X_i)) - q_{1-d}(y) - q_{1-d}'(y)(g_{d}(X_i) -
		y)).
	\end{align*}
	
	The weights
	$K_b({g}_{d}(X_{i}) - y)w_{i}(y) - K_b(\tilde{g}_{d}(X_{i}) - y)\tilde w_{i}(y)$ are
	nonzero only if either $\lvert \tilde{g}_{d}(X_{i}) - y \rvert \leq b$ or
	$ \lvert {g}_{d}(X_{i}) - y \rvert \leq b$. Under the assumption that $q_{1-d}$
	has uniformly bounded second derivatives, note that
	\begin{equation*}
		A_{n} := \lvert q_{1-d}(g_{d}(X_i)) - q_{1-d}(y) - q_{1-d}'(y)(g_{d}(X_i) -
		y) \rvert =  O_{p} (\abs{g_{d}(X_i) -
			y}^{2})
	\end{equation*}
	Hence,  $A_{n} = O_{p}(b^{2})$ if $\lvert \tilde{g}_{d}(X_{i}) -
	y \rvert \leq b$ and $ A_{n} =O_{p}(b^{2} + a_{n}^{2}
	)$ if $\lvert {g}_{d}(X_{i}) -
	y \rvert \leq b$. However, since $a_{n}/b = o(1)$, we have  $A_{n} =
	O_{p}(b^{2})$ in either case. Using \eqref{eq:des_conv}, it follows that
	\begin{align*}
		& \frac{1}{n\varepsilon}\lVert \tilde{N}_{2}(y) - N_{2}(y) \rVert \\
		\leq & A_{n} \frac{1}{n\varepsilon} \sum_{i}
		W_i\lVert K_b({g}_{d}(X_{i}) - y)w_{i}(y)
		- K_b(\tilde{g}_{d}(X_{i}) - y)\tilde w_{i}(y)
		\rVert \\
		= & O_{p}(a_{n}b)
	\end{align*}
	Here, we use the fact that
	\begin{equation*}
		\frac{1}{n\varepsilon} \sum_{i}
		W_i\lVert K_b({g}_{d}(X_{i}) - y)w_{i}(y)
		- K_b(\tilde{g}_{d}(X_{i}) - y)\tilde w_{i}(y)
		\rVert = O_{p}(a_{n}/b),
	\end{equation*}
	which follows from an argument identical to the one we used to prove
	\eqref{eq:des_conv}.
	
	Finally, we have
	\begin{align*}
		& \hat R_{2} - R^{*}_{2} \\
		= & e_{1}'\tilde{S}(y)^{-1}\tilde{N}_{2}(y) - e_{1}'{S}(y)^{-1}\tilde{N}_{2}(y) + e_{1}'{S}(y)^{-1}\tilde{N}_{2}(y) -
		e_{1}'{S}(y)^{-1}{N}_{2}(y) \\
		= & e_{1}'(\tilde{S}(y)^{-1}-{S}(y)^{-1})\tilde{N}_{2}(y) + e_{1}'{S}(y)^{-1}(\tilde{N}_{2}(y) -
		N_{2}(y)) \\
		=& O_{p}(a_{n}b)
	\end{align*}
	Here, we used the fact that
	\begin{align*}
		\sup_{y \in \mathcal{Y}} \lVert (n \varepsilon)^{-1}{S}_{r}(y) - E[ (W_i/\varepsilon)\,K_b({g}_{d}(X_{i}) -
		y)\,(({g}_{d}(X_{i}) - y)/b)^{r}] \rVert = o_{p}(1),
	\end{align*}
	which follows by a standard uniform law of large number argument (see, e.g.,
	\cite{andrews1992generic}),
	\begin{align*}
		& \frac{1}{n \varepsilon} \lVert \tilde{N}_{2}(y) \rVert  \\
		\leq & (n\varepsilon)^{-1} A_{n}\sum_{i}
		W_i\,K_b(\tilde{g}_{d}(X_{i}) - y) 
		\lVert \tilde w_{i}(y) \rVert \\
		= & O_{p}(b^{2}),
	\end{align*}
	uniformly in $y \in \mathcal{Y}$.
	
	\textbf{$\hat R_3 -R^*_3$: } Again, we first investigate the difference between the numerators:
	\begin{align*}
		& \frac{1}{n\varepsilon}\lVert \tilde{N}_{3}(y) - N_{3}(y) \rVert \\
		= &  \frac{1}{n\varepsilon} \sum_{i}
		W_i \lVert K_b({g}_{d}(X_{i}) - y)w_{i}(y)
		- K_b(\tilde{g}_{d}(X_{i}) - y)\tilde w_{i}(y)
		\rVert \cdot \lvert q_{1-d}'(y) ( \tilde{g}_{d}(X_i) - g_{d}(X_i) ) \rvert
		\\
		= & O_{p}(a_{n}^{2}/b),
	\end{align*}
	uniformly in $y \in \mathcal{Y}$. Hence, we have
	\begin{align*}
		& \hat R_{3} - R^{*}_{3} \\
		= & e_{1}'\tilde{S}(y)^{-1}\tilde{N}_{3}(y) - e_{1}'{S}(y)^{-1}\tilde{N}_{3}(y) + e_{1}'{S}(y)^{-1}\tilde{N}_{3}(y) -
		e_{1}'{S}(y)^{-1}{N}_{3}(y) \\
		= & e_{1}'(\tilde{S}(y)^{-1}-{S}(y)^{-1})\tilde{N}_{3}(y) + e_{1}'{S}(y)^{-1}(\tilde{N}_{3}(y) -
		N_{3}(y)) \\
		=& O_{p}(a_{n}^{2}/b),
	\end{align*}
	where we used the fact that $(n\varepsilon)^{-1}\lVert \tilde{N}_{3}(y) \rVert =
	O_{p}(a_{n})$ (again, up to a log term).
	
	For the last term, $R^*_3 - q_{1-d}'(y) \hat{\Delta}(y)$, consider the decomposition
	\begin{align*}
		\lvert R^{*}_{3,r} - q_{1-d}'(y)\hat{\Delta}(y) \rvert =  & \lvert
		e_{1}'((n\varepsilon)^{-1}S(y))^{-1} (n\varepsilon)^{-1}N_{3}(y) -
		q_{1-d}'(y)\hat{\Delta}(y)
		\rvert  \\
		\leq &  \lvert
		e_{1}'((n\varepsilon)^{-1}S(y))^{-1} ((n\varepsilon)^{-1}N_{3}(y) -
		E[(n\varepsilon)^{-1}N_{3}(y)])\rvert \\
		& + \lvert e_{1}'(((n\varepsilon)^{-1}S(y))^{-1} - B(K)^{-1}/f_{g}(y))
		E[(n\varepsilon)^{-1}N_{3}(y)]
		\rvert \\
		& + \lvert e_{1}' B(K)^{-1}/f_{g}(y)
		E[(n\varepsilon)^{-1}N_{3}(y)]- 
		q_{1-d}'(y)\tilde{\Delta}(y)
		\rvert \\
		=&O_{p}\Bigl(a_{n} \left(\frac{\log n}{n \varepsilon b}\right)^{1/2} +
		a_{n}^{2}/b + a_{n}b\Bigr)
	\end{align*}
	where $B(K) =
	\begin{pmatrix}
		1 & 0 \\
		0 & \mu_{2}(K)
	\end{pmatrix}
	$ so that $S(y) \to_{p} B(K)f_{g}(y)$. This concludes the proof.
\end{proof}

\section{Additional Results and Discussion}

\subsection{Comonotonicity in a Model of Skill Formation}

In order to provide further intuition for the comonotonicity condition in the
context of the summer schools example here we consider a simple model of skill
accumulation. For brevity we focus on the case in which the outcome $Y$ is the
math score a year after the initial tests.

Suppose that at the time he takes the initial tests, a student has an initial
latent `reading skill' $\xi_{r}$ and `math skill' $\xi_{m}$.  These skill-levels
may be correlated and let us assume they are jointly normally distributed in the
population. Test scores are noisy measurements of underlying skills and we
assume the measurement error is zero-mean normally distributed and independent
between the two tests. Formally, we assume that
\begin{align}
	&\begin{pmatrix}\xi_{r}\\ \xi_{m}
	\end{pmatrix}\sim N\bigg(\begin{pmatrix}\mu_{r}\\ \mu_{m}
	\end{pmatrix},\begin{pmatrix}\sigma_{r}^{2} & \sigma_{r,m}\\ \sigma_{r,m} &
		\sigma_{m}^{2}
	\end{pmatrix}\bigg)\label{norm1}\\ &\begin{pmatrix}X_{r}\\ X_{m}
	\end{pmatrix}\bigg|\xi_{r},\xi_{m}\sim N\bigg(\begin{pmatrix}\xi_{r}\\ \xi_{m}
	\end{pmatrix},\begin{pmatrix}\omega_{r}^{2} & 0\\ 0 & \omega_{m}^{2}
	\end{pmatrix}\bigg),\label{norm2}
\end{align}
where we assume the variance -covariance matrices above are strictly
positive definite.

Let $\zeta_{m}(0)$ be the student's potential math skill one year later if
untreated and $\zeta_{m}(1)$ his ability if treated. We suppose
$\zeta_{m}(0)=g_{0}(\xi_{m},\eta_{0})$ and
$\zeta_{m}(1)=g_{1}(\xi_{m},\eta_{1})$ where $g_{0}$ and $g_{1}$ are
deterministic functions strictly increasing in their first arguments and
$\eta_{0}$, $\eta_{1}$ are exogenous noise terms that capture random variation
in skill formation. Finally, we suppose that the potential math score a year
later is an unbiased signal of the potential skill, formally
$E[Y(d)|\xi_{m}(d)]=\xi_{m}(d)$, and that the noise in the score is exogenous so
that $Y(d)\indep X_{r},X_{m}|\xi_{m}(d)$ for $d=1,2$.

Under the modelling assumptions above comonotonicity holds, and in particular,
conditional average treated and untreated potential outcomes are both strictly
increasing functions of $X_{m}+\gamma X_{r}$ where $\gamma>0$ if and only if
$\sigma_{r,m}>0$ (to be precise,
$\gamma=\frac{\sigma_{r,m}\omega_{m}^{2}}{\sigma_{m}^{2}\omega_{r}^{2}+\sigma_{m}^{2}\sigma_{r}^{2}-\sigma_{r,m}^{2}}$). Moreover,
the model implies that conditional average potential outcomes equal conditional
average potential math skills. That is,
$E[Y(d)|X_{r},X_{m}]=E[\zeta_{m}(d)|X_{r},X_{m}]$ for $d=1,2$. We state this
formally in Proposition B1 below.

\theoremstyle{plain} \newtheorem*{PPA1}{Proposition B1} \begin{PPA1} Suppose
	(\ref{norm1}) holds with the variance-covariance matrix strictly positive
	definite and (\ref{norm2}) holds with $\omega_{r}^{2},\omega_{m}^{2}>0$. In
	addition, let $\zeta_{m}(d)=g_{d}(\xi_{m},\eta_{d})$,
	$E[Y_{m}(d)|\zeta_{m}(d)]=\zeta_{m}(d)$, and $Y_{m}(d)\indep
	X_{r},X_{m}|\zeta_{m}(d)$ for $d=1,2$ where $g_{0}$ and $g_{1}$ are
	deterministic functions strictly increasing in their first arguments and
	$\eta_{0}$, $\eta_{1}$ are each independent of $(X_{r},X_{m})$. Then
	$E[Y(d)|X_{r},X_{m}]=E[\zeta_{m}(d)|X_{r},X_{m}]=E[\zeta_{m}(d)|X_{m}+\gamma
	X_{r}]$ for $d=1,2$ and Assumption 2 holds.
\end{PPA1}

\subsection{Identification Under Rank Invariance}

In this subsection we provide additional identification results that hold under
rank invariance of potential outcomes. This rank invariance condition is given
formally below.

\newtheorem*{AA1}{Assumption B1 (Rank Invariance)} \begin{AA1} Consider
	$(Y_{1}(0),Y_{1}(1))$ and $(Y_{2}(0),Y_{2}(1))$ two independent copies of
	$(Y(0),Y(1))$, then \[Y_{1}(0)\geq Y_{2}(0)\iff Y_{1}(1)\geq Y_{2}(1).\]
	
\end{AA1}

Lemma B1 shows that Assumption B1 implies a comonotonicity condition for
conditional quantiles. This is of a similar form to the comonotonicity for
conditional means stated in Assumption 2 in the main body of the paper.

\theoremstyle{plain} \newtheorem*{LA1}{Lemma B1} \begin{LA1} Assumption B1
	implies that for any $q\in[0,1]$ and $x_{1},x_{2}\in\mathcal{X}$,
	\begin{equation} Q_{Y(0)|X}(t|x_{1})\leq Q_{Y(0)|X}(q|x_{2})\iff
		Q_{Y(1)|X}(q|x_{1})\leq Q_{Y(1)|X}(q|x_{2}).\label{condI}
	\end{equation}
	
\end{LA1}

The comonotonicity of conditional quantiles in the conclusion of Lemma B1 allows
for extrapolation of quantiles of potential outcomes. In order to identify the
distributions of potential outcomes at the frontier we require an additional
assumption that imposes conditional quantiles are continuous.

\newtheorem*{A4}{Assumption B2}
\begin{A4}[Continuous Conditional Quantiles] For $d=\{0,1\}$, the function
	$x\mapsto Q_{Y(d)|X}(q|x)$ is continuous on $\mathcal{X}$.
	
\end{A4}

Theorem B1 shows that under assumptions A1 and A2, conditional quantiles of both
potential outcomes can be identified at some points away from the frontier. The
result is analogous to Theorem 2 and follows by almost identical
reasoning. Under Assumption B1, the difference in quantiles
$Q_{Y(1)|X}(q|x)-Q_{Y(0)|X}(q|x)$ is in fact the $t$-th conditional quantile of
the individual treatment effect. Thus under Assumptions A1 and A2 it may be
possible to identify the entire distribution of individual treatment effects
away from the frontier.  \theoremstyle{plain} \newtheorem*{TA1}{Theorem
	B1} \begin{TA1} Suppose Assumptions 1.1, A2, and the conclusion of Lemma B1
	hold.  Then for each $d=0,1$ there is a unique continuous function $g_{d,q}$ on
	$\mathcal{X}_{d}\cup\mathcal{F}$ so that for all $x\in\mathcal{X}_{d}$,
	$g_{d,q}(x)=Q_{Y|X}(t|x)$. Suppose that for some $x\in\mathcal{X}_{d}$, there
	exists an $x^{*}\in\mathcal{F}$ so that $Q_{Y|X}(t|x)=g_{d,q}(x^{*})$. Then
	$Q_{Y(1-d)|X}(t|x)=g_{1-d,q}(x^{*})$ and $Q_{Y(d)|X}(t|x)=g_{d,q}(x^{*})$.
	
\end{TA1}

Note that Theorem B1 directly assumes the conclusion to Lemma B1 holds rather
than Assumption B1. The result in Lemma B1 is not an equivalence, that is
(\ref{condI}) may hold even when Assumption B1 does not.

\subsection{Local Comonotonicity}

Another condition that weakens Assumption 2 but may still be sufficient for
identification is local comonotonicity. In this case, we require that the vector
of first derivatives of the conditional average treated and untreated potential
outcomes are identical.

\theoremstyle{definition} \newtheorem*{AA3}{Assumption B2} \begin{AA3}[Local
	Comonotonicity] for any $x_{1}\in\text{supp}(X)$, there is some neighborhood of
	$x_{1}$ so that for all $x_{2}$ within this neighborhood
	\[ E[Y(1)|X=x_{1}]\geq E[Y(1)|X=x_{2}]\iff E[Y(0)|X=x_{1}]\geq
	E[Y(0)|X=x_{2}].
	\]
\end{AA3}

It is not difficult to see that Assumption 2 implies Assumption B2.  In effect,
the condition says that if a sufficiently small change in $x$ is associated with
an increase in $E[Y(0)|X=x]$ then it is also associated with an increase in
$E[Y(1)|X=x]$, and vice versa.  Thus comonotonicity holds locally to the point
$x$. Unlike in Assumption 2, the same need not hold for large changes in $x$.

Assumption B2 leads us to the following result which allows us to impute
conditional average potential outcomes for values of $x$ that are connected to
the frontier by a continuous contour curve.

\theoremstyle{plain} \newtheorem*{T4}{Theorem B2} \begin{T4} Suppose Assumptions
	1 and A2 hold and define $g_{0}$ and $g_{1}$ as in Theorem 1. Suppose that for
	some $x\in\mathcal{X}_{d}$, there is a continuous path
	$q:\,[0,1]\to\mathcal{X}_{d}\cup\mathcal{F}$ with $q(0)=x$ and $q(1)=x^{*}$
	where $x^{*}\in\mathcal{F}$, and $E[Y|X=q(t)]=E[Y|X=x]$ for all
	$t\in(0,1)$. Then $E[Y(1)|X=x]=g_{1}(x^{*})$ and $E[Y(0)|X=x]=g_{0}(x^{*})$.
	
\end{T4}

Note that continuity of $g_{d}$ implies that $g_{d}(x^{*})=E[Y|X=x]$ as in
Theorems 2 and $2C$. However, Theorem $A2$ also requires that $x^{*}$ is
connected to $x$ by a continuous curve through $\mathcal{X}_{d}$ along which the
conditional average outcome is constant. On the one hand, this further restricts
the set of points away from the frontier with which we can match a point on the
frontier and impute causal effect. On the other hand, the additional requirement
reduces the set of points on the frontier with which we can match some point
away from the frontier.

Finally we note that we can combine Assumptions 2C and A2 as follows to get an
even weaker condition.

\theoremstyle{definition} \newtheorem*{A5}{Assumption B3} \begin{A5}[Conditional
	Local Comonotonicity] For any $x_{1}\in\text{supp}(X)$, there is some
	neighborhood of $x_{1}$ so that for all $x_{2}$ within this neighborhood so that
	$x_{2}^{(2)}=x_{1}^{(2)}$
	\[ E[Y(1)|X=x_{1}]\geq E[Y(1)|X=x_{2}]\iff E[Y(0)|X=x_{1}]\geq
	E[Y(0)|X=x_{2}].
	\]
\end{A5}

Using this weaker condition we can weaken Theorem B2, just as Theorem 2C weakens
Theorem 2.

\theoremstyle{plain} \newtheorem*{T5}{Theorem B3} \begin{T5} Suppose Assumptions
	1 and A3 hold and define $g_{0}$ and $g_{1}$ as in Theorem 1. Suppose that for
	some $x\in\mathcal{X}_{d}$, there is a continuous path
	$q:\,[0,1]\to\mathcal{X}_{d}\cup\mathcal{F}$ with $q(0)=x$, $q(1)=x^{*}$ where
	$x^{*}\in\mathcal{F}$, $q(t)^{(2)}=x^{(2)}$ for all $t\in[0,1]$, and
	$E[Y|X=q(t)]=E[Y|X=x]$ for all $t\in(0,1)$.  Then $E[Y(1)|X=x]=g_{1}(x^{*})$ and
	$E[Y(0)|X=x]=g_{0}(x^{*})$.
	
\end{T5}

\subsection{Further Comparison with Existing Approaches}

\cite{Angrist2015} extrapolate in multivariate RDD using a related
strategy. They assume individuals are treated if and only if a scalar running
variable $R$ exceeds a cut-off, which can be normalized to zero. That is,
$X=(R,W')'$ and $D=1\iff R\geq 0$. They assume potential outcomes are mean
independent of $R$ after conditioning on $W$. Formally, $E[Y(d)|W,R]=E[Y(d)|W]$.
Thus contours of conditional average potential outcomes are perpendicular to the
frontier as illustrated in the figure below.
\begin{figure}[h]
	\caption{Angrist Rokkanen in two-dimensional RDD} \centering
	\includegraphics[scale=0.27]{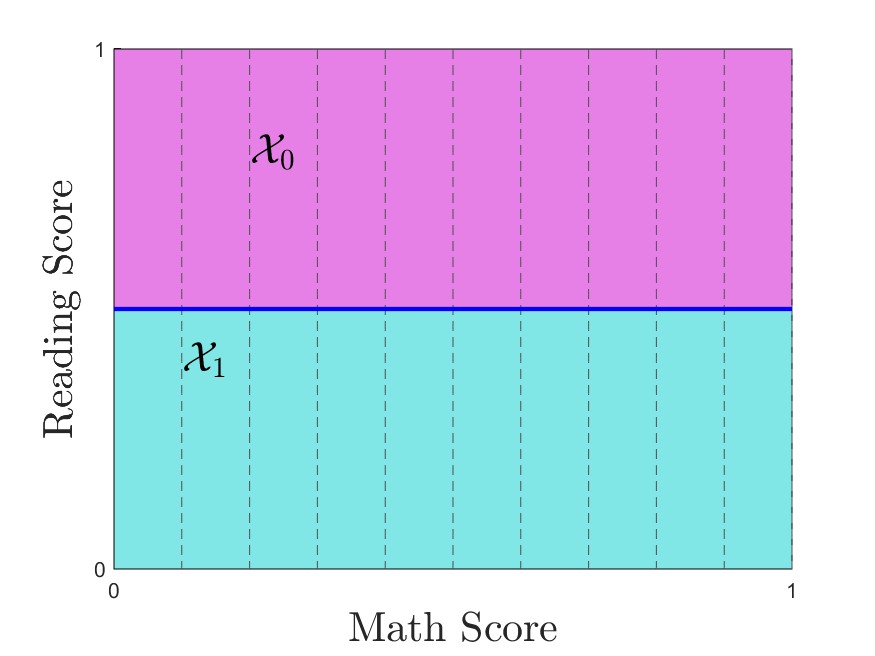}
\end{figure} Note that, for a social planner who wishes to treat those
individuals with the greatest conditional average effect, the treatment rule
above is highly suboptimal: treatment is decided based on a covariate which has
no power to predict causal effects.  Of course, the social planner may not have
access to the covariates $W$ or may be constrained to treat based on $R$ alone.

In a formal sense, our comonotonicity condition is neither weaker nor stronger
than the assumption of \cite{Angrist2015}. While the contours for $x\mapsto
E[Y(1)|X=x]$ and $x\mapsto E[Y(0)|X=x]$ are identical under their assumption,
they do not require a monotone relationship between $E[Y(1)|X=x]$ and
$E[Y(0)|X=x]$. However, if the joint distribution of observables is compatible
with comonotonicity (i.e., comonotonicity holds along the frontier) then our
extrapolation approach is valid under weaker conditions: if the approach in
\cite{Angrist2015} correctly identifies conditional average treatment effects
then so does ours, while the converse does not hold.

Our approach also bears some similarity to the identification strategies of
\cite{Imbens2009}, \cite{Chernozhukov2005}, \cite{Chernozhukov2007},
\cite{Torgovitsky2015}, and \cite{DHaultfoeuille2015}, who employ rank
invariance or rank similarity in the context of instrumental variables (IV)
estimation. As we discuss further in the next section, our comonotonicity
assumption is closely related to (but distinct from) rank invariance of
potential outcomes, which states that for two individuals $i$ and $j$ sampled
iid from the population, $Y_i(1)\geq Y_j(1)$ if and only if $Y_i(0)\geq
Y_j(0)$. The latter four papers mentioned above consider identification under
rank-invariance (or rank similarity) of potential outcomes. \cite{Imbens2009},
\cite{Torgovitsky2015}, and \cite{DHaultfoeuille2015} consider identification
under rank invariance/similarity of potential treatments, the latter two in
addition to rank invariance/similarity of potential outcomes.

Identification in IV using rank invariance of potential treatments can be
understood as a form of identification by extrapolation. Let $D$ be a continuous
treatment and $Z$ a continuous instrument. Suppose the IV is valid in that it
satisfies the exclusion restriction $Y(d,z)=Y(d)$ and unconfoundedness $Z\indep
Y(d),D(z)$. These assumptions identify the conditional distribution of potential
outcomes, conditional on certain values of potential treatments. To be precise,
for every $(d,z)$ in the support of $(D,Z)$, without any further assumptions it
follows that
\[ Y(d)|D(z)=d\sim Y|D=d,Z=z
\] and additionally, $D(z)\sim D|Z=z$. Unfortunately, this is generally
insufficient to identify average or conditional average causal effects. Suppose
we wish to identify the conditional mean of $Y(d_{1})-Y(d_{2})$ for some
$d_{1}\neq d_{2}$. The above can only identify objects like
$E[Y(d_{k})|D(z)=d_{k}]$ for $k=1,2$, which is the conditional mean of
$Y(d_{k})$ among a subpopulation that depends on $d_{k}$. To identify causal
effects, one must identify the means of $Y(d_{1})$ and $Y(d_{2})$ among the
\textbf{same} subpopulation. An insight of \cite{Imbens2009} is that under an
appropriate rank-invariance condition, one can find values $z_{1}$ and $z_{2}$
so that the subpopulation whose potential treatment is $d_{1}$ under instrument
level $z_{1}$ is precisely the same sub-population with potential treatment
$d_{2}$ under instrument level $z_{2}$.  It follows that
$E[Y(d_{2})|D(z_{2})=d_{2}]=E[Y(d_{2})|D(z_{1})=d_{1}]$ and so (assuming
sufficiently rich support of $Z$ and $D$) one may identify a conditional average
causal effect:
\[ E[Y|D=d_{1},Z=z_{1}]-E[Y|D=d_{2},Z=z_{2}]=E[Y(d_{1})-Y(d_{2})|D(z_{1})=d_{1}]
\] A sufficient rank invariance condition in the case of a scalar $D$ is as
follows. Taking $F_{D(z)}$ to be the cumulative distribution function of $D(z)$
(which is identified), one assumes that the quantile rank
$F_{D(z)}\big(D(z)\big)$ is the same for all $z$. Then an appropriate choice of
$z_{1}$ and $z_{2}$ is given by $F_{D(z_{1})}(d_{1})=F_{D(z_{2})}(d_{2})$.

\subsection{Evidence of Comonotonicity in an RCT}

Using data from a randomized controlled trial, it is possible to directly assess
whether comonotonicity holds. In an ideal RCT, both conditional average treated
and untreated potential outcomes are identified at all points in the support of
the covariates. Under comonotonicity there is a comonotonic relationship between
these quantities. Using estimates of the conditional mean potential outcomes we
can thus assess whether comonotonicity is credible in the empirical setting.

\cite{Alan} implement a randomized controlled trial in order to examine the
impact of a program designed to foster `grit' among elementary school students
in Istanbul, Turkey. Similar to our empirical application, the authors collect
scores on reading and math tests prior to the implementation of treatment for
any students. Students are then randomly assigned to treatment (receipt of the
program) or non-treatment. Two of the outcomes examined are reading and math
test scores two and a half years after the initial tests (which is after the
application of treatment), much as in our empirical application.

The figure below provides a scatter plot of the initial normalized test scores
for the students in the sample.

\begin{figure}[h] \centering \subfloat{
		\includegraphics[scale=0.25]{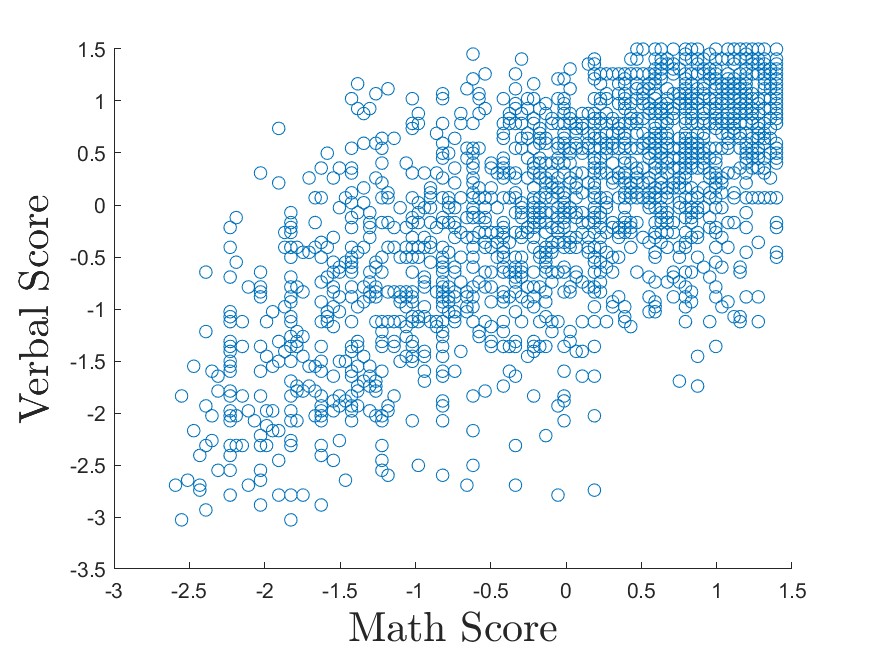} }
\end{figure}

We provide an informal assessment of the plausibility of comonotonicity in this
setting. In particular, we use the data in \cite{Alan} to estimate conditional
mean treated and untreated potential outcomes separately using the treated and
untreated samples respectively. Following \cite{Alan}, we adjust for attrition
by inverse probability weighting with a logit specification. We estimate the
conditional mean treated and untreated potential outcomes for each data point,
scatter plots are given below.

\begin{figure}[h]
	\caption{Math Outcome, Quadratic and Local-Linear} \centering \subfloat{
		\includegraphics[scale=0.18]{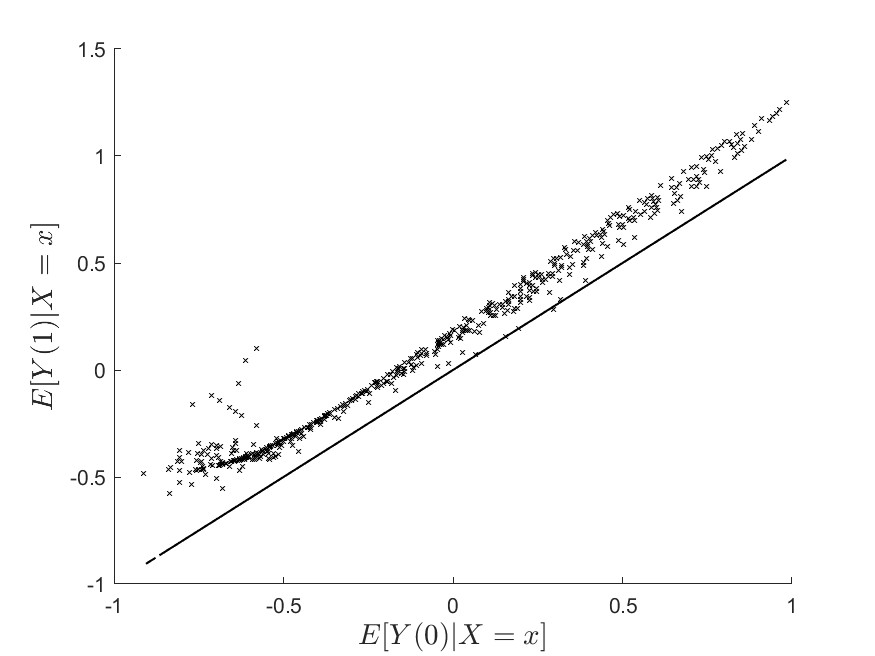}
		\label{fig:gritmathquad} } \subfloat{
		\includegraphics[scale=0.19]{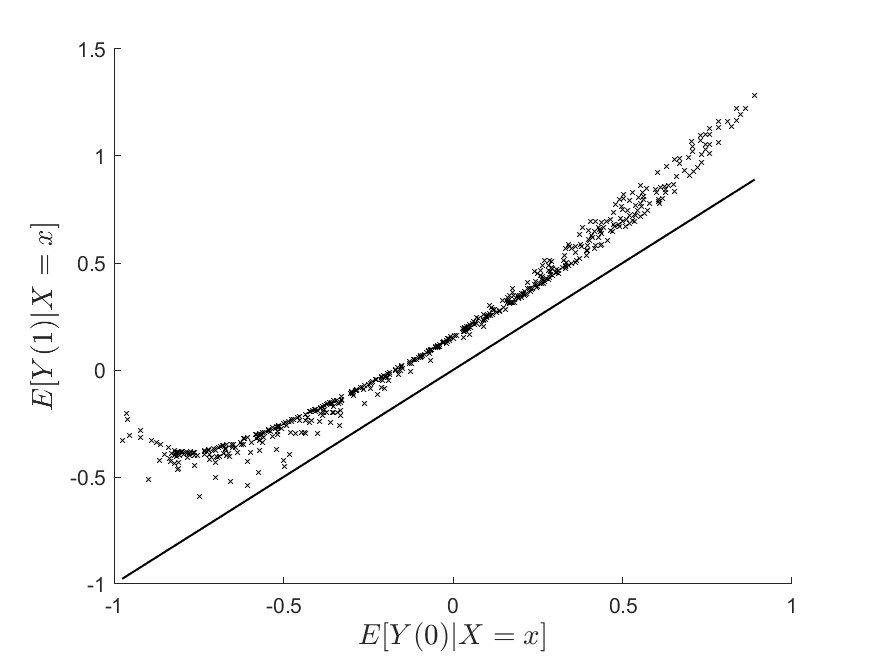}
		\label{fig:gritmathLL} }
	
\end{figure}

Figure \ref{fig:gritmathquad} provides results in which conditional mean
outcomes are obtain by quadratic regression and \ref{fig:gritmathLL} by local
linear. The results are similar between these two sub-figures. In both cases,
the conditional average potential outcome estimates are above the 45-degree line
suggesting positive treatment effects. The points in the figure appear
approximately obey a monotonically increasing functional relationship, in-line
with comonotonicity.

\begin{figure}[h]
	\caption{Reading Outcome, Quadratic and Local-Linear} \centering \subfloat{
		\includegraphics[scale=0.18]{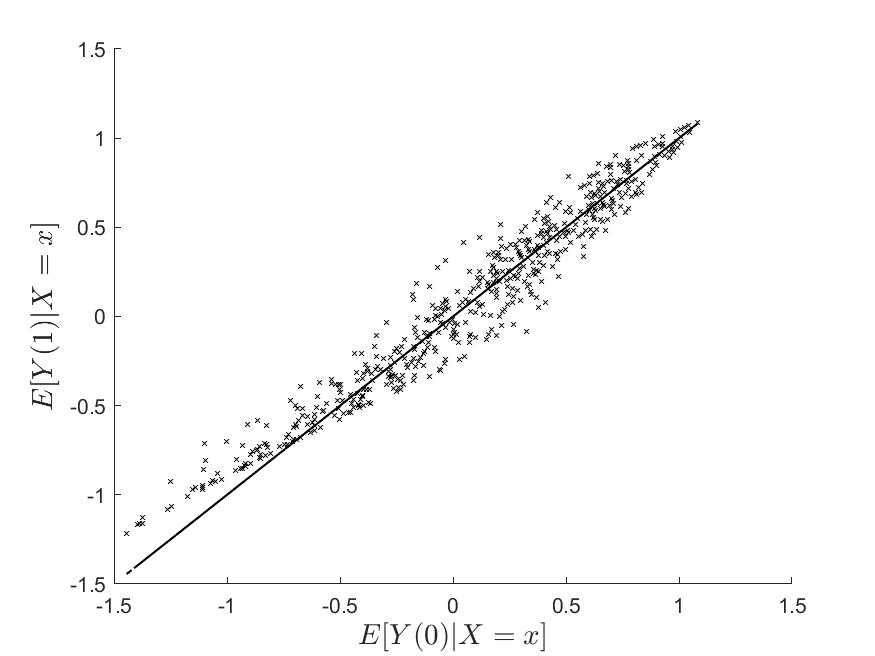}
		\label{fig:gritreadquad} } \subfloat{
		\includegraphics[scale=0.19]{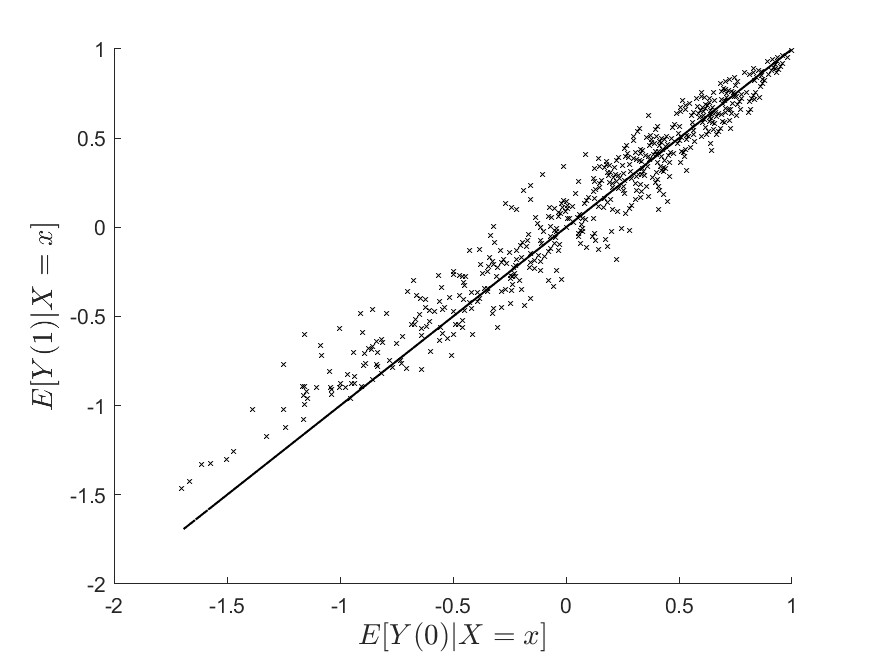}
		\label{fig:gritreadLL} }
\end{figure}

Figures \ref{fig:gritreadquad} and \ref{fig:gritreadLL} present analogous
results for the reading test score follow-up. The points in the scatter plot are
less concentrated around an increasing curve than in the case of the math
outcome. However, the general trend is still increasing and in the absence of a
formal test, we cannot rule out that the apparent deviations from a functional
relationship are due to estimation error.

\section{Proof of Additional Results}
\label{sec:proof-addl-results}

\begin{proof}[Proof of Proposition B1]

Applying Bayes' rule, for some constants $\gamma_{0}$, $\gamma_{r}$,
$\gamma_{m}$, and $v_{r}^{2}$ that do not depend on the test scores $X_{r}$ and
$X_{m}$, we have
\begin{equation} \xi_{m}|X_{r},X_{m}\sim
	N(\gamma_{0}+\gamma_{r}X_{r}+\gamma_{m}X_{m},v_{r}^{2}).\label{eq:bayesr}
\end{equation}

where $\gamma_{m}>0$, and $\gamma_{r}>0$ if and only if $\sigma_{r,m}>0$. In
particular,
\begin{align*}
	\gamma_{r}=&\frac{\omega_{m}^{2}\sigma_{r,m}}{\omega_{m}^{2}\omega_{r}^{2}+\sigma_{r}^{2}\omega_{m}^{2}+\sigma_{m}^{2}\omega_{r}^{2}+(\sigma_{m}^{2}\sigma_{r}^{2}-\sigma_{r,m}^{2})},
	\text{ and}\\
	\gamma_{m}=&\frac{\sigma_{m}^{2}\omega_{r}^{2}+(\sigma_{m}^{2}\sigma_{r}^{2}-\sigma_{r,m}^{2})}{\omega_{m}^{2}\sigma_{t,1}^{2}+\sigma_{r}^{2}\omega_{m}^{2}+\sigma_{m}^{2}\omega_{r}^{2}+(\sigma_{m}^{2}\sigma_{r}^{2}-\sigma_{r,m}^{2})}.
\end{align*}
Now note that

\begin{align*} E[Y(d)|X_{r},X_{m}] &
	=E\big[E[Y(d)|\zeta_{m}(d),X_{r},X_{m}]\big|X_{r},X_{m}]\\ &
	=E\big[E[Y(d)|\zeta_{m}(d)]\big|X_{r},X_{m}]\\ & =E[\zeta_{m}(d)|X_{r},X_{m}]\\
	& =E[g_{d}(\xi_{m},\eta_{d})|X_{r},X_{m}]\\ &
	=E[g_{d}(\xi_{m},\eta_{d})|X_{m}+(\gamma_{r}/\gamma_{m})X_{r}]
\end{align*}

where the first equality holds by iterated expectations, the second because
$Y(d)\indep X_{r},X_{m}|\zeta_{m}(d)$, the third by
$E[Y(d)|\zeta_{m}(d)]=\zeta_{m}(d)$, the fourth by substituting
$\zeta_{m}(d)=g_{d}(\xi_{m},\eta_{d})$.  The final equality follows because
$\eta_{d}$ is independent of $X_{r},X_{m}$ and (\ref{eq:bayesr}) shows that the
distribution of $\xi_{m}$ only depends on $X_{r}$ and $X_{m}$ through the
statistic $X_{m}+(\gamma_{r}/\gamma_{m})X_{r}$.

Finally, define
$f_{d}(z):=E[g_{d}(\xi_{m},\eta_{d})|X_{m}+(\gamma_{r}/\gamma_{m})X_{r}=z]$.
Recall that $g_{d}$ is strictly increasing in its first argument by supposition,
and by $(\ref{eq:bayesr})$ if $a>b$ then the conditional distribution of
$\xi_{m}$ given $X_{m}+(\gamma_{r}/\gamma_{m})X_{r}=a$ strictly first order
stochastically dominates the conditional distribution of $\xi_{m}$ given
$X_{m}+(\gamma_{r}/\gamma_{m})X_{r}=b$. It follows that $f_{d}$ is strictly
increasing for $d=1,2$, whence follows comonotonicity.
\end{proof}

\begin{proof}[Proof of Leamma B1]

First note that by rank invariance the quantile ranks of $Y(d)$ and $Y(1-d)$
are the same within any stratum of $X$, that is for any $x$:
\[ F_{Y(d)|X}\big(Y(d)\big|x\big)=F_{Y(1-d)|X}\big(Y(1-d)\big|x\big)
\]

It follows that $Y(d)\leq Q_{Y(d)|X}(t|x)\iff Y(1-d)\leq Q_{Y(1-d)|X}(t|x)$.
Now suppose that $Q_{Y(d)|X}(t|x_{1})<Q_{Y(d)|X}(t|x_{1})$. Then there are
realizations $(Y_{1}(0),Y_{1}(1))$ and $(Y_{2}(0),Y_{2}(1))$ of $(Y(0),Y(1))$ so
that:
\[ Q_{Y(d)|X}(t|x_{1})<Y_{1}(d)\leq Y_{2}(d)\leq Q_{Y(d)|X}(t|x_{1})
\]

By rank invariance $Q_{Y(d)|X}(t|x_{1})<Y_{1}(d)$ implies
$Y_{1}(1-d)<Q_{Y(1-d)|X}(t|x_{1})$, $Y_{2}(d)\leq Q_{Y(d)|X}(t|x_{1})$ implies
$Y_{2}(1-d)\leq Q_{Y(1-d)|X}(t|x_{1})$, and $Y_{1}(d)\leq Y_{2}(d)$ implies
$Y_{1}(1-d)\leq Y_{2}(1-d)$, and so we have:
\[ Q_{Y(1-d)|X}(t|x_{1})<Y_{1}(1-d)\leq Y_{2}(1-d)\leq Q_{Y(1-d)|X}(t|x_{1})
\]

And hence $Q_{Y(1-d)|X}(t|x_{1})<Q_{Y(1-d)|X}(t|x_{2})$. And so
$Q_{Y(d)|X}(t|x_{1})<Q_{Y(d)|X}(t|x_{1})$ implies
$Q_{Y(1-d)|X}(t|x_{1})<Q_{Y(1-d)|X}(t|x_{1})$. Since this holds for $d=0$ and
$d=1$ we see that:
\[ Q_{Y(d)|X}(t|x_{1})<Q_{Y(d)|X}(t|x_{1})\iff
Q_{Y(1-d)|X}(t|x_{1})<Q_{Y(1-d)|X}(t|x_{1})
\]

Which is equivalent to $Q_{Y(d)|X}(t|x_{1})\leq Q_{Y(d)|X}(t|x_{1})\iff
Q_{Y(1-d)|X}(t|x_{1})\leq Q_{Y(1-d)|X}(t|x_{1})$, as required.

\end{proof}

\begin{proof}[Proof of Theorem B1]

Note that for $x\in\mathcal{X}_{d}$, $Q_{Y(d)|X}(t|x)=Q_{Y|X}(t|x)$, the
latter of which is trivially identified. Then following the same steps as in
Theorem 1 we see that if $x\in\mathcal{X}_{d}$ and $x^{*}\in\mathcal{F}$, then
$g_{d,q}(x)=g_{d,q}(x^{*})$ implies $Q_{Y(d)|X}(t|x)=Q_{Y(d)|X}(t|x^{*})$, so
from the conclusion of Leamma B1 we get:
\begin{align*} Q_{Y(1-d)|X}(t|x) & =Q_{Y(1-d)|X}(t|x^{*})\\ &
	=g_{1-d,q}(x^{*})
\end{align*}
\end{proof}

\begin{proof}[Proof of Theorem B2]

From Theorem 1, if $q(t)\in\mathcal{X}_{d}$ and $x^{*}\in\mathcal{F}$, then
$E[Y|X=q(t)]=g_{d}(x^{*})$ implies $E[Y(d)|X=q(t)]=E[Y(d)|X=x^{*}]$.  Now we
will show that :
\[ E[Y(1-d)|X=x]=E[Y(1-d)|X=x^{*}]=g_{1-d}(x^{*})
\]
To show the above, we assume the contrapositive. In particular, let us suppose
that:
\begin{equation} E[Y(1-d)|X=x]<E[Y(1-d)|X=x^{*}]\label{eq:nottrue}
\end{equation}
We will show the above contradicts Assumption B2. A contradiction for the case
in which the strict inequality is in the reverse direction follows from
essentially identical arguments. If \eqref{eq:nottrue} holds then there is some
$b$ so that:
\[ E[Y(1-d)|X=x]<b<E[Y(1-d)|X=x^{*}]
\]
Define $t^{*}$ by:
\[ t^{*}=\sup\big[t\in[0,1]:\,E[Y(1-d)|X=q(t)]\leq b\big]
\]
Now, by Assumption B2, there is a neighborhood $\mathcal{N}$ of $q(t^{*})$ in
which comonotonicity holds. Because $E[Y(d)|X=q(t)]$ is constant for all $t$, it
thus follows by local comonotonicity that $E[Y(1-d)|X=q(t)]$ is constant over
all $t$ with $q(t)\in\mathcal{N}$. Now, because $q$ is continuous, the pre-image
of $\mathcal{N}$ under $q$ is a neighborhood of $t^{*}$. But then we have
established $E[Y(1-d)|X=q(t)]$ is constant over $t$ in a neighborhood of
$t^{*}$, but this contradicts the definition of the supremum.

\end{proof}

\begin{proof}[Proof of Theorem B3]

Following the same steps as in Theorem B2 we see
$E[Y(d)|X=q(t)]=E[Y(d)|X=x^{*}]$.  Now we will (as in Theorem B3) show that:
\[ E[Y(1-d)|X=x]=E[Y(1-d)|X=x^{*}]=g_{1-d}(x^{*})
\]

Again we assume the contrapositive. In particular, let us suppose that:
\begin{equation} E[Y(1-d)|X=x]<E[Y(1-d)|X=x^{*}]\label{eq:nottrue-1}
\end{equation}

We will show the above contradicts Assumption B3. A contradiction for the
reverse inequality follows from essentially identical arguments.

If \eqref{eq:nottrue-1} holds then there is some $b$ so that:
\[ E[Y(1-d)|X=x]<b<E[Y(1-d)|X=x^{*}]
\]
Define $t^{*}$ by:
\[ t^{*}=\sup\big[t\in[0,1]:\,E[Y(1-d)|X=q(t)]\leq b\big]
\] 
By Assumption B3, there is a neighborhood $\mathcal{N}$ of $q(t^{*})$ in which
conditional comonotonicity holds. Because $E[Y(d)|X=q(t)]$ is constant and
$q(t)^{(2)}=x^{(2)}$ for all $t$, it thus follows by conditional local
comonotonicity that $E[Y(1-d)|X=q(t)]$ is constant over all $t$ with
$q(t)\in\mathcal{N}$. The rest of the proof proceeds exactly as the proof of
Theorem B3.

\end{proof}

\end{document}